\documentclass[12pt]{article}
\usepackage{amsmath}
\usepackage{graphicx}
\usepackage{enumerate}
\usepackage{natbib}
\usepackage{url} 

\RequirePackage{amsthm,amsmath,mathrsfs,multirow,morefloats,booktabs,subfigure,pdfpages,color,algorithm,array,hhline,makecell,epstopdf,algorithm}
\RequirePackage{algorithm,algpseudocode}
\algrenewcommand\algorithmicrequire{\textbf{Input:}}
\algrenewcommand\algorithmicensure{\textbf{Output:}}
\usepackage{amsmath,amssymb}

\theoremstyle{plain}
\newtheorem{thm}{Theorem}[section]

\newtheorem{lem}[thm]{Lemma}

\newtheorem{rem}[thm]{Remark}
\newtheorem{cor}[thm]{Corollary}
\allowdisplaybreaks[4]
\newcommand{\blind}{1}

\begin{document}

	\def\spacingset#1{\renewcommand{\baselinestretch}%
		{#1}\small\normalsize} \spacingset{1}

	
	\if1\blind
	{
		\title{\bf Consistency of regularized spectral clustering in degree-corrected mixed membership model}
		\author{Huan Qing\\ 
			School of Mathematics, China University of Mining and Technology\\
			and \\
			Jingli Wang \thanks{E-mail: jlwang@nankai.edu.cn} \\
			School of Statistics and Data Science, Nankai University}
		\maketitle
	} \fi
	
	\if0\blind
	{
		\bigskip
		\bigskip
		\bigskip
		\begin{center}
			{\LARGE\bf Consistency of regularized spectral clustering in degree-corrected mixed membership model}
		\end{center}
		\medskip
	} \fi
	
	\bigskip
	\begin{abstract}



Community detection in network analysis is an attractive research area recently.
Here, under the degree-corrected mixed membership (DCMM) model, we propose an efficient approach called mixed regularized spectral clustering (Mixed-RSC for short) based on the regularized Laplacian matrix. Mixed-RSC is designed based on an ideal cone structure of the variant for the eigen-decomposition of the population regularized Laplacian matrix. We show that the algorithm is asymptotically consistent under mild conditions by providing error bounds for the inferred membership vector of each node.  As a byproduct of our bound, we provide the theoretical optimal choice for the regularization parameter $\tau$. 
To demonstrate the performance of our method, we apply it with previous benchmark methods on both simulated and real-world networks.  
To our knowledge, this is the first work to design spectral clustering algorithm for mixed membership community detection problem under DCMM model based on the application of regularized Laplacian matrix.
\end{abstract}

\noindent%
\newpage
\section{Introduction}
\label{sec:intro}
The study of networks has received substantial attentions in past few years, see \citep{MMSB, football,Newman2004,Newman2007,Tutorial,papadopoulos2012community,RSC}. Networks often have some underlying structures, `communities', that is, nodes are in groups.  Thus it is essential to detect communities to study how a network is organized. If in a network one node only belongs to one community, then the problem is known as (non-overlapping/non-mixed membership) community detection. While if some nodes share among communities, it is known as mixed membership community detection.
While, in a real network some nodes often belong to more than one communities. Thus, it is meaningful and crucial to study the problem of mixed membership community detection.

The stochastic blockmodel (SBM) \citep{SBM} is a well-known and popular model to generate non-mixed membership networks. SBM assumes that nodes in a same community are expected to have same degrees (popularity). While in real cases, the degrees may vary among nodes. Thus some degree corrected models are developed, such as the degree-corrected stochastic block
model (DCSBM) \citep{DCSBM} and overlapping continuous community assignment model (OCCAM) \citep{OCCAM}. For mixed membership networks, the mixed membership stochastic blockmodel (MMSB) \citep{MMSB} is well known and it is an extension of SBM. However, similar as SBM, MMSB doesn't consider the degree  heterogeneity. To overcome this issue,  \cite{mixedSCORE} proposed a Degree Corrected Mixed Membership (DCMM) model which considered both mixed membership and degree heterogeneity. In this paper, we will analyze the performance of regularized spectral clustering for mixed membership community detection in DCMM.

Consider an undirected, unweighted network $\mathcal{N}$ and assume that there are $K$ disjoint blocks $V^{(1)}, V^{(2)}, \ldots, V^{(K)}$ where $K$ is known in this paper. Let the symmetric matrix $A$ be its adjacency matrix such that $A(i,j)=1$ if there is an edge between node $i$ and $j$, $A(i,j)=0$ otherwise, for $i,j=1, \dots,n$. The DCMM model assumes that node $i$ belongs to cluster $V^{(k)}$ with probability $\pi_{i}(k)$, that is,
\begin{align*}
\mathbb{P}(i\in V^{(k)})=\pi_{i}(k), ~~~~\sum_{k=1}^{K}\pi_{i}(k)=1,~~~~ 1\leq k\leq K,~~ 1\leq i\leq n.
\end{align*}
 Denote $\pi_{i}=(\pi_{i}(1), \pi_{i}(2), \ldots, \pi_{i}(K))$ which is known as the Probability Mass Function (PMF) \citep{mixedSCORE}.
A node $i$ is `pure' if one element of $\pi_i$ is 1, and the remaining $K-1$ entries are 0; and it is  a `mixed' node otherwise. Furthermore, $\underset{1\leq k\leq K}{\mathrm{max}}\pi_{i}(k)$ can be used to measure the \textit{purity} of node $i$, for $1\leq i\leq n$.
The model generates the adjacency matrix as follows:
\begin{align*}
\mathbb{P}(A(i,j)=1)=\theta(i)\theta(j)\sum_{k=1}^{K}\sum_{l=1}^{K}\pi_{i}(k)\pi_{j}(l)P(k,l),\\
A(i,j)=A(j,i) \sim Bernoulli(\mathbb{P}(A(i,j)=1)), ind., 1\leq i,j\leq n,
\end{align*}
where $P$ is a $K\times K$ symmetric non-negative, non-singular and irreducible matrix (called mixing matrix in this paper) and $P(i,j)\in [0,1] \mathrm{~for~}1\leq i,j\leq K$, $\theta=(\theta(1), \ldots, \theta(n))'$ is a positive vector which models the degree heterogeneity. Note that since $\mathbb{P}(A(i,j))\in [0,1]$, we need $\theta(i)\in (0,1]$ for $1\leq i\leq n$, where $\theta(i)$ can not be zero otherwise node $i$ is an isolated node which does not belong to any community and should be removed from the network first. Define $\Omega(i,j)=\mathbb{P}(A(i,j)=1), 1\leq i<j\leq n$. Then the expected matrix of A, $\Omega$, can be presented as 
\begin{align}\label{Omega}
\mathbb{E}[A]=\Omega=\Theta \Pi P \Pi' \Theta,
\end{align}
where $\Theta$ is an $n\times n$ diagonal matrix whose $i$-th diagonal entry is $\theta(i)$ for $1\leq i\leq n$, and $\Pi$  is an $n\times K$ membership matrix such that the $i$-th row of $\Pi$ (denoted as $\Pi(i,:)$) is $\pi_{i}$ for all $i\in \{1,2,\ldots, n\}$.  
To emphasize that the DCMM model is closely related with the four model parameters $n, P,\Theta, \Pi$, we call DCMM as $DCMM(n, P,\Theta, \Pi)$.

If all nodes are pure, DCMM reduces to DCSBM \citep{DCSBM}. In the case where $\theta(i)=c_{0}$ (a positive constant) for all $i\in\{1,2,\ldots, n\}$, DCMM degenerates as MMSB \citep{MMSB}.
Given $(A,K)$, the primary goal for mixed membership community detection is to estimate the membership matrix $\Pi$. The identifiability of the DCMM model has been studied by many papers, such as \cite{jin2017a,mixedSCORE,MaoSVM}. Similar as \cite{mixedSCORE}, \cite{MaoSVM} and \cite{OCCAM}, the following two conditions are assumed throughout this paper to guarantee the identifiability of the DCMM model. 
\begin{itemize}
	\item [(I1)] $\mathrm{rank}(P)=K$, and all diagonal entries of $P$ are ones;
	\item [(I2)] Each community has at least one pure node.
\end{itemize}

Many papers have provided very nice literature reviews for community detection including non-mixed and mixed membership networks, such as \cite{Cai2016survey,fortunato2010community,Fortunato2016Community, goldenberg2010a}. Here, we give a brief review of methods for mixed membership community detection. \cite{ZHANG2007ident} identified overlapping communities by mapping the network to Euclidean space and then applying fuzzy $c$-means clustering method and finally obtaining the optimal communities by maximizing a modularity function.
To detect directed, weighted and overlapping communities, \cite{Lancichinetti2011} locally optimized the statistical significance of clusters  with the help of some tools of Extreme and Order Statistics. \cite{Gillis2014ieee} proposed a global optimization algorithm by computing non-negative matrix factorization approximation to the adjacency matrix. \cite{OCCAM} constructed the model OCCAM in which they defined a new vector which can measure the degree of a node belonging to some other communities, and proposed a spectral clustering method based on K-median method. \cite{GeoNMF} designed an optimization method called  GeoNMF for mixed membership community  based on the nonnegative matrix factorization under the MMSB model. \cite{mixedSCORE} proposed the DCMM model and modified the Spectral Clustering On Ratios-of-Eigenvectors (SCORE) \citep{SCORE} (which was designed for non-mixed community detection) to the mixed membership community detection problem by considering a vertex hunting procedure and a membership reconstruction step, and called it as Mixed-SCORE. \cite{mao2020estimating} developed a spectral clustering algorithm based on the leading eigenvectors' simplex structure of the population adjacency matrix  under MMSB and  provided upper bounds of error rates for the inferred community membership vector of each node. \cite{SRSC} designed two regularized spectral clustering approaches based on the ideal simplex structure and the ideal cone structure of the eigen-decomposition of the population regularized Laplacian matrix under MMSB. In this paper, we aim at studying the consistency and the impact of regularization on spectral clustering under DCMM.

This paper makes four contributions in relation to the use of regularized Laplacian matrix on mixed membership community detection. First, based on DCMM model, we propose a regularized spectral clustering method based on the regularized Laplacian matrix instead of directly on the adjacency matrix under the degree-corrected mixed membership model. Thus we call our proposed method as mixed regularized spectral clustering (mixed-RSC for short). Our method is designed based on the ideal cone structure appeared in a carefully designed variant of the eigen-decomposition of the population regularized Laplacian matrix, and we apply the SVM-cone algorithm \citep{MaoSVM} to hunt for the corners of the variants of the eigen-decomposition of the regularized Laplacian matrix for Mixed-RSC. Second, we show the asymptotical consistency of the proposed method and give a upper bound for the error rate of each node under mild condition, where our condition only needs a upper bound requirement of the network sparsity. Third, we study the impact of regularizer on the proposed method and give a theoretical optimal choice for the regularization parameter of the Laplacian matrix based on the error rate's upper bound. 
Our last contribution is, by carefully analyzing the upper bound of error rate, we find that our theoretical results reach the separation condition of a balanced network with $K$ clusters and the sharp threshold of the Erd\"os-R\'enyi random graph $G(n,p)$ \citep{erdos2011on}.

Notations in the paper: $\|\cdot\|_{F}$ for a matrix denotes the Frobenius norm,  $\|\cdot\|$ for a matrix denotes the spectral norm, $\|\cdot\|_{1}$ for a vector denotes the $l_{1}$ norm and $|C|$ means the absolute value of number $C$. For any matrix $X$, set the matrix $\mathrm{max}(X,0)$ such that its $(i,j)$ entry is $\mathrm{max}(X_{ij},0)$. For any matrix $X$, $\|X\|_{2\rightarrow\infty}$ denotes the maximum $l_{2}$-norm of all the rows of $X$, $\|X\|_{\infty}=\mathrm{max}_{i}\sum_{j}|X(i,j)|$, and $\kappa(X)$ denotes the condition number of $X$. For any matrix or vector $X$, $X'$ denotes the transpose of $X$.  For convenience, when we say ``leading eigenvalues'' or ``leading eigenvectors'', we are comparing the \emph{magnitudes} of the eigenvalues and their respective eigenvectors with unit-norm. 
Let $\lambda_{k}(X)$ be the $k$-th leading eigenvalue of the matrix $X$. $X(i,:)$ and $X(:,j)$  denote the $i$-th row and the $j$-th column of matrix $X$, respectively. $X(S_{r},:)$ and $X(:,S_{c})$ denote the rows and columns in the index sets $S_{r}$ and $S_{c}$ of matrix $X$, respectively. For any vector $x$, we use $x_{i}$ or $x(i)$ to denote the $i$-th entry of it occasionally. For any matrix $X\in\mathbb{R}^{m\times m}$, let $\mathrm{diag}(X)$ be the $m\times m$ diagonal matrix whose $i$-th diagonal entry is $X(i,i)$. $\mathbf{1}$ is a column vector with all entries being ones. 
$e_{i}$ is a column vector whose $i$-th entry is 1 while other entries are zero.

\section{Methodology}\label{sec3}
\subsection{The Ideal Cone (IC) and the Ideal algorithms}

First, we introduce a population regularized Laplacian matrix.  Let $\mathscr{D}_{\tau}=\mathscr{D}+\tau I$, where $\mathscr{D}$ is an $n\times n$ diagonal matrix whose $i$-th diagonal entry is $\mathscr{D}(i,i)=\sum_{j=1}^{n}\Omega(i,j)$, and $\tau$ is a nonnegative regularizer (call $\tau$ regularizer or regularization parameter). The population Laplacian matrix with regularization is defined as
\begin{align}\label{Laplacian} \mathscr{L}_{\tau}=\mathscr{D}^{-1/2}_{\tau}\Omega\mathscr{D}^{-1/2}_{\tau}.
\end{align}
Let $\tilde{\Theta}=\mathscr{D}^{-1/2}_{\tau}\Theta$, and $\tilde{\theta}$ be an $n\times 1$ vector whose $i$-th entry is $\tilde{\Theta}(i,i)$ for $1\leq i\leq n$.  Plugging Eq (\ref{Omega}) into Eq (\ref{Laplacian}), we have $$\mathscr{L}_{\tau}=\tilde{\Theta}\Pi P\Pi'\tilde{\Theta}.$$
By basic algebra, we know the rank of $\mathscr{L}_{\tau}$ is $K$. Thus $\mathscr{L}_{\tau}$ has $K$ nonzero eigenvalues. Let $\{\lambda_{i},\eta_{i}\}_{i=1}^{K}$ be such leading $K$ eigenvalues and their respective eigenvectors with unit-norm.

Next two lemmas guarantee the existence of the Ideal Cone structure in the variant of eigen-decomposition of the population regularized Laplacian matrix in the mixed membership network under $DCMM(n, P,\Theta, \Pi)$, where the Ideal Cone is introduced in Problem 1 \cite{MaoSVM}. For convenience, set  $F=P\Pi'\tilde{\Theta}^{2}\Pi$ as a $K\times K$ matrix with full rank.
\begin{lem}\label{ExistB}
Under $DCMM(n, P, \Theta, \Pi)$, let $\mathscr{L}_{\tau}=VEV'$ be the compact eigenvalue decomposition of $\mathscr{L}_{\tau}$ such that $V$ is an $n\times K$ matrix containing the leading $K$ eigenvectors $\{\eta_{1},\eta_{2}, \ldots, \eta_{K}\}$ and $E$ is a $K\times K$ diagonal matrix whose diagonal entries are the leading $K$ eigenvalues $\{\lambda_{1}, \lambda_{2}, \ldots, \lambda_{K}\}$, then there exists an unique $K\times K$ matrix $B$ such that
\begin{itemize}
  \item[(1)] $V=\tilde{\Theta}\Pi B$, and the $k$-th column of $B$ is the $k$-th right eigenvector of $F$, and $\lambda_{k}$ is the $k$-th eigenvalue of $F$ for $1\leq k\leq K$.
  \item[(2)] $B$ can also be written as $B=\tilde{\Theta}^{-1}(\mathcal{I},\mathcal{I})V(\mathcal{I},:)$, where $\mathcal{I}$ is the indices of rows corresponding to $K$ pure nodes, one from each community.
\end{itemize}
\end{lem}
\begin{rem}\label{IndexFix1}
Note that if there is another index set $\tilde{\mathcal{I}}$ such that $\tilde{\mathcal{I}}$ is the indices of rows corresponding to $K$ pure nodes, one from each community, where these $K$ pure nodes may differ from those in $\mathcal{I}$. Since $B$ is unique, we have  $\tilde{\Theta}^{-1}(\mathcal{I},\mathcal{I})V(\mathcal{I},:)\equiv \tilde{\Theta}^{-1}(\tilde{\mathcal{I}},\tilde{\mathcal{I}})V(\tilde{\mathcal{I}},:)$.
\end{rem}
Actually, under MMSB, $V(i,:)=V(j,:)$ if $\Pi(i,:)=\Pi(j,:)$. However, under DCMM, it does not hold. Since from Lemma \ref{ExistB}, we find that $V(i,:)=e'_{i}V=e'_{i}\tilde{\Theta}\Pi B=\tilde{\theta}(i)\Pi(i,:)B$, thus only if $\Pi(i,:)=\Pi(j,:)$ and $\theta(i)=\theta(j)$, we can draw the conclusion that $V(i,:)=V(j,:)$. However, if we consider the row-normalized version of $V$, denoted by $V_{*,1}$, i.e.,$V_{*,1}(i,:)=\frac{V(i,:)}{\|V(i,:)\|_{F}}$, then we can find that if $\Pi(i,:)=\Pi(j,:)$,  $V_{*,1}(i,:)=V_{*}(j,:)$ hold. We present this conclusion in the following lemma. For convenience, let $N_{V}$ be the $n\times n$ diagonal matrix such that $N_{V}(i,i)=\frac{1}{\|V(i,:)\|_{F}}$ for $1\leq i\leq n$, and then $V_{*,1}=N_{V}V$. Next lemma shows that each row of $V_{*,1}$ can be expressed by a scaled combination of $V_{*,1}(\mathcal{I},:)$  and exhibits the existence of the Ideal Cone.
\begin{lem}\label{IdealCone}
Under $DCMM(n, P, \Theta, \Pi)$, there exists a $Y_{1}\in\mathbb{R}^{n\times K}_{\geq 0}$ and no row of $Y_{1}$ is 0 such that
\begin{align*}
V_{*,1}=Y_{1}V_{*,1}(\mathcal{I},:),
\end{align*}
where $Y_{1}$ can be written as $Y_{1}=N_{M_{1}}\Pi \tilde{\Theta}^{-1}(\mathcal{I},\mathcal{I})N_{V}^{-1}(\mathcal{I},\mathcal{I})$, $N_{M_{1}}$ is an $n\times n$ diagonal matrix whose diagonal entries are positive. Meanwhile, for any two distinct nodes $i,j$, when $\Pi(i,:)=\Pi(j,:)$, we have $V_{*,1}(i,:)=V_{*,1}(j,:)$.
\end{lem}
\begin{rem}\label{IndexFix2}
Since $B=\tilde{\Theta}^{-1}(\mathcal{I},\mathcal{I})V(\mathcal{I},:)$, where $\tilde{\Theta}^{-1}(\mathcal{I},\mathcal{I})$ is a diagonal matrix, we have $V_{*,1}(\mathcal{I},:)$ is also obtained by normalizing each rows of $B$ to have unit-length. Since $B$ is unique by Lemma \ref{ExistB}, we say that $V_{*,1}(\mathcal{I},:)$ is unique, i.e., if there is another index set $\tilde{\mathcal{I}}\neq \mathcal{I}$, we still have $V_{*,1}(\mathcal{I},:)=V_{*,1}(\tilde{\mathcal{I}},:)$.
\end{rem}
Lemma \ref{IdealCone} shows that the form of $V_{*,1}=Y_{1}V_{*,1}(\mathcal{I},:)$ is actually the Ideal Cone mentioned in \cite{MaoSVM}. Since $\mathrm{rank}(V_{*,1})=K$, which gives that $\mathrm{rank}(V_{*,1}(\mathcal{I},:))=K$, suggesting that the inverse of $V_{*,1}(\mathcal{I},:)V'_{*,1}(\mathcal{I},:)$ exists. Therefore, Lemma \ref{IdealCone} also gives that
\begin{align}\label{Y1}
Y_{1}=V_{*,1}V'_{*,1}(\mathcal{I},:)(V_{*,1}(\mathcal{I},:)V'_{*,1}(\mathcal{I},:))^{-1}.
\end{align}
\begin{rem}\label{Y1alternativechoice}
Since  $V_{*,1}(\mathcal{I},:)\in\mathbb{R}^{K\times K}$ is full rank and $V^{-1}_{*,1}(\mathcal{I},:)\equiv V'_{*,1}(\mathcal{I},:)(V_{*,1}(\mathcal{I},:)V'_{*,1}(\mathcal{I},:))^{-1}$, we can also set  $Y_{1}=V_{*,1}V^{-1}_{*,1}(\mathcal{I},:)$.
\end{rem}
Since $V_{*,1}=N_{V}V, Y_{1}=N_{M_{1}}\Pi\tilde{\Theta}^{-1}(\mathcal{I},\mathcal{I})N_{V}^{-1}(\mathcal{I},\mathcal{I})$, we have $
N_{V}^{-1}N_{M_{1}}\Pi\tilde{\Theta}^{-1}(\mathcal{I},\mathcal{I})N_{V}^{-1}(\mathcal{I},\mathcal{I})=VV'_{*,1}(\mathcal{I},:)(V_{*,1}(\mathcal{I},:)V'_{*,1}(\mathcal{I},:))^{-1}$, which gives that
\begin{align}\label{ZYJ1}
N_{V}^{-1}N_{M_{1}}\Pi=VV'_{*,1}(\mathcal{I},:)(V_{*,1}(\mathcal{I},:)V'_{*,1}(\mathcal{I},:))^{-1}N_{V}(\mathcal{I},\mathcal{I})\tilde{\Theta}(\mathcal{I},\mathcal{I}).
\end{align}
Recall that $\mathscr{L}_{\tau}=\tilde{\Theta}\Pi P\Pi'\tilde{\Theta}=VEV'$, we have $\mathscr{L}_{\tau}(\mathcal{I},\mathcal{I})=\tilde{\Theta}(\mathcal{I},\mathcal{I})\Pi(\mathcal{I},:) P\Pi'(\mathcal{I},:)\tilde{\Theta}(\mathcal{I},\mathcal{I})=V(\mathcal{I},:)EV'(\mathcal{I},:)$, where we have used  $\Pi(\mathcal{I},:)=I$ as in the proof of Lemma \ref{ExistB}. Then we have $\tilde{\Theta}(\mathcal{I},\mathcal{I})P\tilde{\Theta}(\mathcal{I},\mathcal{I})=V(\mathcal{I},:)EV'(\mathcal{I},:)$, combine the above equality with the fact that  all diagonal entries of $P$ are ones, we have
\begin{align}\label{ThetaTildeIIExpression}
\tilde{\Theta}(\mathcal{I},\mathcal{I})=\sqrt{\mathrm{diag}(V(\mathcal{I},:)EV'(\mathcal{I},:))}.
\end{align}
Set $J_{1}=N_{V}(\mathcal{I},\mathcal{I})\tilde{\Theta}(\mathcal{I},\mathcal{I})$, then we have $J_{1}=\sqrt{\mathrm{diag}(N_{V}(\mathcal{I},\mathcal{I})V(\mathcal{I},:))EV'(\mathcal{I},:)N_{V}(\mathcal{I},\mathcal{I})}\equiv\sqrt{\mathrm{diag}(V_{*}(\mathcal{I},:)EV'_{*}\mathcal{I},:)}$.

For convenience, set $Z_{1}=N_{V}^{-1}N_{M_{1}}\Pi, Y_{\bullet,1}=VV'_{*,1}(\mathcal{I},:)(V_{*,1}(\mathcal{I},:)V'_{*,1}(\mathcal{I},:))^{-1}$. By Eq (\ref{ZYJ1}), we have
\begin{align}\label{Z1}
Z_{1}=Y_{\bullet,1}J_{1}\equiv VV'_{*,1}(\mathcal{I},:)(V_{*,1}(\mathcal{I},:)V'_{*,1}(\mathcal{I},:))^{-1}\sqrt{\mathrm{diag}(V_{*,1}(\mathcal{I},:)EV'_{*,1}\mathcal{I},:)}.
\end{align}
Meanwhile, since $N_{V}^{-1}N_{M_{1}}$ is an $n\times n$ positive diagonal matrix, we have $\Pi(i,:)=\frac{Z_{1}(i,:)}{\|Z_{1}(i,:)\|_{1}}$ for $1\leq i\leq n$. The above analysis shows when $\Omega, K$ are given, we can obtain $\mathscr{L}_{\tau}$ and its leading $K$ eigenvalues and eigenvectors, once we know the $K\times K$ corner matrix $V_{*,1}(\mathcal{I},:)$,  we can exactly recover $\Pi$ by setting $\Pi(i,:)=\frac{Z_{1}(i,:)}{\|Z_{1}(i,:)\|_{1}}$ for $1\leq i\leq n$.

Thus, the only difficulty is in finding the corner matrix $V_{*,1}(\mathcal{I},:)$. From Lemma \ref{IdealCone}, we know that $V_{*,1}=Y_{1}V_{*,1}(\mathcal{I},:)$ forms the Ideal Cone. In \cite{MaoSVM}, their SVM-cone algorithm (presented in the supplemental material) can exactly \footnote{The corner indices set $\mathcal{I}_{SVM-cone}$ returned by SVM-cone may not equal to $\mathcal{I}$, but it is the indices of rows of $V_{*,1}$ corresponding to $K$ pure nodes, one from each community, see section D in supplementary material for detail. By Remark \ref{IndexFix2}, we have $V_{*,1}(\mathcal{I},:)=V_{*,1}(\mathcal{I}_{SVM-cone},:)$, i.e., when the input is $V_{*,1}$ in  the SVM-cone algorithm, we can exactly obtain $V_{*,1}(\mathcal{I},:)$ by the index set returned from SVM-cone algorithm,  hence we state here that SVM-cone can exactly recover $\mathcal{I}$.} obtain the corner indices $\mathcal{I}$ from the Ideal Cone such that if the condition $(Y_{P}Y'_{P})^{-1}\textbf{1}>0$ (in \cite{MaoSVM}'s notations) holds. 
We find that $(V_{*,1}(\mathcal{I},:)V'_{*,1}(\mathcal{I},:))^{-1}\mathbf{1}>0$ holds (see Lemma \ref{LSVM}), which suggests that we can take the advantage of SVM-cone algorithm to deal with $V_{*,1}$ which has the ideal cone structure such that $V_{*,1}=Y_{1}V_{*,1}(\mathcal{I},:)$.
\begin{lem}\label{LSVM}
Under $DCMM(n, P, \Theta, \Pi)$, $(V_{*,1}(\mathcal{I},:)V'_{*,1}(\mathcal{I},:))^{-1}\mathbf{1}>0$ holds.
\end{lem}
The above analysis gives rise to the following three-stage algorithm which we call Ideal Mixed-RSC.
\begin{itemize}
	\item  Input $\Omega, K$. Output: $\Pi$.
  \item Obtain $\mathscr{L}_{\tau}$. Compute $V, E$ from  $\mathscr{L}_{\tau}$, and obtain $V_{*,1}$.
  \item Run SVM-cone algorithm on $V_{*,1}$ and $K$ to obtain the index set $\mathcal{I}_{SVM-cone}$. Then obtain the corner matrix $V_{*,1}(\mathcal{I},:)$ (By Lemma \ref{IndexFix2},  $V_{*,1}(\mathcal{I},:)\equiv V_{*,1}(\mathcal{I}_{SVM-cone},:)$.).
  \item Obtain $Y_{\bullet,1}$, $J_{1}$, and $Z_{1}=Y_{\bullet,1}J_{1}$.
          \item Recover $\Pi(i,:)$ by setting
              $\Pi(i,:)=\frac{Z_{1}(i,:)}{\|Z_{1}(i,:)\|_{1}}$ for $1\leq i\leq n$.
\end{itemize}
The above analysis shows that the Ideal Mixed-RSC exactly recovers the membership matrix $\Pi$.

To demonstrate that $V_{*,1}$ has the form of ideal cone structure, we drew Figure \ref{PlotVstar}. The result shows that all rows respective to mixed nodes of $V_{*,1}$ are located  at one side of the hyperplane formed by the $K$ (where $K$ is 3 in this figure) rows respective to pure nodes of $V_{*,1}$.  Meanwhile, we can exactly obtain the corner matrix $V_{*,1}(\mathcal{I},:)$ using the index set returned by SVM-cone algorithm. 
The data in Figure \ref{PlotVstar} is generated with the following settings: $n=400, K=3$,  and each cluster has 40 pure nodes. For a mixed node $j$ ($j=1,2,\ldots, 280$), we set $\Pi(j,1)=\mathrm{rand}(1)/2, \Pi(j,2)=\mathrm{rand}(1)/2, \Pi(j,3)=1-\Pi(j,1)-\Pi(j,2)$ where $\mathrm{rand}(1)$ is an arbitrary value in $(0,1)$. For $1\leq i\leq n$, $\theta(i)$ is a random value in $(0,1)$. The matrix $P$ is set as \[P=\begin{bmatrix}
    1&0.4&0.3\\
    0.4&1&0.1\\
    0.3&0.1&1\\
\end{bmatrix}.
\]
Then, when $n, \Pi, K, P, \Theta$ are fixed, after computing $\Omega$, we obtain $\mathscr{L}_{\tau}$, and then obtain $V_{*,1}$. Run SVM-cone algorithm on $V_{*,1}$ to obtain $\mathcal{I}_{SVM-cone}$. Since $V_{*,1}(\mathcal{I},:)\equiv V_{*,1}(\mathcal{I}_{SVM-cone},:)$ by Remark \ref{IndexFix2}, we obtain the $K$ corners returned by SVM-cone algorithm. After finishing the above settings, we can plot Figure \ref{PlotVstar}.
\begin{figure}
	\centering
{\includegraphics[width=0.66\textwidth]{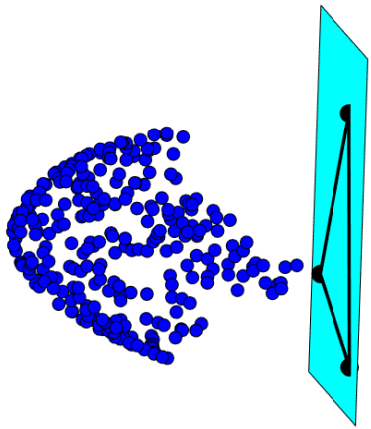}}
	\caption{Plot of $V_{*,1}$ and the hyperplane formed by $V_{*,1}(\mathcal{I},:)$. Blue points denote rows of $V_{*,1}$ where these rows are  respective to mixed nodes; Black points denote the $K$ rows of $V_{*,1}(\mathcal{I}_{SVM-cone},:)$ when the input is $V_{*,1}$ (note that $V_{*,1}(\mathcal{I}_{SVM-cone},:)$ is actually $V_{*,1}(\mathcal{I},:)$  by Remark \ref{IndexFix2}.); We also plot the hyperplane formed by the triangle of the 3 rows of $V_{*,1}(\mathcal{I},:)$. For visualization, we have projected and rotated these points from $\mathbb{R}^{3}$ to $\mathbb{R}^{2}$. Note that by Lemma \ref{IdealCone}, rows respective to pure nodes are same if these pure nodes are from the same cluster, hence rows refer to pure nodes coincide in $\mathbb{R}^{3}$ in this figure.}
\label{PlotVstar}
\end{figure}

\subsection{The algorithm: Mixed-RSC}
We now extend the ideal case to the real case. The following algorithm, which we call Mixed Regularized Spectral Clustering (Mixed-RSC for short) method, is a natural extension of the Ideal Mixed-RSC.
\begin{algorithm}
\caption{\textbf{Mixed-RSC}}
\label{alg:MRSC}
\begin{algorithmic}[1]
\Require The adjacency matrix $A\in \mathbb{R}^{n\times n}$, the number of communities $K$, and a regularizer $\tau\geq 0$.
\Ensure The estimated $n\times K$ membership matrix $\hat{\Pi}_{1}$.
\State Obtain the regularized Laplacian  matrix by
\begin{align*}
	L_{\tau}=D_{\tau}^{-1/2}AD_{\tau}^{-1/2},
\end{align*}
where $D_{\tau}=D+\tau I$, $D$ is an $n\times n$ diagonal matrix whose  $i$-th diagonal entry is $D(i,i)=\sum_{j=1}^{n}A(i,j)$ (unless specified, a good default $\tau$ is $\tau=0.1\mathrm{log}(n)$).
 \State Let $\hat{V}$ be the $n\times K$ matrix containing the leading $K$ eigenvectors $\{\hat{\eta}_{1}, \ldots,\hat{\eta}_{K}\}$ of $L_{\tau}$, let $\hat{E}$ be the $K\times K$ diagonal matrix whose diagonal entries are the leading $K$ eigenvalues $\{\hat{\lambda}_{1},\ldots, \hat{\lambda}_{K}\}$ of $L_{\tau}$ (i.e., $\hat{V}\hat{E}\hat{V}'$ is the leading $K$ eigen-decomposition of $L_{\tau}$). Normalize each row of $\hat{V}$  to have unit length, and denote by $\hat{V}_{*,1}$, i.e.,
$\hat{V}_{*,1}(i,j)=\hat{V}(i,j)/(\sum_{j=1}^{K}\hat{V}(i,j)^{2})^{1/2}, i = 1, \dots, n, j=1, \dots, K$.
\State Apply SVM-cone algorithm on the rows of $\hat{V}_{*,1}$ assuming there are $K$ clusters to obtain the near-corners matrix $\hat{V}_{*,1}(\mathcal{\hat{I}}_{1},:)$, where $\mathcal{\hat{I}}_{1}$ is the index set returned by SVM-cone.
\State Compute  $\hat{Y}_{\bullet,1}=\hat{V}\hat{V}'_{*,1}(\mathcal{\hat{I}}_{1},:)(\hat{V}_{*,1}(\mathcal{\hat{I}}_{1},:)\hat{V}'_{*,1}(\mathcal{\hat{I}}_{1},:))^{-1}$. Set $\hat{Y}_{\bullet,1}=\mathrm{max}(0, \hat{Y}_{\bullet,1})$. Compute $\hat{J}_{1}=\sqrt{\mathrm{diag}(\hat{V}_{*,1}(\mathcal{\hat{I}}_{1},:)\hat{E}\hat{V}'_{*,1}(\mathcal{\hat{I}}_{1},:))}$. Estimate $Z_{1}$ by $\hat{Z}_{1}=\hat{Y}_{\bullet,1}\hat{J}_{1}$.
\State Estimate $\Pi(i,:)$ by setting $\hat{\Pi}_{1}(i,:)=\hat{Z}_{1}(i,:)/\|\hat{Z}_{1}(i,:)\|_{1}, 1\leq i\leq n$.
\end{algorithmic}
\end{algorithm}
\begin{rem}
Mixed-RSC is a straightforward extension of Ideal Mixed-RSC except that we set $\hat{Y}_{\bullet,1}=\mathrm{max}(0,\hat{Y}_{\bullet,1})$ to transform negative entries of $\hat{Y}_{\bullet,1}$ into positive in the MR step due to the fact that $\hat{V}\hat{V}'_{*,1}(\mathcal{\hat{I}}_{1},:)(\hat{V}_{*,1}(\mathcal{\hat{I}}_{1},:)\hat{V}'_{*,1}(\mathcal{\hat{I}}_{1},:))^{-1}$ may contain a few negative entries in practice.
\end{rem}
\subsection{Equivalence algorithm}
As stated in Lemma G.1 in  \cite{MaoSVM}, one can apply SVM-cone on $\hat{V}_{*,2}$ ($\hat{V}_{*,2}$ is the row-normalization of $\hat{V}_{2}$, where $\hat{V}_{2}=\hat{V}\hat{V}'$.) instead of $\hat{V}_{*,1}$ with the same results, but it helps a lot on the theoretical analysis. Thus, we give an  ideal equivalence algorithm and an empirical equivalence algorithm  based on $V_{*,2}$, and then we show it returns the same outputs as Mixed-RSC.

By Lemma \ref{ExistB}, we know that $V=\tilde{\Theta}\Pi\tilde{\Theta}^{-1}(\mathcal{I},\mathcal{I})V(\mathcal{I},:)$.
Since $V_{2}(\mathcal{I},:)=V(\mathcal{I},:)V'$,
we have $V_{2}=VV'=\tilde{\Theta}\Pi\tilde{\Theta}^{-1}(\mathcal{I},\mathcal{I})V(\mathcal{I},:)V'=\tilde{\Theta}\Pi\tilde{\Theta}^{-1}(\mathcal{I},\mathcal{I})V_{2}(\mathcal{I},:)$.
The normalization can be written as $V_{2,*}=N_{V_{2}}V_{2}$, where $N_{V_{2}}$ be an $n \times n$ diagonal matrix with $N_{V_{2}}(i,i)=(\|V_{2}(i,:)\|_{F})^{-1}$. 
As $V_{2}\in\mathbb{R}^{n\times n}$, we have $V_{2,*}(\mathcal{I},:)\in\mathbb{R}^{K\times n}$. Similar as Lemma \ref{IdealCone}, we have the following lemma for $V_{*,2}$.

\begin{lem}\label{IdealCone2}
Under $DCMM(n, P, \Theta, \Pi)$, there exists a $Y_{2}\in\mathbb{R}^{n\times K}_{\geq 0}$ and no row of $Y_{2}$ is 0 such that
\begin{align*}
V_{*,2}=Y_{2}V_{*,2}(\mathcal{I},:).
\end{align*}
And $Y_{2}$ can be presented as $N_{M_{2}}\Pi \tilde{\Theta}^{-1}(\mathcal{I},\mathcal{I})N_{V_{2}}^{-1}(\mathcal{I},\mathcal{I})$ where $N_{M_{2}}$ is an $n\times n$ diagonal matrix whose diagonal entries are positive. For any two distinct nodes $i,j$, if $\Pi(i,:)=\Pi(j,:)$, we have $V_{*,2}(i,:)=V_{*,2}(j,:)$.
\end{lem}

Again, the form of $V_{*,2}=Y_{2}V_{*,2}(\mathcal{I},:)$ is the Ideal Cone. Since $\mathrm{rank}(V_{*,2})=K$ and $V_{*,2}\in\mathbb{R}^{K\times n}$, we have $\mathrm{rank}(V_{*,2}(\mathcal{I},:))=K$, suggesting that $V_{*,2}(\mathcal{I},:)$ is not invertible but the inverse of $V_{*,2}(\mathcal{I},:)V'_{*,2}(\mathcal{I},:)$ exists. From Lemma \ref{IdealCone2}, we have
\begin{align*}
Y_{2}=V_{*,2}V'_{*,2}(\mathcal{I},:)(V_{*,2}(\mathcal{I},:)V'_{*,2}(\mathcal{I},:))^{-1}.
\end{align*}
Since $V_{*,2}=N_{V_{2}}V_{2}$  and $ Y_{2}=N_{M_{2}}\Pi\tilde{\Theta}^{-1}(\mathcal{I},\mathcal{I})N_{V_{2}}^{-1}(\mathcal{I},\mathcal{I})$, we have $
N_{V_{2}}^{-1}N_{M_{2}}\Pi\tilde{\Theta}^{-1}(\mathcal{I},\mathcal{I})N_{V_{2}}^{-1}(\mathcal{I},\mathcal{I})=V_{2}V'_{*,2}(\mathcal{I},:)(V_{*,2}(\mathcal{I},:)V'_{*,2}(\mathcal{I},:))^{-1}$, then
\begin{align}\label{ZYJ2}
N_{V_{2}}^{-1}N_{M_{2}}\Pi=V_{2}V'_{*,2}(\mathcal{I},:)(V_{*,2}(\mathcal{I},:)V'_{*,2}(\mathcal{I},:))^{-1}N_{V_{2}}(\mathcal{I},\mathcal{I})\tilde{\Theta}(\mathcal{I},\mathcal{I}).
\end{align}
Set $J_{2}=N_{V_{2}}(\mathcal{I},\mathcal{I})\tilde{\Theta}(\mathcal{I},\mathcal{I})$. By Eq (\ref{ThetaTildeIIExpression}), we have $J_{2}=\sqrt{\mathrm{diag}(N_{V_{2}}(\mathcal{I},\mathcal{I})V(\mathcal{I},:))EV'(\mathcal{I},:)N_{V_{2}}(\mathcal{I},\mathcal{I})}$. Set $Z_{2}=N_{V_{2}}^{-1}N_{M_{2}}\Pi$ and $ Y_{\bullet,2}=V_{2}V'_{*,2}(\mathcal{I},:)(V_{*,2}(\mathcal{I},:)V'_{*,2}(\mathcal{I},:))^{-1}$. By Eq (\ref{ZYJ2}), we have
\begin{align}\label{Z2}
Z_{2}=Y_{\bullet,2}J_{2}\equiv V_{2}V'_{*,2}(\mathcal{I},:)(V_{*,2}(\mathcal{I},:)V'_{*,2}(\mathcal{I},:))^{-1}\sqrt{\mathrm{diag}(N_{V_{2}}(\mathcal{I},\mathcal{I})V(\mathcal{I},:))EV'(\mathcal{I},:)N_{V_{2}}(\mathcal{I},\mathcal{I})}.
\end{align}
Since $N_{V_{2}}^{-1}N_{M_{2}}$ is an $n\times n$ diagonal matrix, we have $$\Pi(i,:)=\frac{Z_{2}(i,:)}{\|Z_{2}(i,:)\|_{1}}, ~~~1\leq i\leq n.$$ 
By applying the SVM-cone algorithm on $V_{*,2}$, we can exactly obtain the corner matrix $V_{*,2}(\mathcal{I},:)$. The  Ideal Mixed-RSC(equivalence) can be presented as following.
\begin{itemize}
	\item Input: $\Omega, K$. Output: $\Pi$.
  \item  Obtain $\mathscr{L}_{\tau}$, $V_{2}$, and $V_{*,2}$.
  \item  Obtain the corner matrix $V_{*,2}(\mathcal{I},:)$ by applying SVM-cone algorithm on $V_{*,2}$ and $K$.
  \item Compute $Y_{\bullet,2}$, $J_{2}$ and $Z_{2}=Y_{\bullet,2}J_{2}$.
          \item Recover $\Pi(i,:)$ by setting
              $\Pi(i,:)=\frac{Z_{2}(i,:)}{\|Z_{2}(i,:)\|_{1}}$ for $1\leq i\leq n$.
\end{itemize}

Then the  empirical Mixed-RSC(equivalence) algorithm can be presented as:
\begin{algorithm}
\caption{\textbf{Mixed-RSC(equivalence)}}
\label{alg:MRSCequivalence}
\begin{algorithmic}[1]
\Require $A\in \mathbb{R}^{n\times n}$, $K$, and $\tau\geq 0$.
\Ensure The estimated $n\times K$ membership matrix $\hat{\Pi}_{2}$.
\State  Obtain $L_{\tau}, \hat{V},\hat{E}$ as in Algorithm \ref{alg:MRSC}. Let $\hat{V}_{2}=\hat{V}\hat{V}'$, then obtain $\hat{V}_{*,2}$ and the diagonal matrix $N_{\hat{V}_{2}}$ such that $\hat{V}_{*,2}=N_{\hat{V}_{2}}\hat{V}_{2}$.
\State Assuming there are $K$ clusters, apply SVM-cone algorithm on the rows of $\hat{V}_{*,2}$ to obtain the near-corners matrix $\hat{V}_{*,2}(\mathcal{\hat{I}}_{2},:)$, where $\mathcal{\hat{I}}_{2}$ is the index set returned by SVM-cone.
\State Compute  $\hat{Y}_{\bullet,2}=\hat{V}_{2}\hat{V}'_{*,2}(\mathcal{\hat{I}}_{2},:)(\hat{V}_{*,2}(\mathcal{\hat{I}}_{2},:)\hat{V}'_{*,2}(\mathcal{\hat{I}}_{2},:))^{-1}$. Set $\hat{Y}_{\bullet,2}=\mathrm{max}(0, \hat{Y}_{\bullet,2})$. Compute $\hat{J}_{2}=\sqrt{\mathrm{diag}(N_{\hat{V}_{2}}(\mathcal{\hat{I}}_{2},\mathcal{\hat{I}}_{2})\hat{V}(\mathcal{\hat{I}}_{2},:))\hat{E}\hat{V}'(\mathcal{\hat{I}}_{2},:)N_{\hat{V}_{2}}(\mathcal{\hat{I}}_{2},\mathcal{\hat{I}}_{2})}$.  Estimate $Z_{2}$ by  $\hat{Z}_{2}=\hat{Y}_{\bullet,2}\hat{J}_{2}$.
\State  Estimate $\Pi(i,:)$ by setting $\hat{\Pi}_{2}(i,:)=\hat{Z}_{2}(i,:)/\|\hat{Z}_{2}(i,:)\|_{1}, 1\leq i\leq n$.
\end{algorithmic}
\end{algorithm}

\subsection{The Equivalence}
We now emphasize the equivalence of  Algorithm \ref{alg:MRSC} and Algorithm \ref{alg:MRSCequivalence} from the ideal case to the empirical case by Lemma \ref{Equivalence}.
\begin{lem}\label{Equivalence}
For the ideal case, under $DCMM(n, P, \Theta, \Pi)$, we have $
N_{V_{2}}\equiv N_{V}, V_{*,2}(\mathcal{I},:)V'_{*,2}(\mathcal{I},:)\equiv V_{*,1}(\mathcal{I},:)V'_{*,1}(\mathcal{I},:), N_{M_{2}}\equiv N_{M_{1}}, Y_{2}\equiv Y_{1},Y_{\bullet,2}\equiv Y_{\bullet,1}, J_{2}\equiv J_{1}, Z_{2}\equiv Z_{1}$. For the empirical case, we have
$\mathcal{\hat{I}}_{2}\equiv \mathcal{\hat{I}}_{1}, \hat{V}_{*,2}(\hat{\mathcal{I}}_{2},:)\hat{V}'_{*,2}(\hat{\mathcal{I}}_{2},:)\equiv \hat{V}_{*,1}(\hat{\mathcal{I}}_{1},:)\hat{V}'_{*,1}(\hat{\mathcal{I}}_{1},:), \hat{Y}_{1}\equiv\hat{Y}_{2}, \hat{Y}_{\bullet,2}\equiv\hat{Y}_{\bullet,1}, \hat{J}_{2}\equiv\hat{J}_{1}, \hat{Z}_{2}\equiv\hat{Z}_{1}, \hat{\Pi}_{2}\equiv\hat{\Pi}_{1}$.
\end{lem}


From now on, for notation convenience, set $N\equiv N_{V}, N_{M}\equiv N_{M_{1}}, Y\equiv Y_{1},Y_{\bullet}\equiv Y_{\bullet,1}, J\equiv J_{1}, Z\equiv Z_{1}$, and $\hat{N}\equiv N_{\hat{V}_{2}}, \mathcal{\hat{I}}\equiv\mathcal{\hat{I}}_{1},
\hat{Y}\equiv\hat{Y}_{1},
\hat{Y}_{\bullet}\equiv\hat{Y}_{\bullet,1}, \hat{J}\equiv\hat{J}_{1}, \hat{Z}\equiv\hat{Z}_{1}, \hat{\Pi}\equiv\hat{\Pi}_{1}$.

\section{Theoretical analysis}\label{sec4}

In this section, we establish the performance guarantee for Mixed-RSC. First, we make the following assumption
\begin{itemize}
	\item [(A1)] For two positive numbers $\alpha$ and $\beta$, $\frac{\mathrm{log}(n^{\alpha}K^{-\beta})}{\theta_{\mathrm{max}}\|\theta\|_{1}}\rightarrow 0,\qquad \mathrm{as~}n\rightarrow \infty$.
\end{itemize}
Assumption (A1) means the network can not be too sparse when $n$ is large. Note that since $O(\mathrm{log}(n^{\alpha}K^{-\beta}))=O(\mathrm{log}(n))$, assumption (A1) also reads $\frac{\mathrm{log}(n)}{\theta_{\mathrm{max}}\|\theta\|_{1}}\rightarrow 0$ as $n\rightarrow\infty$. We consider the two positive numbers $\alpha$ and $\beta$ here mainly for the convenience of theoretical analysis.

For simplification, set $\delta_{\mathrm{min}}=\underset{1\leq i\leq n}{\mathrm{min}}\mathscr{D}(i,i), $ $\delta_{\mathrm{max}}=\underset{1\leq i\leq n}{\mathrm{max}}\mathscr{D}(i,i),$ $ \theta_{\mathrm{min}}=\mathrm{min}_{1\leq i\leq n}\theta(i), $ $\theta_{\mathrm{max}}=\mathrm{max}_{1\leq i\leq n}\theta(i),$ $ \tilde{\theta}_{\mathrm{min}}=\mathrm{min}_{1\leq i\leq n}\tilde{\theta}(i),$ $ \tilde{\theta}_{\mathrm{max}}=\mathrm{max}_{1\leq i\leq n}\tilde{\theta}(i)$.
\begin{lem}\label{boundL}
	Under $DCMM(n, P, \Theta, \Pi)$, if assumption (A1) holds, with probability at least $1-o(\frac{K^{4\beta}}{n^{4\alpha-1}})$, we have
	\begin{align*}
	\|L_{\tau}-\mathscr{L}_{\tau}\|=\begin{cases}
	O(\frac{\sqrt{\theta_{\mathrm{max}}\|\theta\|_{1}\mathrm{log}(n^{\alpha}K^{-\beta})}}{\tau+\delta_{\mathrm{min}}}), & \mbox{when } C\sqrt{\theta_{\mathrm{max}}\|\theta\|_{1}\mathrm{log}(n^{\alpha}K^{-\beta})}\leq \tau+\delta_{\mathrm{min}}\leq C\theta_{\mathrm{max}}\|\theta\|_{1}, \\
	O(\frac{\theta_{\mathrm{max}}\|\theta\|_{1}\mathrm{log}(n^{\alpha}K^{-\beta})}{(\tau+\delta_{\mathrm{min}})^{2}}), & \mbox{when~} \tau+\delta_{\mathrm{min}}<C\sqrt{\theta_{\mathrm{max}}\|\theta\|_{1}\mathrm{log}(n^{\alpha}K^{-\beta})}.
	\end{cases},
	\end{align*}
\end{lem}
In order to directly study the influence of parameters on the proposed method, the theoretical error bound given in this paper is directly related with the model parameters $(n,P,\Theta,\Pi)$ and $K$. For convenience, denote $err_{n}=\|L_{\tau}-\mathscr{L}_{\tau}\|$. Note that when $\alpha=1, \beta=0$, we have a general probability $1-o(n^{-3})$, and some authors use this probability for their theoretical analysis  \citep{SCORE, mixedSCORE}.

For Mixed-RSC, the main theoretical result (i.e., Theorem \ref{Main}) relies on the row-wise deviation bound for the eigenvector of the regularized Laplacian matrix.  In fact, \cite{mixedSCORE, MaoSVM, mao2020estimating} also hinge on a row-wise deviation bound but they are on the eigenvectors of the adjacency matrix. Next lemma provides the row-wise deviation bound for the eigenvectors of the regularized Laplacian matrix under $DCMM(n, P, \Theta, \Pi)$.
\begin{lem}\label{rowwiseerror}
	(Row-wise eigenvector error) Under $DCMM(n, P,\Theta,\Pi)$, suppose assumption (A1) holds. Assume $|\lambda_{K}|\geq C\frac{\theta_{\mathrm{max}}\sqrt{n\mathrm{log}(n)}}{\tau+\delta_{\mathrm{min}}}$, with probability at least $1-o(\frac{K^{4\beta}}{n^{4\alpha-1}})$, we have
	\begin{align*}	\|\hat{V}\hat{V}'-VV'\|_{2\rightarrow\infty}=O(\frac{(\tau+\delta_{\mathrm{max}})\theta_{\mathrm{max}}\sqrt{K\mathrm{log}(n)}}{(\tau+\delta_{\mathrm{min}})\theta^{2}_{\mathrm{min}}|\lambda_{K}(P)|\lambda_{K}(\Pi'\Pi)}).
	\end{align*}
\end{lem}
For convenience, we set $\varpi_{1}=\|\hat{V}\hat{V}'-VV'\|_{2\rightarrow\infty}$.
We emphasize that Lemma \ref{rowwiseerror} considers both positive and negative eigenvalues of $\mathscr{L}_{\tau}$ and $L_{\tau}$. Now, by Lemma \ref{boundL} and the conditions in Lemma \ref{rowwiseerror}, we can obtain the choice of $\tau$ as following: \begin{itemize}
	\item By Lemma \ref{boundL}, when $C\sqrt{\theta_{\mathrm{max}}\|\theta\|_{1}\mathrm{log}(n^{\alpha}K^{-\beta})}\leq \tau+\delta_{\mathrm{min}}\leq C\theta_{\mathrm{max}}\|\theta\|_{1}$, we have $err_{n}=C\frac{\sqrt{\theta_{\mathrm{max}}\|\theta\|_{1}\mathrm{log}(n^{\alpha}K^{-\beta})}}{\tau+\delta_{\mathrm{min}}}$. And  by the assumption in Lemma \ref{rowwiseerror} $|\lambda_{K}|\geq C\frac{\theta_{\mathrm{max}}\sqrt{n\mathrm{log}(n)}}{\tau+\delta_{\mathrm{min}}}$ and the facts $\theta^{2}_{\mathrm{max}}n\geq \theta_{\mathrm{max}}\|\theta\|_{1}$, we can find that $|\lambda_{K}|\geq C err_{n}$.
	As shown in Lemma 5 in the supplementary materials, we know $|\lambda_{K}|\leq \lambda_{1}\leq 1$, which could lead to $err_{n}\leq 1/C$. Then by the expression of $err_{n}$, we have $\tau+\delta_{\mathrm{min}}\geq C\sqrt{\theta_{\mathrm{max}}\|\theta\|_{1}\mathrm{log}(n^{\alpha}K^{-\beta})}$ which is consistent with the condition of $\tau+\delta_{\mathrm{min}}$.
	While, by Lemma \ref{boundL}, when $\tau+\delta_{\mathrm{min}}< C\sqrt{\theta_{\mathrm{max}}\|\theta\|_{1}\mathrm{log}(n^{\alpha}K^{-\beta})}$, we have $err_{n}=C\frac{\theta_{\mathrm{max}}\|\theta\|_{1}\mathrm{log}(n^{\alpha}K^{-\beta})}{(\tau+\delta_{\mathrm{min}})^{2}}\leq \lambda^{2}_{K}/C\leq 1/C$. Then we see that $\tau+\delta_{\mathrm{min}}\geq \sqrt{C\theta_{\mathrm{max}}\|\theta\|_{1}\mathrm{log}(n^{\alpha}K^{-\beta})}$, which is a contradiction. Hence, to make the condition of the lower bound of $|\lambda_{K}|$ hold, we need $C\sqrt{\theta_{\mathrm{max}}\|\theta\|_{1}\mathrm{log}(n^{\alpha}K^{-\beta})}\leq\tau+\delta_{\mathrm{min}}\leq C\theta_{\mathrm{max}}\|\theta\|_{1}$, then $err_{n}$ should be  written as $err_{n}=O(\frac{\sqrt{\theta_{\mathrm{max}}\|\theta\|_{1}\mathrm{log}(n^{\alpha}K^{-\beta})}}{\tau+\delta_{\mathrm{min}}})$.
\end{itemize}

Lemma \ref{boundC} provides the bound of the difference between $\hat{V}_{2,*}$ and $V_{2,*}$, which is the corner stone to characterize the behavior of the proposed algorithm.
\begin{lem}\label{boundC}
	Under $DCMM(n,P,\Theta,\Pi)$, when conditions in Lemma \ref{rowwiseerror} hold, there exists a permutation matrix $\mathcal{P}\in\mathbb{R}^{K\times K}$ such that with probability at least $1-o(\frac{K^{4\beta}}{n^{4\alpha-1}})$, we have
	\begin{align*}
	\|\hat{V}_{2,*}(\mathcal{\hat{I}},:)-\mathcal{P}V_{2,*}(\mathcal{I},:)\|_{F}=O(\frac{\tilde{\theta}^{7}_{\mathrm{max}}K^{2.5}\varpi_{1}\kappa^{3}(\Pi'\Pi)\sqrt{\lambda_{1}(\Pi'\Pi)}}{\eta \tilde{\theta}^{7}_{\mathrm{min}}}),
	\end{align*}
	where $\eta=\mathrm{min}_{1\leq k\leq K}((V_{*}(\mathcal{I},:)V'_{*}(\mathcal{I},:))^{-1}\mathbf{1})(k)$.
\end{lem}
Now if we know the bounds for  the row-wise deviation between $\hat{Y}_{\bullet}$ and $Y_{\bullet}$, and $\hat{Z}$ and $Z$, we can get the error rate bound for the estimation of the proposed method. The following two lemmas give such bounds.
\begin{lem}\label{boundY}
	Under $DCMM(n,P,\Theta,\Pi)$, when assumptions in Lemma \ref{rowwiseerror} hold, then with probability at least $1-o(\frac{K^{4\beta}}{n^{4\alpha-1}})$, we have
	\begin{align*}
	\mathrm{max}_{1\leq i\leq n}\|e'_{i}(\hat{Y}_{\bullet}-Y_{\bullet}\mathcal{P})\|_{F}=O(\frac{\tilde{\theta}^{10}_{\mathrm{max}}K^{3.5}\varpi_{1}\kappa^{4.5}(\Pi'\Pi)}{\eta \tilde{\theta}^{10}_{\mathrm{min}}}).
	\end{align*}
\end{lem}
\begin{lem}\label{BoundZ}
	Under $DCMM(n,P,\Theta,\Pi)$, when conditions in Lemma \ref{rowwiseerror} hold, with probability at least $1-o(\frac{K^{4\beta}}{n^{4\alpha-1}})$, we have
	\begin{align*}
	\mathrm{max}_{1\leq i\leq n}\|e'_{i}(\hat{Z}-Z\mathcal{P})\|_{F}=O(\frac{\tilde{\theta}^{10}_{\mathrm{max}}K^{3.5}\kappa^{4}(\Pi'\Pi)\varpi_{1}}{\tilde{\theta}^{11}_{\mathrm{min}}\eta\sqrt{\lambda_{K}(\Pi'\Pi)}}).
	\end{align*}
\end{lem}
For convenience, set  $\pi_{\mathrm{min}}=\mathrm{min}_{1\leq k\leq K}\mathbf{1}'\Pi e_{k}$ which measures the minimum summation of nodes belong to a certain community. Increasing $\pi_{\mathrm{min}}$ makes the network tend to be more balanced, vice verse. Next theory guarantees that the estimation of Mixed-RSC is consistent.
\begin{thm}\label{Main}
	Under $DCMM(n,P,\Theta,\Pi)$, when conditions in Lemma \ref{rowwiseerror} hold, with probability at least $1-o(\frac{K^{4\beta}}{n^{4\alpha-1}})$, we have
	\begin{align*}
	\mathrm{max}_{1\leq i\leq n}\|e'_{i}(\hat{\Pi}-\Pi\mathcal{P})\|_{F}=O(\frac{\theta^{17}_{\mathrm{max}}(\tau+\delta_{\mathrm{max}})^{9}K^{7}\kappa^{4.5}(\Pi'\Pi)\lambda_{1}(\Pi'\Pi)\theta_{\mathrm{max}}\|\theta\|_{1}\sqrt{\mathrm{log}(n)}}{\theta^{20}_{\mathrm{min}}(\tau+\delta_{\mathrm{min}})^{9}|\lambda_{K}(P)|\pi_{\mathrm{min}}\lambda^{1.5}_{K}(\Pi'\Pi)}).
	\end{align*}
\end{thm}
When we closely look into the assumption (A1) and the parametric probability $1-o(\frac{K^{4\beta}}{n^{4\alpha-1}})$, we can find that when decreasing $\alpha$ and/or increasing $\beta$, the network could be less sparse, however,  the parametric probability decreases. Therefore, we can conclude that there is a trade-off between the sparsity of a network and the probability for successfully detecting its mixed memberships. Especially, if a network is too sparse (which can be seen as $\alpha$ is too small or $\beta$ is too large in assumption (A1), then the probability of successfully detect such network decreases.

Since $\delta_{\mathrm{min}}\leq \delta_{\mathrm{max}}$, increasing $\tau$ decreases error bound in Theorem \ref{Main}, suggesting that a larger $\tau$ gives better estimations. Recall that $\tau+\delta_{\mathrm{min}}\leq C\theta_{\mathrm{max}}\|\theta\|_{1}$ (the analysis after Lemma 3.2), therefore the theoretical optimal choice of $\tau$ is:
\begin{align}\label{tauoptG}
\tau_{\mathrm{opt}}=O(\theta_{\mathrm{max}}\|\theta\|_{1}).
\end{align}
Meanwhile, the theoretical optimal choice of $\tau$ is larger than 0, suggesting the benefits of regularization (i.e., $\tau>0$) compared with no regularization (i.e., $\tau=0$)  in regularized spectral clustering. As is known, most real-world networks are sparse, and if we consider  the sparest network with $\theta_{\mathrm{max}}\|\theta\|_{1}=O(\mathrm{log}^{1+2\gamma}(n^{\alpha}K^{-\beta}))$ for $\gamma\rightarrow 0^{+}$ satisfying assumption (A1), the  optimal choice for the regularization parameter $\tau$  is
\begin{align}\label{tauoptimal}
\tau_{\mathrm{opt}}=O(\mathrm{log}(n^{\alpha}K^{-\beta}))\equiv O(\mathrm{log}(n)).
\end{align}

After plugging the optimal value for $\tau$ in Eq (\ref{tauoptG}) into Theorem \ref{Main}, we have the following corollary.
\begin{cor}\label{MainCor}
	Same as the conditions in Theorem \ref{Main}, with probability at least $1-o(\frac{K^{4\beta}}{n^{4\alpha-1}})$, we have
	\begin{align*}
	\mathrm{max}_{1\leq i\leq n}\|e'_{i}(\hat{\Pi}-\Pi\mathcal{P})\|_{F}=O(\frac{\theta^{17}_{\mathrm{max}}K^{7}\kappa^{4.5}(\Pi'\Pi)\lambda_{1}(\Pi'\Pi)\theta_{\mathrm{max}}\|\theta\|_{1}\sqrt{\mathrm{log}(n)}}{\theta^{20}_{\mathrm{min}}|\lambda_{K}(P)|\pi_{\mathrm{min}}\lambda^{1.5}_{K}(\Pi'\Pi)}).
	\end{align*}
	Especially, for the sparest case when $\theta_{\mathrm{max}}\|\theta\|_{1}=O(\mathrm{log}^{1+2\gamma}(n))$ for $\gamma\rightarrow 0^{+}$, we have
	\begin{align*}
	\mathrm{max}_{1\leq i\leq n}\|e'_{i}(\hat{\Pi}-\Pi\mathcal{P})\|_{F}=O(\frac{\theta^{17}_{\mathrm{max}}K^{7}\kappa^{4.5}(\Pi'\Pi)\lambda_{1}(\Pi'\Pi)\mathrm{log}^{1.5+2\gamma}(n)}{\theta^{20}_{\mathrm{min}}|\lambda_{K}(P)|\pi_{\mathrm{min}}\lambda^{1.5}_{K}(\Pi'\Pi)}).
	\end{align*}
\end{cor}
If we further make more assumptions on $K, \theta_{\mathrm{max}}, \theta_{\mathrm{min}}, \pi_{\mathrm{min}}$ and $\lambda_{1}(\Pi'\Pi)$ as Corollary 3.1 in \cite{mao2020estimating}, we can have a reduced error bound which is showed in the following corollary.
\begin{cor}\label{AddConditions}
	Under $DCMM(n, P,\Theta,\Pi)$, conditions in Theorem \ref{Main} hold, and suppose $K=O(1), \pi_{\mathrm{min}}=O(\frac{n}{K}),  \lambda_{1}(\Pi'\Pi)=O(\frac{n}{K})$, and $\theta_{\mathrm{max}}\leq C\theta_{\mathrm{min}}$, with probability at least $1-o(\frac{K^{4\beta}}{n^{4\alpha-1}})$, we have
	\begin{align*}
	\mathrm{max}_{1\leq i\leq n}\|e'_{i}(\hat{\Pi}-\Pi\mathcal{P})\|_{F}=O(\frac{\sqrt{\mathrm{log}(n)}}{\theta_{\mathrm{min}}|\lambda_{K}(P)|\sqrt{n}})\equiv O(\frac{1}{|\lambda_{K}(P)|}\sqrt{\frac{\mathrm{log}(n)}{\theta_{\mathrm{max}}\|\theta\|_{1}}}).
	\end{align*}
	Especially, for the sparest case when $\theta_{\mathrm{max}}\|\theta\|_{1}=O(\mathrm{log}^{1+2\gamma}(n))$ for $\gamma\rightarrow 0^{+}$, we have
	\begin{align*}
	\mathrm{max}_{1\leq i\leq n}\|e'_{i}(\hat{\Pi}-\Pi\mathcal{P})\|_{F}=O(\frac{1}{|\lambda_{K}(P)|\mathrm{log}^{\gamma}(n)}).
	\end{align*}
\end{cor}
Actually, in Corollary \ref{AddConditions} the assumption for the lower bound of $|\lambda_{K}(P)|$ (in Lemma \ref{rowwiseerror}) can be presented as $|\lambda_{K}(P)|\geq O(\sqrt{\frac{\mathrm{log}(n)}{\theta_{\mathrm{max}}\|\theta\|_{1}}})$. Please refer to the Remark 8 in the supplementary material for more details.
\begin{rem}\label{comparewithMixedSCORE}
	(Comparison to Theorem 2.2 in \cite{mixedSCORE}) It is easy to see that their conditions in Theorem 2.2 are our Condition (A1) and $\lambda_{K}(\Pi'\Pi)=O(\frac{n}{K})$ actually. When $K=O(1)$ and $\theta_{\mathrm{max}}\leq C\theta_{\mathrm{min}}$ (i.e., the settings in our Corollary \ref{AddConditions}), we see that the error bound in Theorem 2.2 in \cite{mixedSCORE} is $O(\frac{1}{|\lambda_{K}(P)|}\sqrt{\frac{\mathrm{log}(n)}{\theta_{\mathrm{max}}\|\theta\|_{1}})})$ (where their $\beta_{n}$ is just our $|\lambda_{K}(P)|$ actually). Therefore, the error rate for the proposed method is consistent with Mixed-SCORE \citep{mixedSCORE} for networks generated from $P$ whose $K$-th leading eigenvalue should also satisfy $|\lambda_{K}(P)|\geq O(\sqrt{\frac{\mathrm{log}(n)}{\theta_{\mathrm{max}}\|\theta\|_{1}}})$ under the settings of Corollary \ref{AddConditions}.
\end{rem}
\begin{rem}\label{comparSPACL}
	(Comparison to Theorem 3.2 in \cite{mao2020estimating}) When $\Theta=\sqrt{\rho_{0}}I$, $DCMM(n, P, \Theta, \Pi)$ degenerates to the MMSB model considered in \cite{mao2020estimating} (here, $\rho_{0}$ is the sparsity parameter). First, we'd note that as stated in  Theorem VI.1 \citep{mao2020estimating} and Table 1 \citep{lei2019unified}, \cite{mao2020estimating}'s Assumption 3.1 on $\rho_{0} n$ should be $\rho_{0} n\geq O(\mathrm{log}^{2\xi}(n))$ for $\xi>1$ instead of $\rho_{0} n\geq O(\mathrm{log}(n))$. For comparison, our requirement on $\rho_{0} n$ in our assumption (A1) is $\rho_{0} n\geq O(\mathrm{log}^{\xi}(n))$ ($\theta_{\mathrm{max}}\|\theta\|_{1}=\rho_{0}n$ under settings considered in this remark). Theorem 3.2 \citep{mao2020estimating} gives that their error bound is $O(\frac{1}{|\lambda_{K}(P)|\sqrt{\rho_{0} n}})$ under the settings of Corollary \ref{AddConditions} while our error bound for Mixed-RSC is $O(\frac{1}{|\lambda_{K}(P)|}\sqrt{\frac{\mathrm{log}(n)}{\rho_{0}n}})$. Though the error bound  $O(\frac{1}{|\lambda_{K}(P)|\sqrt{\rho_{0} n}})$  for \cite{mao2020estimating} is smaller than our bound for Mixed-RSC, \cite{mao2020estimating} needs stronger requirement on the network sparsity parameter $\rho_{0}$. Meanwhile, by Assumption 3.1 \citep{mao2020estimating}, we know that their $|\lambda_{K}(\Omega)|$ should be larger than $\sqrt{\rho_{0}n}\mathrm{log}^{\xi}(n)$. By their Lemma II.4, $|\lambda_{K}(\Omega)|$ has a lower bound $\rho_{0}|\lambda_{K}(P)|\lambda_{K}(\Pi'\Pi)$ (which is $O(\rho_{0}|\lambda_{K}(P)|n)$ under the settings of Corollary \ref{AddConditions}). Therefore, to make the requirement $|\lambda_{K}(\Omega)|\geq \sqrt{\rho_{0}n}\mathrm{log}^{\xi}(n)$ always hold, one only need  $\rho_{0}|\lambda_{K}(P)|n\geq\sqrt{\rho_{0}n}\mathrm{log}^{\xi}(n)$, which gives that Theorem 3.2 \citep{mao2020estimating} requires $|\lambda_{K}(P)|\geq O(\frac{\mathrm{log}^{\xi}(n)}{\sqrt{\rho_{0}n}})$. For comparison, ours error bound for Mixed-RSC requires $|\lambda_{K}(P)|\geq O(\sqrt{\frac{\mathrm{log}(n)}{\rho_{0}n}})$, and surely our requirement on the lower bound of $|\lambda_{K}(P)|$ is weaker than that of \cite{mao2020estimating}.
\end{rem}

Now, we consider a standard network by setting $\Theta=\sqrt{\rho_{0}}I$ and $P=\omega I_{K}+(1-\omega)I_{K}I'_{K}$ for $0<\omega<1$ (we have $\lambda_{K}(P)=\omega$) under the settings of Corollary \ref{AddConditions}. Note that when $\Theta=\sqrt{\rho_{0}}I$, we have $\Omega=\Theta\Pi P\Pi'\Theta=\Pi\rho_{0} P\Pi'=\Pi'B\Pi'$, where $B=\rho_{0} P$ and $B$ is the probability matrix now. Then the error rate in Corollary \ref{AddConditions} is $O(\frac{1}{|\lambda_{K}(P)|}\sqrt{\frac{\mathrm{log}(n)}{\rho_{0} n})}$. For convenience, set $B_{\mathrm{max}}=\mathrm{max}_{1\leq k,l\leq K}B(k,l)\equiv \rho_{0}, B_{\mathrm{min}}=\mathrm{min}_{1\leq k,l\leq K}B(k,l)\equiv \rho_{0}(1-\omega)$. Under such $P$ and settings in Corollary \ref{AddConditions}, since the error rate is $O(\frac{1}{\omega}\sqrt{\frac{\mathrm{log}(n)}{\rho_{0} n}})$, to obtain consistency estimation, $\omega$ should grow faster than $\sqrt{\frac{\mathrm{log}(n)}{\rho_{0} n}}$. Therefore, the probability gap $B_{\mathrm{max}}-B_{\mathrm{min}}=\rho_{0} \omega$ should grow faster than $\sqrt{\frac{\rho_{0}\mathrm{log}(n)}{n}}$, and the relative edge probability gap $\frac{B_{\mathrm{max}}-B_{\mathrm{min}}}{\sqrt{B_{\mathrm{max}}}}=\omega\sqrt{\rho_{0}}$ should grow faster than $\frac{\mathrm{log}(n)}{n}$. Especially, for the sparest case when $\rho_{0} n=O(\mathrm{log}^{1+2\gamma}(n))$ with $\gamma\rightarrow 0^{+}$, the probability gap should grow faster than $\frac{\mathrm{log}(n)}{n}$. Undoubtedly, this two separations are consistent with that of \cite{mixedSCORE},  since Theorem 2.2 \citep{mixedSCORE} shares the same error rate $O(\frac{1}{|\lambda_{K}(P)|}\sqrt{\frac{\mathrm{log}(n)}{\rho_{0}n}})$ for the standard network.

Next, we consider the  Erd\"os-R\'enyi (ER) random graph $G(n,p)$ \citep{erdos2011on}. To construct the ER random graph $G(n,p)$, set $\Omega=\sqrt{\rho_{0}}I, K=1$ and $\Pi$ is an $n\times 1$ vector with all entries being ones. Since $K=1$ and $P$ is assumed to have unit diagonal entries by the default condition (I1), we have $P=1$ in $G(n,p)$ and hence $\lambda_{K}(P)=1$. Then we have $\Omega=\Pi \rho P\Pi'=\Pi\rho\Pi'=\Pi p\Pi'$, i.e, $p=\rho_{0}$. Since the error rate is $O(\frac{1}{|\lambda_{K}(P)|}\sqrt{\frac{\mathrm{log}(n)}{\rho_{0} n}})=O(\sqrt{\frac{\mathrm{log}(n)}{pn}})$. For consistency estimation, we see that $p$ should grow faster than $\frac{\mathrm{log}(n)}{n}$, which is just the sharp threshold in Theorem 4.6 \citep{blum2020foundations} and the first bullet in Section 2.5 \citep{abbe2017community}. Meanwhile, since our assumption (A1) requires $\rho_{0}n\geq O(\mathrm{log}^{\xi}(n))$ for $\xi>1$, it gives that $p$ should grow faster than $\frac{\mathrm{log}(n)}{n}$ since $p=\rho_{0}$ under $G(n,p)$, which is consistent with the sharp threshold. 

\section{Simulations}\label{sec5}
In this section, a small-scale numerical study is applied to investigate the performance of our Mixed-RSC by comparing it with Mixed-SCORE \citep{mixedSCORE}, OCCAM \citep{OCCAM} GeoNMF \citep{GeoNMF} and SVM-cone-DCMMB \citep{MaoSVM}. We measure the performance of these methods by the mixed-Hamming error rate:
\begin{align*}
\mathrm{min}_{O\in\{ K\times K\mathrm{permutation~matrix}\}}\frac{1}{n}\|\hat{\Pi}O-\Pi\|_{1},
\end{align*}
where $\Pi$ and $\hat{\Pi}$ are the true and estimated mixed membership matrices respectively. Here, we also consider the permutation of labels since the measurement of error should not depend on how we label each of the K communities. For simplicity, we write the mixed-Hamming error rate as $\sum_{i=1}^{n}\|\hat{\pi}_{i}-\pi_{i}\|_{1}/n$.

For all cases, we set $n=500$ and $K=3$. Let each block own $n_{0}$ number of pure nodes for $0\leq n_{0}\leq 160$. Let the top $3n_{0}$ nodes $\{1,2, \ldots, 3n_{0}\}$ be pure and let nodes $\{3n_{0}+1, 3n_{0}+2,\ldots, 500\}$ be mixed. Assume all the mixed nodes have four different memberships $(x, x, 1-2x), (x, 1-2x, x), (1-2x, x, x)$ and $(1/3,1/3,1/3)$ with $x\in [0, 1/2)$, each with $(500-3n_{0})/4$ number of nodes. For $\rho\in(0, 1)$, the mixing matrix $P$ has unit diagonals  and off-diagonals $\rho$. For $z\geq 1$, we generate the degree parameters such that $1/\theta(i)\overset{iid}{\sim}U(1,z)$, where $U(1,z)$ denotes the uniform distribution on $[1, z]$. For all settings, we report the averaged mixed-Hamming error rate over 50 repetitions.

\texttt{Case 1: } Fix $(x,\rho, z)=(0.4, 0.4, 4)$ and let $n_{0}$ range in $\{40, 60, 80, 100, 120, 140, 160\}$.

\texttt{Case 2: } Fix $(x,n_{0},z)=(0.4,$ $100,4)$ and let $\rho$ range in $\{0, 0.05, 0.1, \ldots, 0.5\}$.

\texttt{Case 3: } Fix $(n_{0},\rho, z)=(100, 0.4, 4)$, and let $x$ range in $\{0, 0.05, \ldots, 0.5\}$.

\texttt{Case 4: } Fix $(n_{0},\rho, x)=(100, 0.4, 0.4)$, and let $z$ range in $\{1,1.5,2,\ldots,5\}$.

\begin{figure}
\centering
\subfigure[Experiment1]
{\includegraphics[width=0.45\textwidth]{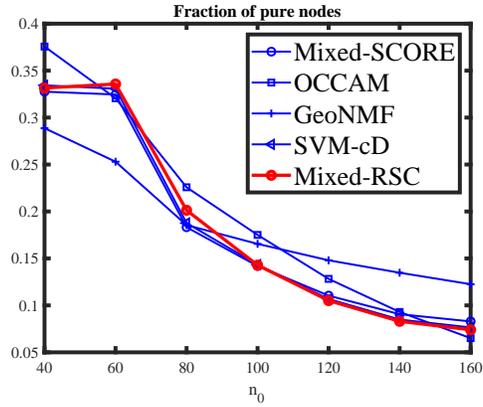}}
\subfigure[Experiment 2]
{\includegraphics[width=0.45\textwidth]{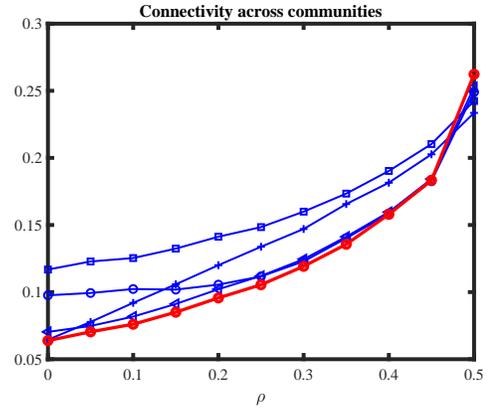}}
\subfigure[Experiment 3]
{\includegraphics[width=0.45\textwidth]{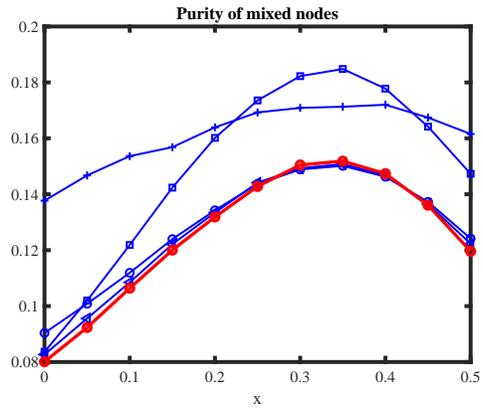}}
\subfigure[Experiment 4]
{\includegraphics[width=0.45\textwidth]{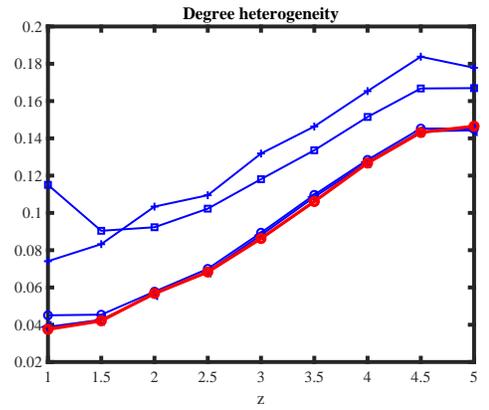}}
\caption{Estimation errors of Experiments 1-4 (y-axis: $\sum_{i=1}^{n}n^{-1}\|\hat{\pi}_{i}-\pi_{i}\|_{1}$).}
\label{EX}
\end{figure}

As is known, a larger $n_{0}$ indicates a case with higher fraction of pure nodes, thus we study how the number of pure nodes influence the performance of methods.  The numerical results are shown in the subfigure (a) in Figure \ref{EX} (note that SVM-cD is used to denote SVM-cone-DCMMSB.). From this figure we can find that all methods perform poor when $n_0\leq 60$, but when $n_0>60$ the error rates for all methods decrease rapidly. In detail, for a large $n_0$ Mixed-RSC performs similar as Mixed-SCORE and SVM-cone-DCMMSB while OCCAM and GeoNMF perform poorer than the other methods in this case.

 The results for case 2 are displayed in the subfigure (b) in Figure \ref{EX}. From this figure we can find that all methods perform poorer as $\rho$ increases. This phenomenon occurs due to the fact that a lager $\rho$ generate more edges across different communities (hence a dense network), and  more edges across different communities lead to a case that these communities tend to be in a giant community and hence a case that is more challenging to detect for any algorithms. Meanwhile, the results suggest that Mixed-RSC has similar performances as Mixed-SCORE and SVM-cone-DCMMSB, and they perform better than OCCAM and GeoNMF.

In case 3, $x$ is changed which has effect on the purity of nodes. By the setting, we can find that when $x$ increases to 1/3, these mixed nodes become less pure and they become more pure as $x$ increases further. The subfigure (c) of Figure \ref{EX} records the numerical results of this case. From the results we can see that when $x$ increase up to 1/3, the error rates for all methods increase, while they decrease when $x$ increases from 1/3 to 1/2. Thus we can make a conclusion that purity of nodes is higher, all methods perform better. Overall Mixed-RSC performs slightly better than Mixed-SCORE and SVM-cone-DCMMSB, and the three methods significantly outperform OCCAM and GeoNMF.

In case 4 we study the effect of degree heterogeneity. A larger $z$ gives smaller $\theta(i)$ for any node $i$, hence a more heterogeneous case and fewer edges generated.  The last panel of Figure \ref{EX} presents the results. We see that the error rates for almost all methods increase when the value of $z$ increases. Thus  all methods perform poor when a network has high degree heterogeneity. When we make a comparison of these five methods, we can drew a similar conclusion as in other three cases, i.e., Mixed-RSC,Mixed-SCORE and SVM-cone-DCMMSB have competitive performances and all the three methods enjoy better performances than OCCAM and GeoNMF.

\section{Real data analysis}\label{SNAPego}
The SNAP ego-networks dataset contains substantial ego-networks from three platforms Facebook, GooglePlus, and Twitter.  The dataset can be find in  http://snap.stanford.edu/data/.  Some others are also worked on this dataset, such as \cite{Mcauley2012snap, OCCAM}. We obtain the SNAP ego-networks parsed by Yuan Zhang (the first author of the OCCAM method \citep{OCCAM}).
The parsed SNAP ego-networks are slightly different from those used in \cite{OCCAM}, for readers reference, we report the following summary statistics for each network: (1) number of nodes $n$; (2) number of communities $K$; (3) average node degree $\bar{d}$ where $\bar{d}=\sum_{i=1}^{n}D(i,i)/n$; (4) density $\sum_{i,j}A(i,j)/(n(n-1))$, i.e., the overall edge probability; (5) the proportion of overlapping nodes $r_{o}$, i.e., $r_{o}=\frac{\mathrm{number~of~nodes~with~mixed~membership}}{n}$. We report the means and standard deviations of these measurements in Table \ref{dataSNAP}.
\begin{table}[h!]
\centering
\caption{Mean (SD) of summary statistics for ego-networks.}
\label{dataSNAP}
\begin{tabular}{cccccccccc}
\toprule
&\#Networks&$n$&$K$&$\bar{d}$&Density&$r_{o}$\\
\midrule
Facebook&7&236.57&3&30.61&0.15&0.009\\
&-&(228.53)&(1.15)&(29.41)&(0.058)&(0.008)\\
\hline
GooglePlus&58&433.22&2.22&66.81&0.18&0.005\\
&-&(327.70)&(0.46)&(65.2)&(0.11)&(0.005)\\
\hline
Twitter&255&60.64&2.63&17.87&0.33&0.02\\
&-&(30.77)&(0.83)&(9.97)&(0.17)&(0.008)\\
\bottomrule
\end{tabular}
\end{table}

We present the average mixed Hamming error rates over each of the social platforms and the corresponding standard deviation in Table \ref{ErrorSNAP}.
For the Facebook platform which only has 7 networks, the proposed Mixed-RSC method has smallest averaged error rate, 0.2473, which is slightly smaller than 0.2483 for  Mixed-cone-DCMMSB and 0.2496 for Mixed-SCORE. OCCAM has the largest averaged error rate. When we turn to GooglePlus networks we can find that the Mixed-RSC performs much better than other four methods., and OCCAM, GeoNMF and SVM-cone-DCMMSB have similar results. The averaged error rate for Mixed-RSC is 0.3182, while error rates for other methods are all larger than 0.35, and it for Mixed-SCORE even reaches to 0.3766. The Twitter has a large number of networks, 255.  The averaged error rates for the proposed method is 0.2601 which is the smallest value among all compared methods. OCCAM and GeoNMF have similar results, 0.2864 and 0.2858 respectively. Mixed-SCORE has the largest averaged error rate for Twitter's networks.
In all, Mixed-RSC always outperforms Mixed-SCORE, OCCAM, GeoNMF and SVM-cone-DCMMSB for all the networks in these three platforms. From Table \ref{dataSNAP}, we see that $\bar{d}$ is much smaller than the network size $n$, suggesting that most SNAP-ego networks are sparse. Our Mixed-RSC enjoys better performances on empirical networks because it is designed based on regularized Laplacian matrix which can successfully detect sparse networks with a good choice of $\tau$, as the discussion after Theorem \ref{Main}.

\begin{table}
\centering
\caption{Mean (SD) of mixed-Hamming error rates for ego-networks.}
\label{ErrorSNAP}
\begin{tabular}{cccccccccc}
\toprule
&Facebook&GooglePlus&Twitter\\
\midrule
Mixed-SCORE&0.2496 (0.1322)&0.3766 (0.1053)&0.3088 (0.1296)\\
OCCAM&0.2610 (0.1367)&0.3564 (0.1210)&0.2864 (0.1406)\\
GeoNMF&0.2537 (0.1266)&0.3520 (0.1078)&0.2858 (0.1292)\\
SVM-cone-DCMMSB&0.2483 (0.1496)&0.3563 0.1047)&0.2985 (0.1327)\\
\hline
Mixed-RSC&\textbf{0.2473} (0.1340)&\textbf{0.3182} (0.1259)&\textbf{0.2601} (0.1378)\\
\bottomrule
\end{tabular}
\end{table}
\section{Conclusion}\label{sec7}
In this paper, we propose a regularized spectral clustering method Mixed-RSC to mixed membership community detection under the DCMM model and study the impact of regularized Laplacian matrix on spectral clustering with the proposed method. We show the consistency of the estimation of Mixed-RSC under mild conditions. By analyzing the theoretical results, we find the optimal choice of the regularization parameter $\tau$ for our Mixed-RSC. 
We also compared our theoretical results with two previous works \citep{mixedSCORE, mao2020estimating}, and find that our error bound is consistent with \cite{mixedSCORE} and competitive with \cite{mao2020estimating}.
Furthermore, our theoretical results match the classical separation condition of a network with two equal size clusters and the sharp threshold of the Erd\"os-R\'enyi  random graph $G(n,p)$. Numerically, Mixed-RSC enjoys competitive performances with the benchmark methods in simulated networks and has excellent performances in empirical data.
\section*{Acknowledgements}
The authors would like to thank Dr. Zhang Yuan  (the first author of the OCCAM method \citep{OCCAM}) for sharing the SNAP ego-networks with us.

\bibliographystyle{agsm} 
\bibliography{reference}    

@article{lei2019unified,
	title="Unified $\ell_{2\rightarrow\infty}$ Eigenspace Perturbation Theory for Symmetric Random Matrices",
	author="Lihua {Lei}",
	journal="arXiv: Probability",
	notes="Sourced from Microsoft Academic - https://academic.microsoft.com/paper/2972485853",
	year="2019"
}

@article{lei2015consistency,
	title={Consistency of spectral clustering in stochastic block models},
	author={Lei, Jing and Rinaldo, Alessandro and others},
	journal={Annals of Statistics},
	volume={43},
	number={1},
	pages={215--237},
	year={2015},
	publisher={Institute of Mathematical Statistics}
}

@book{blum2020foundations,
	title="Foundations of Data Science",
	author="Avrim {Blum} and John {Hopcroft} and Ravindran {Kannan.}",
	number="1",
	pages="1--465",
	notes="Sourced from Microsoft Academic - https://academic.microsoft.com/paper/2304387544",
	year="2020"
}

@article{abbe2017community,
	title="Community detection and stochastic block models: recent developments",
	author="Emmanuel {Abbe}",
	journal="arXiv preprint arXiv:1703.10146",
	notes="Sourced from Microsoft Academic - https://academic.microsoft.com/paper/2602671714",
	year="2017"
}

@inbook{erdos2011on,
	author = {P. Erd\"os and A. R\'enyi},
	doi = {doi:10.1515/9781400841356.38},
	title = {'On the evolution of random graphs',The Structure and Dynamics of Networks},
	year = {2011},
	publisher = {Princeton University Press},
	pages = {38--82}
}

@article{chen2020spectral,
	title="Spectral Methods for Data Science: A Statistical Perspective",
	author="Yuxin {Chen} and Yuejie {Chi} and Jianqing {Fan} and Cong {Ma}",
	journal="arXiv preprint arXiv:2012.08496",
	notes="Sourced from Microsoft Academic - https://academic.microsoft.com/paper/3111358795",
	year="2020"
}

@article{Newman2007,
	title="Mixture models and exploratory analysis in networks",
	author="M.E.J.  {Newman} and EA {Leicht}",
	journal="Proceedings of the National Academy of Sciences",
	volume="104",
	number="23",
	pages="9564--9569",
	year="2007"
}

@article{Newman2004,
	title="Detecting community structure in networks",
	author="M.E.J.  {Newman}",
	journal="The European Physical Journal B",
	volume="38",
	pages="321--330",
	year="2004"
}

@article{Mcauley2012snap,
	title="Learning to discover social circles in ego networks",
	author="Julian {McAuley} and Jure {Leskovec}",
	journal="In Advances in Neural Information Processing Systems 25",
	volume="1",
	number="",
	pages="539--547",
	year="2012"
}

@article{SRSC,
	title="Impact of regularization on spectral clustering under the mixed membership stochasticblock model",
	author="Huan {Qing} and Jingli {Wang}",
	journal="arXiv preprint arXiv:2107.14705",
	year="2021"
}

@article{tropp2012user,
	title="User-Friendly Tail Bounds for Sums of Random Matrices",
	author="Joel A. {Tropp}",
	journal="Foundations of Computational Mathematics",
	volume="12",
	number="4",
	pages="389--434",
	notes="Sourced from Microsoft Academic - https://academic.microsoft.com/paper/2107411554",
	year="2012"
}

@book{chung2006complex,
	title={Complex graphs and networks},
	author={Chung, Fan and Chung, Fan RK and Graham, Fan Chung and Lu, Linyuan and Chung, Kian Fan and others},
	number={107},
	year={2006},
	publisher={American Mathematical Soc.}
}

@article{mao2020estimating,
	title="Estimating Mixed Memberships With Sharp Eigenvector Deviations",
	author="Xueyu {Mao} and Purnamrita {Sarkar} and Deepayan {Chakrabarti}",
	journal="Journal of the American Statistical Association",
	pages="1--13",
	notes="Sourced from Microsoft Academic - https://academic.microsoft.com/paper/3017845708",
	year="2020"
}

@article{goldenberg2010a,
	title="A Survey of Statistical Network Models",
	author="Anna {Goldenberg} and Alice X. {Zheng} and Stephen E. {Fienberg} and Edoardo M. {Airoldi}",
	journal="Foundations and Trends® in Machine Learning",
	volume="2",
	number="2",
	pages="129--233",
	notes="Sourced from Microsoft Academic - https://academic.microsoft.com/paper/2032005951",
	year="2010"
}

@article{mixedSCORE,
	title="Estimating network memberships by simplex vertex hunting",
	author="Jiashun {Jin} and Zheng Tracy {Ke} and Shengming {Luo}",
	journal="arXiv preprint arXiv:1708.07852",
	notes="Sourced from Microsoft Academic - https://academic.microsoft.com/paper/2753220974",
	year="2017"
}

@inproceedings{MaoSVM,
	title="Overlapping Clustering Models, and One (class) SVM to Bind Them All",
	author="Xueyu {Mao} and Purnamrita {Sarkar} and Deepayan {Chakrabarti}",
	booktitle="Advances in Neural Information Processing Systems",
	volume="31",
	pages="2126--2136",
	notes="Sourced from Microsoft Academic - https://academic.microsoft.com/paper/2963092149",
	year="2018"
}

@article{yu2015a,
	title={{A useful variant of the Davis--Kahan theorem for statisticians}},
	author="Yi {Yu} and Tengyao {Wang} and Richard J. {Samworth}",
	journal="Biometrika",
	volume="102",
	number="2",
	pages="315--323",
	notes="Sourced from Microsoft Academic - https://academic.microsoft.com/paper/2088911135",
	year="2015"
}

@article{OCCAM,
	title={{Detecting overlapping communities in networks using spectral methods}},
	author="Yuan {Zhang} and Elizaveta {Levina} and Ji {Zhu}",
	journal="SIAM Journal on Mathematics of Data Science",
	volume="2",
	number="2",
	pages="265--283",
	notes="Sourced from Microsoft Academic - https://academic.microsoft.com/paper/3015216926",
	year="2020"
}

@inproceedings{RSC,
	title={{Regularized spectral clustering under the degree-corrected stochastic blockmodel}},
	author="Tai {Qin} and Karl {Rohe}",
	booktitle="Advances in Neural Information Processing Systems 26",
	pages="3120--3128",
	notes="Sourced from Microsoft Academic - https://academic.microsoft.com/paper/2963664118",
	year="2013"
}

@article{SCORE,
	title={{Fast community detection by SCORE}},
	author="Jiashun {Jin}",
	journal="Annals of Statistics",
	volume="43",
	number="1",
	pages="57--89",
	notes="Sourced from Microsoft Academic - https://academic.microsoft.com/paper/2027966435",
	year="2015"
}

@article{DCSBM,
	title="Stochastic blockmodels and community structure in networks",
	author="Brian {Karrer} and M. E. J. {Newman}",
	journal="Physical Review E",
	volume="83",
	number="1",
	pages="16107",
	notes="Sourced from Microsoft Academic - https://academic.microsoft.com/paper/2119998616",
	year="2011"
}

@article{SBM,
	title="Stochastic blockmodels: First steps",
	author="Paul W. {Holland} and Kathryn Blackmond {Laskey} and Samuel {Leinhardt}",
	journal="Social Networks",
	volume="5",
	number="2",
	pages="109--137",
	notes="Sourced from Microsoft Academic - https://academic.microsoft.com/paper/2102907934",
	year="1983"
}

@article{MMSB,
	title="Mixed Membership Stochastic Blockmodels",
	author="Edoardo M. {Airoldi} and David M. {Blei} and Stephen E. {Fienberg} and Eric P. {Xing}",
	journal="Journal of Machine Learning Research",
	volume="9",
	pages="1981--2014",
	notes="Sourced from Microsoft Academic - https://academic.microsoft.com/paper/2107107106",
	year="2008"
}

@article{jin2017a,
	title="A Sharp Lower Bound for Mixed-membership Estimation",
	author="Jiashun {Jin} and Zheng Tracy {Ke}",
	journal="arXiv preprint arXiv:1709.05603",
	notes="Sourced from Microsoft Academic - https://academic.microsoft.com/paper/2755908554",
	year="2017"
}

@article{football,
	title={{Community structure in social and biological networks}},
	author={Girvan, Michelle and Newman, Mark EJ},
	journal={Proceedings of the National Academy of Sciences},
	volume={99},
	number={12},
	pages={7821--7826},
	year={2002},
	publisher={National Acad Sciences}
}

@article{fortunato2010community,
	title={Community detection in graphs},
	author={Fortunato, Santo},
	journal={Physics Reports},
	volume={486},
	number={3},
	pages={75--174},
	year={2010}
}

@article{papadopoulos2012community,
	title={{Community detection in social media}},
	author={Papadopoulos, Symeon and Kompatsiaris, Yiannis and Vakali, Athena and Spyridonos, Ploutarchos},
	journal={Data Mining and Knowledge Discovery},
	volume={24},
	number={3},
	pages={515--554},
	year={2012}
}

@article{Tutorial,
	title="A tutorial on spectral clustering",
	author="Ulrike {Luxburg}",
	journal="Statistics and Computing",
	volume="17",
	number="4",
	pages="395--416",
	notes="Sourced from Microsoft Academic - https://academic.microsoft.com/paper/2132914434",
	year="2007"
}

@article{Fortunato2016Community,
	title={Community detection in networks: A user guide},
	author={ Fortunato, Santo  and  Hric, Darko },
	journal={Physics Reports},
	year={2016},
	volume={659},
	pages={1--44}
}

@article{Cai2016survey,
	author = {Cai, Qing and Ma, Lijia and Gong, Maoguo and Tian, Dayong},
	title = {A survey on network community detection based on evolutionary computation},
	year = {2016},
	volume = {8},
	number = {2},
	issn = {1758-0366},
	journal = {International Journal of Bio-Inspired Computation},
	pages = {84–98}
}

@article{ZHANG2007ident,
	title = "Identification of overlapping community structure in complex networks using fuzzy c-means clustering",
	journal = "Physica A: Statistical Mechanics and its Applications",
	volume = "374",
	number = "1",
	pages = "483 -- 490",
	year = "2007",
	author = "Shihua Zhang and Rui-Sheng Wang and Xiang-Sun Zhang"
}

@article{Lancichinetti2011,
	author = {A. Lancichinetti and F. Radicchi and JJ. Ramasco and S. Fortunato},
	title = {Finding statistically significant communities in networks},
	year = {2011},
	volume = {6},
	number = {4},
	issn = {1758-0366},
	journal = { PLoS ONE },
	pages = {e18961}
}

@ARTICLE{Gillis2014ieee,
	author={N. {Gillis} and S. A. {Vavasis}},
	journal={IEEE Transactions on Pattern Analysis and Machine Intelligence},
	title={Fast and robust recursive algorithmsfor separable nonnegative matrix factorization},
	year={2014},
	volume={36},
	number={4},
	pages={698-714},
	doi={10.1109/TPAMI.2013.226}
}

@article{GeoNMF,
	title="On mixed memberships
	and symmetric nonnegative matrix factorizations",
	author="Xueyu {Mao} and Purnamrita {Sarkar} and Deepayan {Chakrabarti}",
	journal="International Conference on Machine Learning",
	volume={70},
	pages="2324--2333",
	year="2017"
}

\section*{Supplementary material}
In this section, we provide the technical proofs of lemmas and theorems in the main manuscript. 
\appendix
\section{Ideal Cone}
\subsection{Proof of Lemma 2.1}
\begin{proof}
	Before we present the proof of Lemma 2.1, first we give one simple lemma.
	\begin{lem}\label{PiX}
		For  any membership matrix $\Pi\in\mathbb{R}^{n\times K}$ whose $i$-th row $[\Pi(i,1), \Pi(i,2), \ldots, \Pi(i,K)]$ is the PMF of node $i$ for $1\leq i\leq n$, such that each  community has at least one pure node, then for any $X,\tilde{X}\in\mathbb{R}^{K\times K}$, if $\Pi X=\Pi\tilde{X}$, we have $X=\tilde{X}$.
	\end{lem}
	\begin{proof}
		Assume that node $i$ is a pure node such that $\Pi(i,k)=1$, then the $i$-th row of $\Pi X$ is $[X(k,1), X(k,2), \ldots, X(k,K)]$ (i.e., the $i$-th row of $\Pi X$ is the $k$-th row of $X$ if $\Pi(i,k)=1$); similarly, the $i$-th row of $\Pi \tilde{X}$ is the $k$-th row of $\tilde{X}$. Since $\Pi X=\Pi\tilde{X}$, we have $[X(k,1), X(k,2), \ldots, X(k,K)]=[\tilde{X}(k,1), \tilde{X}(k,2), \ldots, \tilde{X}(k,K)]$ for $1\leq k\leq K$, hence $X=\tilde{X}$.
	\end{proof}
	Since $\mathscr{L}_{\tau}V=VE$ and $\mathscr{L}_{\tau}=\tilde{\Theta}\Pi P\Pi'\tilde{\Theta}$, if we assume that there exists $B$ such that $V=\tilde{\Theta}\Pi B$, then we have
	\begin{align*}
	\mathscr{L}_{\tau}\tilde{\Theta}\Pi B&=\tilde{\Theta}\Pi BE\\
	&\Downarrow\\
	\tilde{\Theta}\Pi FB&=\tilde{\Theta}\Pi BE\\
	&\Downarrow\\
	\Pi FB&=\Pi BE
	\end{align*}
	which gives that $\Pi(FB-BE)=0$, since we assume that each row community has at least one pure node,  by Lemma \ref{PiX}, we have $FB=BE$. Therefore $B$ exists and its $i$-th column is the right eigenvector of $F$, and $\lambda_{k}$ is the $k$-th eigenvalue of $F$ for $1\leq k\leq K$. Further more, if there exists another $\tilde{B}$ such that $V=\tilde{\Theta}\Pi B=\tilde{\Theta}\Pi\tilde{B}$, then we have $\Pi(B-\tilde{B})=0$, since each community has at least one pure node, by Lemma \ref{PiX}, we have $\tilde{B}=B$, hence $B$ is unique. Note that, for $1\leq k\leq K$, though the $k$-th column of $B$ is the right eigenvector of $F$, it may not be unit-norm.
	
	Since $\mathrm{rank}(P)=K$, we have $\mathscr{L}_{\tau}=\tilde{\Theta}\Pi P\Pi'\tilde{\Theta}=VEV'$. Without loss of generality, reorder the nodes such that $\Pi(\mathcal{I},:)=I$, then we have $V(\mathcal{I},:)EV'=\tilde{\Theta}(\mathcal{I},\mathcal{I})P \Pi'\tilde{\Theta}$. Now $VE=\mathscr{L}_{\tau}V=\tilde{\Theta}\Pi P\Pi'\tilde{\Theta}V=\tilde{\Theta}\Pi (P\Pi'\tilde{\Theta})V=\tilde{\Theta}\Pi (\tilde{\Theta}^{-1}(\mathcal{I},\mathcal{I})V(\mathcal{I},:)EV')V= \tilde{\Theta}\Pi \tilde{\Theta}^{-1}(\mathcal{I},\mathcal{I})V(\mathcal{I},:)E$, right multiplying $E^{-1}$ gives $V=\tilde{\Theta}\Pi\tilde{\Theta}^{-1}(\mathcal{I},\mathcal{I})V(\mathcal{I},:)$, i.e., $B$ can also be written as $B=\tilde{\Theta}^{-1}(\mathcal{I},\mathcal{I})V(\mathcal{I},:)$. And $B$ is full rank surely.
\end{proof}

\subsection{Proof of Lemma 2.3}
\begin{proof}
	For convenience, set $M_{1}=\Pi\tilde{\Theta}^{-1}(\mathcal{I},\mathcal{I})V(\mathcal{I},:)$, since
	\begin{align*}
	V=\tilde{\Theta}\Pi\tilde{\Theta}^{-1}(\mathcal{I},\mathcal{I})V(\mathcal{I},:),
	\end{align*}
	we have $V=\tilde{\Theta}M_{1}$, which gives that $V(i,:)=\tilde{\theta}_{i}M_{1}(i,:)$. Therefore, $V_{*,1}(i,:)=\frac{V(i,:)}{\|V(i,:)\|_{F}}=\frac{M_{1}(i,:)}{\|M_{1}(i,:)\|_{F}}$, which gives that
	\begin{flalign*}
	V_{*,1}&=\begin{bmatrix}
	\tiny
	M_{1}(1,:)/\|M_{1}(1,:)\|_{F}\\
	M_{1}(2,:)/\|M_{1}(2,:)\|_{F}\\
	\vdots\\
	M_{1}(n,:)/\|M_{1}(n,:)\|_{F}
	\end{bmatrix}=\begin{bmatrix}
	\frac{1}{\|M_{1}(1,:)\|_{F}} &  & & \\
	& \frac{1}{\|M_{1}(2,:)\|_{F}}& &\\
	& & \ddots&\\
	&&&\frac{1}{\|M_{1}(n,:)\|_{F}}
	\end{bmatrix}M_{1}\\
	&=\begin{bmatrix}
	\frac{1}{\|M_{1}(1,:)\|_{F}} &  & & \\
	& \frac{1}{\|M_{1}(2,:)\|_{F}}& &\\
	& & \ddots&\\
	&&&\frac{1}{\|M_{1}(n,:)\|_{F}}
	\end{bmatrix}\Pi \tilde{\Theta}^{-1}(\mathcal{I},\mathcal{I})V(\mathcal{I},:)\\
	&=\begin{bmatrix}
	\Pi(1,:)/\|M_{1}(1,:)\|_{F}\\
	\Pi(2,:)/\|M_{1}(2,:)\|_{F}\\
	\vdots\\
	\Pi(n,:)/\|M_{1}(n,:)\|_{F}
	\end{bmatrix}\tilde{\Theta}^{-1}(\mathcal{I},\mathcal{I})V(\mathcal{I},:)\\
	&=\begin{bmatrix}
	\Pi(1,:)/\|M_{1}(1,:)\|_{F}\\
	\Pi(2,:)/\|M_{1}(2,:)\|_{F}\\
	\vdots\\
	\Pi(n,:)/\|M_{1}(n,:)\|_{F}
	\end{bmatrix}\tilde{\Theta}^{-1}(\mathcal{I},\mathcal{I})N_{V}^{-1}(\mathcal{I},\mathcal{I})N_{V}(\mathcal{I},\mathcal{I})V(\mathcal{I},:)\\
	&=\begin{bmatrix}
	\Pi(1,:)/\|M_{1}(1,:)\|_{F}\\
	\Pi(2,:)/\|M_{1}(2,:)\|_{F}\\
	\vdots\\
	\Pi(n,:)/\|M_{1}(n,:)\|_{F}
	\end{bmatrix}\tilde{\Theta}^{-1}(\mathcal{I},\mathcal{I})N_{V}^{-1}(\mathcal{I},\mathcal{I})V_{*}(\mathcal{I},:).
	\end{flalign*}
	Therefore, we have
	\begin{align*}
	Y_{1}=\begin{bmatrix}
	\Pi(1,:)/\|M_{1}(1,:)\|_{F}\\
	\Pi(2,:)/\|M_{1}(2,:)\|_{F}\\
	\vdots\\
	\Pi(n,:)/\|M_{1}(n,:)\|_{F}
	\end{bmatrix}\tilde{\Theta}^{-1}(\mathcal{I},\mathcal{I})N_{V}^{-1}(\mathcal{I},\mathcal{I})=N_{M_{1}}\Pi\tilde{\Theta}^{-1}(\mathcal{I},\mathcal{I})N_{V}^{-1}(\mathcal{I},\mathcal{I}),
	\end{align*}
	where $N_{M_{1}}=\begin{bmatrix}
	\frac{1}{\|M_{1}(1,:)\|_{F}} &  & & \\
	& \frac{1}{\|M_{1}(2,:)\|_{F}}& &\\
	& & \ddots&\\
	&&&\frac{1}{\|M_{1}(n,:)\|_{F}}
	\end{bmatrix}$.
	Sure, all entries of $Y_{1}$ are nonnegative. And since we assume that each community has at least one pure node, no row of $Y_{1}$ is 0.
	
	Then we prove that $V_{*,1}(i,:)=V_{*,1}(j,:)$ when $\Pi(i,:)=\Pi(j,:)$. For $1\leq i\leq n$, we have
	\begin{flalign*}
	V_{*,1}(i,:)&=e'_{i}V_{*,1}=e'_{i}\begin{bmatrix}
	\frac{1}{\|M_{1}(1,:)\|_{F}} &  & & \\
	& \frac{1}{\|M_{1}(2,:)\|_{F}}& &\\
	& & \ddots&\\
	&&&\frac{1}{\|M_{1}(n,:)\|_{F}}
	\end{bmatrix}M_{1}=\frac{1}{\|M_{1}(i,:)\|_{F}}e'_{i}M_{1}\\
	&=\frac{1}{\|e'_{i}M_{1}\|_{F}}e'_{i}M_{1}=\frac{1}{\|e'_{i}\Pi\tilde{\Theta}^{-1}(\mathcal{I},\mathcal{I})V(\mathcal{I},:)\|_{F}}e'_{i}\Pi\tilde{\Theta}^{-1}(\mathcal{I},\mathcal{I})V(\mathcal{I},:)\\
	&=\frac{1}{\|\Pi(i,:)\tilde{\Theta}^{-1}(\mathcal{I},\mathcal{I})V(\mathcal{I},:)\|_{F}}\Pi(i,:)\tilde{\Theta}^{-1}(\mathcal{I},\mathcal{I})V(\mathcal{I},:),
	\end{flalign*}
	which gives that if $\Pi(j,:)=\Pi(i,:)$, we have $V_{*,1}(i,:)=V_{*,1}(j,:)$.
\end{proof}
\subsection{Proof of Lemma 2.6}
\begin{proof}
	Since $I=V'V=V'(\mathcal{I},:)\tilde{\Theta}^{-1}(\mathcal{I},\mathcal{I})\Pi'\tilde{\Theta}^{2}\Pi\tilde{\Theta}^{-1}(\mathcal{I},\mathcal{I})V(\mathcal{I},:)$ and $\mathrm{rank}(V(\mathcal{I},:))=K$ (i.e., the inverse of $V(\mathcal{I},:)$ exists), we have $(V(\mathcal{I},:)V'(\mathcal{I},:))^{-1}=\tilde{\Theta}^{-1}(\mathcal{I},\mathcal{I})\Pi'\tilde{\Theta}^{2}\Pi\tilde{\Theta}^{-1}(\mathcal{I},\mathcal{I})$.
	
	Since $V_{*,1}(\mathcal{I},:)=N_{V}(\mathcal{I},\mathcal{I})V(\mathcal{I},:)$, we have
	\begin{align*}
	(V_{*,1}(\mathcal{I},:)V'_{*,1}(\mathcal{I},:))^{-1}=N_{V}^{-1}(\mathcal{I},\mathcal{I})\tilde{\Theta}^{-1}(\mathcal{I},\mathcal{I})\Pi'\tilde{\Theta}^{2}\Pi\tilde{\Theta}^{-1}(\mathcal{I},\mathcal{I})N_{V}^{-1}(\mathcal{I},\mathcal{I}).
	\end{align*}
	Since all entries of $N_{V}^{-1}(\mathcal{I},\mathcal{I}), \Pi, \tilde{\Theta}$ and nonnegative and $N,\tilde{\Theta}$ are diagonal matrices, we see that all entries of $(V_{*,1}(\mathcal{I},:)V'_{*,1}(\mathcal{I},:))^{-1}$ are nonnegative and its diagonal entries are strictly positive, hence we have $(V_{*,1}(\mathcal{I},:)V'_{*,1}(\mathcal{I},:))^{-1}\mathbf{1}>0$.
\end{proof}
\subsection{Proof of Lemma 2.8}
\begin{proof}
	Set $M_{2}=\Pi\tilde{\Theta}^{-1}(\mathcal{I},\mathcal{I})V_{2}(\mathcal{I},:)$. Since $V_{2}=\tilde{\Theta}\Pi\tilde{\Theta}^{-1}(\mathcal{I},\mathcal{I})V_{2}(\mathcal{I},:)$,
	we have $V_{2}=\tilde{\Theta}M_{2}$.
	Follow a similar proof of Lemma 2.3, we have
	$Y_{2}=N_{M_{2}}\Pi\tilde{\Theta}^{-1}(\mathcal{I},\mathcal{I})N_{V_{2}}^{-1}(\mathcal{I},\mathcal{I})$,
	where $N_{M_{2}}$ is a $n\times n$ diagonal matrix whose $i$-th diagonal entry is $\frac{1}{\|M_{2}(i,:)\|_{F}}$. Meanwhile, all entries of $Y_{2}$ are nonnegative and no row of $Y_{2}$ is 0. The last statement can be proved easily by following similar proof as the one in Lemma 2.1 and we omit it here.
\end{proof}
\subsection{Proof of Lemma 2.9}
\begin{proof}
	For $1\leq i\leq n$, by basic algebra, we have $V_{2}(i,:)=(VV')(i,:)=V(i,:)V'$, which gives that $N_{V_{2}}(i,i)=\frac{1}{\|V_{2}(i,:)\|_{F}}=\frac{1}{\|V(i,:)V'\|_{F}}=\frac{1}{\|V(i,:)\|_{F}}$ where the last equality holds by Lemma A.1 in \cite{yu2015a}. Hence, we have $N_{V}\equiv N_{V_{2}}$. Then, by basic algebra, we have $V_{*,2}(\mathcal{I},:)V=N_{V_{2}}(\mathcal{I},\mathcal{I})V_{2}V=N_{V}(\mathcal{I},\mathcal{I})VV'V=N_{V}(\mathcal{I},\mathcal{I})V\equiv V_{*,1}(\mathcal{I},:)$. By basic algebra, we have $V_{2}(\mathcal{I},:)=(VV')(\mathcal{I},:)=V(\mathcal{I},:)V'$, which gives  $V_{*,2}(\mathcal{I},:)=N_{V_{2}}(\mathcal{I},\mathcal{I})V_{2}(\mathcal{I},:)=N_{V}(\mathcal{I},\mathcal{I})V(\mathcal{I},:)V'\equiv V_{*,1}(\mathcal{I},:)V'$. Then we have  $V_{*,2}(\mathcal{I},:)V'_{*,2}(\mathcal{I},:)=V_{*,1}(\mathcal{I},:)V'VV'_{*,1}(\mathcal{I},:)\equiv V_{*,1}(\mathcal{I},:)V'_{*,1}(\mathcal{I},:)$. Meanwhile, we also have $V_{*,2}=N_{V_{2}}V_{2}=N_{V}VV'=V_{*,1}V'$. Based on the above equalities, we have
	\begin{align*}
	&Y_{2}=V_{*,2}V'_{*,2}(\mathcal{I},:)(V_{*,2}(\mathcal{I},:)V'_{*,2}(\mathcal{I},:))^{-1}=V_{*,1}V'VV'_{*,1}(\mathcal{I},:)(V_{*,1}(\mathcal{I},:)V'_{*,1}(\mathcal{I},:))^{-1}\equiv Y_{1},\\
	&Y_{\bullet,2}=V_{2}V'_{*,2}(\mathcal{I},:)(V_{*,2}(\mathcal{I},:)V'_{*,2}(\mathcal{I},:))^{-1}=VV'VV'_{*,1}(\mathcal{I},:)(V_{*,1}(\mathcal{I},:)V'_{*,1}(\mathcal{I},:))^{-1}\equiv Y_{\bullet,1},\\
	&J_{2}=\sqrt{\mathrm{diag}(N_{V_{2}}(\mathcal{I},\mathcal{I})V(\mathcal{I},:))EV'(\mathcal{I},:)N_{V_{2}}(\mathcal{I},\mathcal{I})}\\
	&=\sqrt{\mathrm{diag}(N_{V}(\mathcal{I},\mathcal{I})V(\mathcal{I},:))EV'(\mathcal{I},:)N_{V}(\mathcal{I},\mathcal{I})}\equiv J_{1}.
	\end{align*}
	Since $Z_{1}=Y_{\bullet,1}J_{1}, Z_{2}=Y_{\bullet,2}J_{2}$, we have $Z_{1}\equiv Z_{2}$.  Meanwhile, note that $M_{1}=\Pi\tilde{\Theta}^{-1}(\mathcal{I},\mathcal{I})V_{1}(\mathcal{I},:)\in\mathbb{R}^{n\times K}, M_{2}=\Pi\tilde{\Theta}^{-1}(\mathcal{I},\mathcal{I})V_{2}(\mathcal{I},:)\in\mathbb{R}^{n\times n}$ gives $M_{1}\neq M_{2}$, we still have $N_{M_{1}}\equiv N_{M_{2}}$ based on the fact that $N_{M_{2}}(i,i)=\frac{1}{\|M_{2}(i,:)\|_{F}}=\frac{1}{\|e'_{i}M_{2}\|_{F}}=\frac{1}{\|e'_{i}\Pi\tilde{\Theta}^{-1}(\mathcal{I},\mathcal{I})V_{2}(\mathcal{I},:)\|_{F}}=\frac{1}{\|e'_{i}\Pi\tilde{\Theta}^{-1}(\mathcal{I},\mathcal{I})V(\mathcal{I},:)V'\|_{F}}\equiv N_{M_{1}}(i,i)$ for $1\leq i\leq n$.

	Similarly, we have $N_{\hat{V}}\equiv N_{\hat{V}_{2}}$, where $N_{\hat{V}}$ is the diagonal matrix such that $\hat{V}_{*,1}=N_{\hat{V}}\hat{V}$. By Lemma G.1 in \cite{MaoSVM}, the outputs of the SVM-cone algorithm using $\hat{V}_{*,1}$ and $\hat{V}_{*,2}$ as inputs are same, therefore we have $\mathcal{\hat{I}}_{1}\equiv \mathcal{\hat{I}}_{2}$. Then, follow a similar analysis as that of the ideal case, for the empirical case, we have $ \hat{V}_{*,1}(\hat{\mathcal{I}}_{1},:)\hat{V}'_{*,1}(\hat{\mathcal{I}}_{1},:)\equiv \hat{V}_{*,2}(\hat{\mathcal{I}}_{2},:)\hat{V}'_{*,2}(\hat{\mathcal{I}}_{2},:), \hat{Y}_{1}\equiv\hat{Y}_{2}, \hat{Y}_{\bullet,2}\equiv\hat{Y}_{\bullet,1}, \hat{J}_{2}\equiv\hat{J}_{1}, \hat{Z}_{2}\equiv\hat{Z}_{1}, \hat{\Pi}_{1}\equiv\hat{\Pi}_{2}$.
\end{proof}

\section{Theoretical properties for Mixed-RSC}\label{AppendixCommon}
Lemma \ref{P1} provides a further study on the Ideal Cone given in Lemma 2.3, it shows that $V_{*,1}(i,:)$ for Mixed-RSC can be written as a scaled convex combination of the $K$ rows of $V_{*,1}(\mathcal{I},:)$. Lemma \ref{P1} is consistent with Lemma A.1. in \cite{MaoSVM}. Meanwhile, Lemma \ref{P1} is one the reasons that the SVM-cone algorithm (i.e, Algorithm \ref{alg:SVMcone}) can return the corner matrix $V_{*,1}(\mathcal{I},:)$ when the inputs are $V_{*,1}$ and $K$ in the SVM-cone algorithm, for detail, refer to Appendix \ref{OneClassSVMandSVMcone}.
\begin{lem}\label{P1}
	Under $DCMM(n,P,\Theta,\Pi)$, for $1\leq i\leq n$, $V_{*,1}(i,:)$ can be written as $V_{*,1}(i,:)=r_{1}(i)\Phi_{1}(i,:)V_{*,1}(\mathcal{I},:)$, where $r_{1}(i)\geq 1$. Meanwhile, $r_{1}(i)=1$ and $\Phi_{1}(i,:)=e'_{k}$ if $i$ is a pure node such that $\Pi(i,k)=1$. Similarly, $V_{*,2}(i,:)$ can be written as $V_{*,2}(i,:)=r_{2}(i)\Phi_{2}(i,:)V_{*,2}(\mathcal{I},:)$, where $r_{2}(i)\geq 1$. Meanwhile, $r_{2}(i)=1$ and $\Phi_{2}(i,:)=e'_{k}$ if $\Pi(i,k)=1$.
\end{lem}
Lemma \ref{P2} is powerful to bound the behaviors of $\|V\|_{2\rightarrow\infty}$ and $\|V_{2}\|_{2\rightarrow\infty}$, and the result in Lemma \ref{P2} is called as the delocalization of population eigenvectors in Lemma 3.2 in \cite{mao2020estimating}.
\begin{lem}\label{P2}
	Under $DCMM(n,P,\Theta,\Pi)$, we have
	\begin{align*}
	\frac{\tilde{\theta}_{\mathrm{min}}}{\tilde{\theta}_{\mathrm{max}}\sqrt{K\lambda_{1}(\Pi'\Pi)}}\leq \|V(i,:)\|_{F}\leq\frac{\tilde{\theta}_{\mathrm{max}}}{\tilde{\theta}_{\mathrm{min}}\sqrt{\lambda_{K}(\Pi'\Pi)}},
	1\leq i\leq n.
	\end{align*}
\end{lem}
Note that since $V_{2}(i,:)=e'_{i}VV'=V(i,:)V'$, by Lemma A.1 in \cite{yu2015a}, we have $\|V_{2}(i,:)\|_{F}=\|V(i,:)V'\|_{F}=\|V(i,:)\|_{F}$, therefore results in Lemma \ref{P2} also holds for $V_{2}(i,:)$.
\begin{lem}\label{P3}
	Under $DCMM(n,P,\Theta,\Pi)$, we have
	\begin{align*}
	\lambda_{1}(V_{*,1}(\mathcal{I},:)V'_{*,1}(\mathcal{I},:))\leq \frac{\tilde{\theta}^{2}_{\mathrm{max}}K\kappa(\Pi'\Pi)}{\tilde{\theta}^{2}_{\mathrm{min}}}\mathrm{~and~}\lambda_{K}(V_{*,1}(\mathcal{I},:)V'_{*,1}(\mathcal{I},:))\geq \frac{\tilde{\theta}^{2}_{\mathrm{min}}\kappa^{-1}(\Pi'\Pi)}{\tilde{\theta}^{2}_{\mathrm{max}}}.
	\end{align*}
\end{lem}
Lemma \ref{P3} will be frequently used in this paper since it is useful for obtaining the bound of $\lambda_{K}(V_{*,1}(\mathcal{I},:)V'_{*,1}(\mathcal{I},:))$ for further study.
\begin{lem}\label{P4}
	Under $DCMM(n, P, \Theta, \Pi)$, we have
	\begin{align*}
	|\lambda_{K}|\geq \tilde{\theta}^{2}_{\mathrm{min}}|\lambda_{K}(P)|\lambda_{K}(\Pi'\Pi)\mathrm{~and~} \lambda_{1}\leq 1.
	\end{align*}
\end{lem}
Lemma \ref{P4} gives the lower bound of $|\lambda_{K}|$ and upper bound of $\lambda_{1}$.
\subsection{Proof of Lemma \ref{P1}}
\begin{proof}
	Since $V_{*,1}=YV_{*,1}(\mathcal{I},:)$, for $1\leq i\leq n$, we have
	\begin{align*}
	V_{*,1}(i,:)=Y(i,:)V_{*,1}(\mathcal{I},:)=Y(i,:)\mathbf{1}\frac{Y(i,:)}{Y(i,:)\mathbf{1}}V_{*,1}(\mathcal{I},:)=r(i)\Phi(i,:)V_{*,1}(\mathcal{I},:),
	\end{align*}
	where we set $r_{1}(i)=Y(i,:)\textbf{1}$, $\Phi_{1}(i,:)=\frac{Y(i,:)}{Y(i,:)\mathbf{1}}$, and $\mathbf{1}$ is a $K\times 1$  vector with all entries being ones.
	
	By the proof of Lemma 2.3, we know that $Y(i,:)=\frac{\Pi(i,:)}{\|M_{1}(i,:)\|_{F}}\tilde{\Theta}^{-1}(\mathcal{I},\mathcal{I})N^{-1}(\mathcal{I},\mathcal{I})$, where $M_{1}=\Pi\tilde{\Theta}^{-1}(\mathcal{I},\mathcal{I})V(\mathcal{I},:)$.  For convenience, set $T=\tilde{\Theta}^{-1}(\mathcal{I},\mathcal{I}), Q=N^{-1}(\mathcal{I},\mathcal{I})$, and $R=V(\mathcal{I},:)$ (note that such setting of $T,Q, R$ is only for notation convenience in the proof of Lemma \ref{P1}).

	One the one hand, if node $i$ is pure such that $\Pi(i,k)=1$ for certain $k$ among $\{1,2,\ldots,K\}$ (i.e., $\Pi(i,:)=e_{k}$ if $\Pi(i,k)=1$), we have $M(i,:)=\Pi(i,:)\tilde{\Theta}^{-1}(\mathcal{I},\mathcal{I})V(\mathcal{I},:)=T(k,k)R(k,:)$, and $\Pi(i,:)TQ=T(k,k)Q(k,:)$, which give that $Y(i,:)=\frac{T(k,k)Q(k,:)}{\|T(k,k)R(k,:)\|_{F}}=\frac{Q(k,:)}{\|R(k,:)\|_{F}}$. Recall that the $k$-th diagonal entry of $N^{-1}(\mathcal{I},\mathcal{I})$ is $\|[V(\mathcal{I},:)](k,:)\|_{F}$, i.e., $Q(k,:)\mathbf{1}=\|R(k,:)\|_{F}$, which gives that $r_{1}(i)=Y(i,:)\mathbf{1}=1$ if $\Pi(i,k)=1$  and $\Phi_{1}(i,:)=e'_{k}$ if $\Pi(i,k)=1$.
	
	On the other hand, if $i$ is not a pure node, since 
		\begin{align*}
	\|M(i,:)\|_{F}=&\|\Pi(i,:)\tilde{\Theta}^{-1}(\mathcal{I},\mathcal{I})V(\mathcal{I},:)\|_{F}=\|\sum_{k=1}^{K}\Pi(i,k)T(k,k)R(k,:)\|_{F}\\
	&\leq \sum_{k=1}^{K}\Pi(i,k)T(k,k)\|R(k,:)\|_{F}
	=\sum_{k=1}^{K}\Pi(i,k)T(k,k)Q(k,k),
	\end{align*}
	 combine it with $\Pi(i,:)TQ\mathbf{1}=\sum_{k=1}^{K}\Pi(i,k)T(k,k)Q(k,k),$
	so $r_{1}(i)=\frac{Y(i,:)\mathbf{1}}{\|M(i,:)\|}_{F}=\frac{\Pi(i,:)TQ\mathbf{1}}{\|M(i,:)\|_{F}}> 1$. Following the above proof, the results for $V_{*,2}$ can be obtained. Here, we omit the details.
\end{proof}

\subsection{Proof of Lemma \ref{P2}}
\begin{proof}
	Since $I=V'V=V'(\mathcal{I},:)\tilde{\Theta}^{-1}(\mathcal{I},\mathcal{I})\Pi'\tilde{\Theta}^{2}\Pi\tilde{\Theta}^{-1}(\mathcal{I},\mathcal{I})V(\mathcal{I},:)$, we have
	\begin{align*}
	((\tilde{\Theta}^{-1}(\mathcal{I},\mathcal{I})V(\mathcal{I},:))((\tilde{\Theta}^{-1}(\mathcal{I},\mathcal{I})V(\mathcal{I},:))')^{-1}=\Pi'\tilde{\Theta}^{2}\Pi,
	\end{align*}
	which gives that
	\begin{align*}
	\mathrm{max}_{k}\|e'_{k}(\tilde{\Theta}^{-1}(\mathcal{I},\mathcal{I})V(\mathcal{I},:))\|_{F}^{2}&=\mathrm{max}_{k}e'_{k}(\tilde{\Theta}^{-1}(\mathcal{I},\mathcal{I})V(\mathcal{I},:))(\tilde{\Theta}^{-1}(\mathcal{I},\mathcal{I})V(\mathcal{I},:))'e_{k}\\
	&\leq \mathrm{max}_{\|x\|=1}x'(\tilde{\Theta}^{-1}(\mathcal{I},\mathcal{I})V(\mathcal{I},:))(\tilde{\Theta}^{-1}(\mathcal{I},\mathcal{I})V(\mathcal{I},:))'x\\
	&=\lambda_{1}((\tilde{\Theta}^{-1}(\mathcal{I},\mathcal{I})V(\mathcal{I},:))(\tilde{\Theta}^{-1}(\mathcal{I},\mathcal{I})V(\mathcal{I},:))')\\
	&=\frac{1}{\lambda_{K}(\Pi'\tilde{\Theta}^{2}\Pi)},
	\end{align*}
	where $x$ is a $K\times 1$ vector whose $l_{2}$ norm is 1. Meanwhile, we also have
	\begin{align*}
	\mathrm{min}_{k}\|e'_{k}(\tilde{\Theta}^{-1}(\mathcal{I},\mathcal{I})V(\mathcal{I},:))\|_{F}^{2}&=\mathrm{min}_{k}e'_{k}(\tilde{\Theta}^{-1}(\mathcal{I},\mathcal{I})V(\mathcal{I},:))(\tilde{\Theta}^{-1}(\mathcal{I},\mathcal{I})V(\mathcal{I},:))'e_{k}\\
	&\geq \mathrm{min}_{\|x\|=1}x'(\tilde{\Theta}^{-1}(\mathcal{I},\mathcal{I})V(\mathcal{I},:))(\tilde{\Theta}^{-1}(\mathcal{I},\mathcal{I})V(\mathcal{I},:))'x\\
	&=\lambda_{K}((\tilde{\Theta}^{-1}(\mathcal{I},\mathcal{I})V(\mathcal{I},:))(\tilde{\Theta}^{-1}(\mathcal{I},\mathcal{I})V(\mathcal{I},:))')\\
	&=\frac{1}{\lambda_{1}(\Pi'\tilde{\Theta}^{2}\Pi)},
	\end{align*}
	By Lemma 2.1, we have $V(i,:)=\tilde{\theta}_{i}\Pi(i,:)\tilde{\Theta}^{-1}(\mathcal{I},\mathcal{I})V(\mathcal{I},:)$ for $1\leq i\leq n$, which gives that
	\begin{align*}
	\|V(i,:)\|_{F}&=\|\tilde{\theta}_{i}\Pi(i,:)\tilde{\Theta}^{-1}(\mathcal{I},\mathcal{I})V(\mathcal{I},:)\|_{F}\\
	&=\tilde{\theta}_{i}\|\Pi(i,:)\tilde{\Theta}^{-1}(\mathcal{I},\mathcal{I})V(\mathcal{I},:)\|_{F}\\
	&\leq \tilde{\theta}_{i} \mathrm{max}_{i}\|\Pi(i,:)\|_{F}\mathrm{max}_{i}\|e'_{i}(\tilde{\Theta}^{-1}(\mathcal{I},\mathcal{I})V(\mathcal{I},:))\|_{F}\\
	&\leq \tilde{\theta}_{i}\mathrm{max}_{i}\|e'_{i}(\tilde{\Theta}^{-1}(\mathcal{I},\mathcal{I})V(\mathcal{I},:))\|_{F}\\
	&\leq\frac{\tilde{\theta}_{i}}{\sqrt{\lambda_{K}(\Pi'\tilde{\Theta}^{2}\Pi)}}\leq \frac{\tilde{\theta}_{\mathrm{max}}}{\sqrt{\lambda_{K}(\Pi'\tilde{\Theta}^{2}\Pi)}},
	\end{align*}
	where we use $\|\Pi(i,:)\|_{F}\leq 1$ since $\sum_{k=1}^{K}\Pi(i,k)=1$, and $e_{i}$ is a $n\times 1$ basis vector whose $i$-th entry is 1. Since $\lambda_{K}(\Pi'\tilde{\Theta}^{2}\Pi)=\lambda_{K}(\tilde{\Theta}^{2}\Pi'\Pi)\geq \lambda_{K}(\tilde{\Theta}^{2})\lambda_{K}(\Pi'\Pi)=\tilde{\theta}^{2}_{\mathrm{min}}\lambda_{K}(\Pi'\Pi)$, we have
	\begin{align*}
	\|V(i,:)\|_{F}\leq \frac{\tilde{\theta}_{\mathrm{max}}}{\tilde{\theta}_{\mathrm{min}}\sqrt{\lambda_{K}(\Pi'\Pi)}}.
	\end{align*}
	Similarly, we have
	\begin{align*}
	\|V(i,:)\|_{F}&=\|\tilde{\theta}_{i}\Pi(i,:)\tilde{\Theta}^{-1}(\mathcal{I},\mathcal{I})V(\mathcal{I},:)\|_{F}\\
	&=\tilde{\theta}_{i}\|\Pi(i,:)\tilde{\Theta}^{-1}(\mathcal{I},\mathcal{I})V(\mathcal{I},:)\|_{F}\\
	&\geq \tilde{\theta}_{i} \mathrm{min}_{i}\|\Pi(i,:)\|_{F}\mathrm{min}_{i}\|e'_{i}(\tilde{\Theta}^{-1}(\mathcal{I},\mathcal{I})V(\mathcal{I},:))\|_{F}\\
	&\geq \tilde{\theta}_{i}\mathrm{min}_{i}\|e'_{i}(\tilde{\Theta}^{-1}(\mathcal{I},\mathcal{I})V(\mathcal{I},:))\|_{F}/\sqrt{K}\\
	&\geq\frac{\tilde{\theta}_{i}}{\sqrt{K\lambda_{1}(\Pi'\tilde{\Theta}^{2}\Pi)}}\geq\frac{\tilde{\theta}_{\mathrm{min}}}{\sqrt{K\lambda_{1}(\Pi'\tilde{\Theta}^{2}\Pi)}},
	\end{align*}
	where we use the fact that $\mathrm{min}_{i}\|\Pi(i,:)\|_{F}\geq \frac{1}{\sqrt{K}}$ since $\sum_{k=1}^{K}\Pi(i,k)=1$ and all entries of $\Pi$ are nonnegative. Since $\lambda_{1}(\Pi'\tilde{\Theta}^{2}\Pi)=\lambda_{1}(\tilde{\Theta}^{2}\Pi'\Pi)\leq \lambda_{1}(\tilde{\Theta}^{2})\lambda_{1}(\Pi'\Pi)=\tilde{\theta}^{2}_{\mathrm{max}}\lambda_{1}(\Pi'\Pi)$, we have
	\begin{align*}
	\|V(i,:)\|_{F}\geq \frac{\tilde{\theta}_{\mathrm{min}}}{\tilde{\theta}_{\mathrm{max}}\sqrt{K\lambda_{1}(\Pi'\Pi)}}.
	\end{align*}
\end{proof}
\subsection{Proof of Lemma \ref{P3}}
\begin{proof}
	In this proof, we will frequently use the fact that for any two matrices $X_{1}$ and $X_{2}$, the nonzero eigenvalues of $X_{1}X_{2}$ are the same as the nonzero eigenvalues of $X_{2}X_{1}$. By the proof of Lemma 2.6, we know that $V(\mathcal{I},:)V'(\mathcal{I},:)=\tilde{\Theta}(\mathcal{I},\mathcal{I})(\Pi'\tilde{\Theta}^{2}\Pi)^{-1}\tilde{\Theta}(\mathcal{I},\mathcal{I})$,which gives
	\begin{align*}
	\lambda_{1}(V_{*,1}(\mathcal{I},:)V'_{*,1}(\mathcal{I},:))&=\lambda_{1}(N(\mathcal{I},\mathcal{I})V(\mathcal{I},:)V'(\mathcal{I},:)N(\mathcal{I},\mathcal{I}))\\
	&=\lambda_{1}(N(\mathcal{I},\mathcal{I})\tilde{\Theta}(\mathcal{I},\mathcal{I})(\Pi'\tilde{\Theta}^{2}\Pi)^{-1}\tilde{\Theta}(\mathcal{I},\mathcal{I})N(\mathcal{I},\mathcal{I}))\\
	&=\lambda_{1}(N^{2}(\mathcal{I},\mathcal{I})\tilde{\Theta}^{2}(\mathcal{I},\mathcal{I})(\Pi'\tilde{\Theta}^{2}\Pi)^{-1})\\
	&\leq\lambda^{2}_{1}(N(\mathcal{I},\mathcal{I})\tilde{\Theta}(\mathcal{I},\mathcal{I}))\lambda_{1}((\Pi'\tilde{\Theta}^{2}\Pi)^{-1})\\
	&=\lambda^{2}_{1}(N(\mathcal{I},\mathcal{I})\tilde{\Theta}(\mathcal{I},\mathcal{I}))/\lambda_{K}(\Pi'\tilde{\Theta}^{2}\Pi)\\
	&\leq(\mathrm{max}_{i\in \mathcal{I}}\tilde{\theta}(i)/\|V(i,:)\|_{F})^{2}/\lambda_{K}(\Pi'\tilde{\Theta}^{2}\Pi)\\
	&\leq K\frac{\lambda_{1}(\Pi'\tilde{\Theta}^{2}\Pi)}{\lambda_{K}(\Pi'\tilde{\Theta}^{2}\Pi)} \qquad \mathrm{by~the~proof~of~Lemma~}\ref{P2}\\
	&\leq K\frac{\tilde{\theta}^{2}_{\mathrm{max}}\lambda_{1}(\Pi'\Pi)}{\tilde{\theta}^{2}_{\mathrm{min}}\lambda_{K}(\Pi'\Pi)} \qquad \mathrm{by~the~proof~of~Lemma~}\ref{P2}\\
	&=\frac{\tilde{\theta}^{2}_{\mathrm{max}}K\kappa(\Pi'\Pi)}{\tilde{\theta}^{2}_{\mathrm{min}}},
	\end{align*}
	where we use the fact that $N(i,i)=\frac{1}{\|V(i,:)\|_{F}}$. Similarly, we have
	\begin{align*}
	\lambda_{K}(V_{*,1}(\mathcal{I},:)V'_{*,1}(\mathcal{I},:))&=\lambda_{K}(N(\mathcal{I},\mathcal{I})V(\mathcal{I},:)V'(\mathcal{I},:)N(\mathcal{I},\mathcal{I}))\\
	&=\lambda_{K}(N(\mathcal{I},\mathcal{I})\tilde{\Theta}(\mathcal{I},\mathcal{I})(\Pi'\tilde{\Theta}^{2}\Pi)^{-1}\tilde{\Theta}(\mathcal{I},\mathcal{I})N(\mathcal{I},\mathcal{I}))\\
	&=\lambda_{K}(N^{2}(\mathcal{I},\mathcal{I})\tilde{\Theta}^{2}(\mathcal{I},\mathcal{I})(\Pi'\tilde{\Theta}^{2}\Pi)^{-1})\\
	&\geq\lambda^{2}_{K}(N(\mathcal{I},\mathcal{I})\tilde{\Theta}(\mathcal{I},\mathcal{I}))\lambda_{K}((\Pi'\tilde{\Theta}^{2}\Pi)^{-1})\\
	&=\lambda^{2}_{K}(N(\mathcal{I},\mathcal{I})\tilde{\Theta}(\mathcal{I},\mathcal{I}))/\lambda_{1}(\Pi'\tilde{\Theta}^{2}\Pi)\\
	&\geq(\mathrm{min}_{i\in \mathcal{I}}\tilde{\theta}(i)/\|V(i,:)\|_{F})^{2}/\lambda_{1}(\Pi'\tilde{\Theta}^{2}\Pi)\\
	&\geq\frac{\lambda_{K}(\Pi'\tilde{\Theta}^{2}\Pi)}{\lambda_{1}(\Pi'\tilde{\Theta}^{2}\Pi)} \qquad \mathrm{by~the~proof~of~Lemma~}\ref{P2}\\
	&\geq\frac{\tilde{\theta}^{2}_{\mathrm{min}}\lambda_{K}(\Pi'\Pi)}{\tilde{\theta}^{2}_{\mathrm{max}}\lambda_{1}(\Pi'\Pi)} \qquad \mathrm{by~the~proof~of~Lemma~}\ref{P2}\\
	&=\frac{\tilde{\theta}^{2}_{\mathrm{min}}\kappa^{-1}(\Pi'\Pi)}{\tilde{\theta}^{2}_{\mathrm{max}}}.
	\end{align*}
\end{proof}
\subsection{Proof of Lemma \ref{P4}}
\begin{proof}
	Set $H=P\Pi'\tilde{\Theta}^{2}\Pi P\in\mathbb{R}^{K\times K}$. By basic algebra, $H$ is full rank and positive definite. Then we have 
	\begin{align*}
	|\lambda_{K}|&=|\lambda_{K}(\mathscr{L}_{\tau})|=|\lambda_{K}(\tilde{\Theta}\Pi P\Pi'\tilde{\Theta})|=\sqrt{\lambda_{K}(\tilde{\Theta}\Pi P\Pi'\tilde{\Theta}^{2}\Pi P\Pi'\tilde{\Theta})}\\
	&=\sqrt{\lambda_{K}(\tilde{\Theta}\Pi H\Pi'\tilde{\Theta})}=\sqrt{\lambda_{K}(\tilde{\Theta}\Pi H^{1/2}H^{1/2}\Pi'\tilde{\Theta})}=\sqrt{\lambda_{K}(H^{1/2}\Pi'\tilde{\Theta}^{2}\Pi H^{1/2})}\\
	&=\sqrt{\lambda_{K}(H\Pi'\tilde{\Theta}^{2}\Pi)}\geq \sqrt{\lambda_{K}(H)\lambda_{K}(\Pi'\tilde{\Theta}^{2}\Pi)}=\sqrt{\lambda_{K}(P\Pi'\tilde{\Theta}^{2}\Pi P)\lambda_{K}(\Pi'\tilde{\Theta}^{2}\Pi)}\\
	&=\sqrt{\lambda_{K}(P^{2}\Pi'\tilde{\Theta}^{2}\Pi)\lambda_{K}(\Pi'\tilde{\Theta}^{2}\Pi)}\geq \sqrt{\lambda_{K}(P^{2})\lambda^{2}_{K}(\Pi'\tilde{\Theta}^{2}\Pi)}\geq \sqrt{\lambda^{2}_{K}(P)\lambda^{2}_{K}(\Pi'\tilde{\Theta}^{2}\Pi)}\\
	&=|\lambda_{K}(P)|\lambda_{K}(\Pi'\tilde{\Theta}^{2}\Pi)\geq |\lambda_{K}(P)|\tilde{\theta}^{2}_{\mathrm{min}}\lambda_{K}(\Pi'\Pi),
	\end{align*}
	where we have use the fact that for any matrix $T\in\mathbb{R}^{n\times K}$ with rank $K<n$, $TT'$ and $T'T$ have the same leading $K$ eigenvalues. Since $\|\mathscr{D}^{-1/2}\Omega\mathscr{D}^{-1/2}\|=1$, we have
	\begin{align*}
	\lambda_{1}&=\|\mathscr{L}_{\tau}\|=\|\mathscr{D}_{\tau}^{-1/2}\Omega\mathscr{D}_{\tau}^{-1/2}\|=\|\mathscr{D}^{-1/2}_{\tau}\mathscr{D}^{1/2}\mathscr{D}^{-1/2}\Omega\mathscr{D}^{-1/2}\mathscr{D}^{1/2}\mathscr{D}^{-1/2}_{\tau}\|\\
	&\leq \|\mathscr{D}^{-1/2}_{\tau}\mathscr{D}^{1/2}\|^{2}\|\mathscr{D}^{-1/2}\Omega\mathscr{D}^{-1/2}\|=\|\mathscr{D}^{-1}_{\tau}\mathscr{D}\|=\mathrm{max}_{1\leq i\leq n}\frac{\mathscr{D}(i,i)}{\tau+\mathscr{D}(i,i)}\leq 1.
	\end{align*}
	Similarly, we have $\hat{\lambda}_{1}=\|L_{\tau}\|\leq 1$.
\end{proof}

\section{Proof of consistency for Mixed-RSC}
\subsection{Proof of Lemma 3.1}
\begin{proof}
	Since
	\begin{align*}
	\|L_{\tau}-\mathscr{L}_{\tau}\|&=\|D_{\tau}^{-1/2}AD^{-1/2}_{\tau}-\mathscr{D}^{-1/2}_{\tau}\Omega\mathscr{D}^{-1/2}_{\tau}\|\\
	&=\|D_{\tau}^{-1/2}AD^{-1/2}_{\tau}-\mathscr{D}^{-1/2}_{\tau}A\mathscr{D}^{-1/2}_{\tau}+\mathscr{D}^{-1/2}_{\tau}A\mathscr{D}^{-1/2}_{\tau}-\mathscr{D}^{-1/2}_{\tau}\Omega\mathscr{D}^{-1/2}_{\tau}\|\\
	&\leq\|\mathscr{D}^{-1/2}_{\tau}A\mathscr{D}^{-1/2}_{\tau}-\mathscr{D}^{-1/2}_{\tau}\Omega\mathscr{D}^{-1/2}_{\tau}\|+\|D_{\tau}^{-1/2}AD^{-1/2}_{\tau}-\mathscr{D}^{-1/2}_{\tau}A\mathscr{D}^{-1/2}_{\tau}\|,
	\end{align*}
	we next bound the two terms of the last inequality separately.
	
	For the first term, since $\|\mathscr{D}^{-1/2}_{\tau}A\mathscr{D}^{-1/2}_{\tau}-\mathscr{D}^{-1/2}_{\tau}\Omega\mathscr{D}^{-1/2}_{\tau}\|=\|\mathscr{D}_{\tau}^{-1/2}(A-\Omega)\mathscr{D}^{-1/2}_{\tau}\|\leq \|\mathscr{D}^{-1}_{\tau}\|\|A-\Omega\|=\frac{\|A-\Omega\|}{\tau+\delta_{\mathrm{min}}}$, we only need to bound $\|A-\Omega\|$. We apply Theorem 1.4 (Bernstein inequality) in \cite{tropp2012user} to bound $\|A-\Omega\|$, and this theorem is written as below
	\begin{thm}\label{Bern}
		Consider a finite sequence $\{X_{k}\}$ of independent, random, self-adjoint matrices with dimension $d$. Assume that each random matrix satisfies
		\begin{align*}
		\mathbb{E}[X_{k}]=0, \mathrm{and~}\lambda_{\mathrm{max}}(X_{k})\leq R~\mathrm{almost~surely}.
		\end{align*}
		Then, for all $t\geq 0$,
		\begin{align*}
		\mathbb{P}(\lambda_{\mathrm{max}}(\sum_{k}X_{k})\geq t)\leq d\cdot \mathrm{exp}(\frac{-t^{2}/2}{\sigma^{2}+Rt/3}),
		\end{align*}
		where $\sigma^{2}:=\|\sum_{k}\mathbb{E}[X^{2}_{k}]\|$.
	\end{thm}
	Let $e_{i}$ be an $n\times 1$ vector, where $e_{i}(i)=1$ and 0 elsewhere, for nodes $1\leq i\leq n$. For convenience, set $W=A-\Omega$. Then we can write $W$ as $W=\sum_{i=1}^{n}\sum_{j=1}^{n}W(i,j)e_{i}e'_{j}$. Set $W^{(i,j)}$ as the $n\times n$ matrix such that $W^{(i,j)}=W(i,j)(e_{i}e'_{j}+e_{j}e_{i}')$, which gives that $W=\sum_{1\leq i< j\leq n}W^{(i,j)}$. Then we have $\mathbb{E}[W^{(i,j)}]=0$ and
	\begin{align*} \|W^{(i,j)}\|&=\|W(i,j)(e_{i}e'_{j}+e_{j}e_{i})\|=|W(i,j)|\|(e_{i}e'_{j}+e_{j}e_{i}')\|=|W(i,j)|=|A(i,j)-\Omega(i,j)|\leq1.
	\end{align*}
	Next we consider the variance parameter
	\begin{align*}
	\sigma^{2}:=\|\sum_{1\leq i<j\leq n}\mathbb{E}[(W^{(i,j)})^{2}]\|.
	\end{align*}
	We obtain the bound of $\mathbb{E}(W^{2}(i,j))$ as below
	\begin{align*} \mathbb{E}(W^{2}(i,j))&=\mathbb{E}((A(i,j)-\Omega(i,j))^{2})=\mathbb{E}((A(i,j)-\mathbb{E}(A(i,j)))^{2})=\mathrm{Var}(A(i,j))\\
	&= \Omega(i,j)(1-\Omega(i,j))\leq \Omega(i,j)=\theta(i)\theta(j)\Pi(i,:)P\Pi'(j,:)\leq \theta(i)\theta(j),
	\end{align*}
	where we have used the fact that $\Pi(i,:)P\Pi'(j,:)\leq 1$.
	Next we bound $\sigma^{2}$ as below
	\begin{align*}
	\sigma^{2}&=\|\sum_{1\leq i<j\leq n}\mathbb{E}(W^{2}(i,j))(e_{i}e_{j}'+e_{j}e_{i}')(e_{i}e_{j}'+e_{j}e_{i}')\|\\
	&=\|\sum_{1\leq i<j\leq n}\mathbb{E}[W^{2}(i,j)(e_{i}e'_{i}+e_{j}e_{j}')]\|\\
	&\leq\underset{1\leq i\leq n}{\mathrm{max}}|\sum_{j=1}^{n}\mathbb{E}(W^{2}(i,j))|\leq \underset{1\leq i\leq n}{\mathrm{max}}\sum_{j=1}^{n}\theta(i)\theta(j)\leq \theta_{\mathrm{max}}\|\theta\|_{1}.
	\end{align*}
	Set $t=\sqrt{\frac{32}{3}\theta_{\mathrm{max}}\|\theta\|_{1}\mathrm{log}(n^{\alpha}K^{-\beta})}$, according to Theorem \ref{Bern} with $\sigma^{2}\leq \theta_{\mathrm{max}}\|\theta\|_{1}, R=1, d=n$, we have
	\begin{align*}
	\mathbb{P}(\|W\|\geq t)&=\mathbb{P}(\|\sum_{1\leq i<j\leq n}W^{(i,j)}\|\geq t)\leq n \mathrm{exp}(\frac{-t^{2}/2}{\sigma^{2}+Rt/3})\leq n\mathrm{exp}(\frac{-\frac{16}{3}\mathrm{log}(n^{\alpha}K^{-\beta})}{1+\frac{1}{3}\sqrt{\frac{32\mathrm{log}(n^{\alpha}K^{-\beta})}{3\theta_{\mathrm{max}}\|\theta\|_{1}}}})\leq \frac{K^{4\beta}}{n^{4\alpha-1}},
	\end{align*}
	where we have  use assumption (A1) such that $1+\frac{1}{3}\sqrt{32\mathrm{log}(n^{\alpha}K^{-\beta})/(3\theta_{\mathrm{max}}\|\theta\|_{1})}\leq \frac{4}{3}$ for sufficiently large $n$ in the  last inequality. Hence, with probability at least $1-o(\frac{K^{4\beta}}{n^{4\alpha-1}})$, we have
	\begin{align*}
	\|\mathscr{D}^{-1/2}_{\tau}A\mathscr{D}^{-1/2}_{\tau}-\mathscr{D}^{-1/2}_{\tau}\Omega\mathscr{D}^{-1/2}_{\tau}\|\leq \frac{C\sqrt{\theta_{\mathrm{max}}\|\theta\|_{1}\mathrm{log}(n^{\alpha}K^{-\beta})}}{\tau+\delta_{\mathrm{min}}}.
	\end{align*}
	Now, set $\tilde{t}=\frac{C\sqrt{\theta_{\mathrm{max}}\|\theta\|_{1}\mathrm{log}(n^{\alpha}K^{-\beta})}}{\tau+\delta_{\mathrm{min}}}$ for convenience.
	
	For the second term $\|D_{\tau}^{-1/2}AD^{-1/2}_{\tau}-\mathscr{D}^{-1/2}_{\tau}A\mathscr{D}^{-1/2}_{\tau}\|$. Since
	\begin{align*}
	\|L_{\tau}\|&=\|D^{-1/2}_{\tau}AD^{-1/2}_{\tau}\|=\|D^{-1/2}_{\tau}D^{1/2}D^{-1/2}AD^{-1/2}D^{1/2}D^{-1/2}_{\tau}\|\\
	&\leq\|D^{-1/2}_{\tau}D^{1/2}\|\|D^{-1/2}AD^{-1/2}\|\|D^{1/2}D^{-1/2}_{\tau}\|=\|D^{-1/2}_{\tau}D^{1/2}\|\|D^{1/2}D^{-1/2}_{\tau}\|\leq 1,
	\end{align*}
	we have
	\begin{align*}
	&\|D_{\tau}^{-1/2}AD^{-1/2}_{\tau}-\mathscr{D}^{-1/2}_{\tau}A\mathscr{D}^{-1/2}_{\tau}\|\\
	&=\|D_{\tau}^{-1/2}AD^{-1/2}_{\tau}-\mathscr{D}_{\tau}^{-1/2}D_{\tau}^{1/2}L_{\tau}D_{\tau}^{1/2}\mathscr{D}_{\tau}^{-1/2}\|\\
	&=\|(I-\mathscr{D}^{-1/2}_{\tau}D^{1/2}_{\tau})L_{\tau}D_{\tau}^{1/2}\mathscr{D}_{\tau}^{-1/2}+L_{\tau}(I-D_{\tau}^{1/2}\mathscr{D}_{\tau}^{-1/2})\|\\
	&\leq\|I-\mathscr{D}^{-1/2}_{\tau}D^{1/2}_{\tau}\|\|L_{\tau}\|\|D_{\tau}^{1/2}\mathscr{D}_{\tau}^{-1/2}\|+\|L_{\tau}\|\|I-D_{\tau}^{1/2}\mathscr{D}_{\tau}^{-1/2}\|\\
	&\leq\|I-\mathscr{D}^{-1/2}_{\tau}D^{1/2}_{\tau}\|\|D_{\tau}^{1/2}\mathscr{D}_{\tau}^{-1/2}\|+\|I-D_{\tau}^{1/2}\mathscr{D}_{\tau}^{-1/2}\|\\
	&\leq\|I-\mathscr{D}^{-1/2}_{\tau}D^{1/2}_{\tau}\|\|D_{\tau}^{1/2}\mathscr{D}_{\tau}^{-1/2}-I+I\|+\|I-D_{\tau}^{1/2}\mathscr{D}_{\tau}^{-1/2}\|\\
	&\leq\|I-\mathscr{D}^{-1/2}_{\tau}D^{1/2}_{\tau}\|(\|D_{\tau}^{1/2}\mathscr{D}_{\tau}^{-1/2}-I\|+\|I\|)+\|I-D_{\tau}^{1/2}\mathscr{D}_{\tau}^{-1/2}\|\\
	&=\|I-\mathscr{D}^{-1/2}_{\tau}D^{1/2}_{\tau}\|(\|D_{\tau}^{1/2}\mathscr{D}_{\tau}^{-1/2}-I\|+1)+\|I-D_{\tau}^{1/2}\mathscr{D}_{\tau}^{-1/2}\|\\
	&=2\|I-D_{\tau}^{1/2}\mathscr{D}_{\tau}^{-1/2}\|+\|I-D_{\tau}^{1/2}\mathscr{D}_{\tau}^{-1/2}\|^{2}.
	\end{align*}
	Next we bound $\|I-D_{\tau}^{1/2}\mathscr{D}_{\tau}^{-1/2}\|$. Apply the two sided concentration inequality for each $i, 1\leq i\leq n$, (see for example \cite{chung2006complex}, chap. 2)
	\begin{align*}
	\mathbb{P}(|D(i,i)-\mathscr{D}(i,i)|\geq \varrho)&\leq \mathrm{exp}(-\frac{\varrho^{2}}{2\mathscr{D}(i,i)})+\mathrm{exp}(-\frac{\varrho^{2}}{2\mathscr{D}(i,i)+\frac{2}{3}\varrho}).
	\end{align*}
	Let $\varrho=\tilde{t}(\mathscr{D}(i,i)+\tau)$, we have
	\begin{align*}
	&\mathbb{P}(|D(i,i)-\mathscr{D}(i,i)|\geq \tilde{t}(\mathscr{D}(i,i)+\tau))\leq \mathrm{exp}(\frac{-\tilde{t}^{2}(\mathscr{D}(i,i)+\tau)^{2}}{2\mathscr{D}(i,i)})+\mathrm{exp}(\frac{-\tilde{t}^{2}(\mathscr{D}(i,i)+\tau)^{2}}{2\mathscr{D}(i,i)+\frac{2}{3}\tilde{t}(\mathscr{D}(i,i)+\tau)})\\
	&\leq 2\mathrm{exp}(-\frac{\tilde{t}^{2}(\mathscr{D}(i,i)+\tau)^{2}}{(2+\frac{2}{3}\tilde{t})(\mathscr{D}(i,i)+\tau)})=2\mathrm{exp}(-\frac{\tilde{t}^{2}(\mathscr{D}(i,i)+\tau)}{2+\frac{2}{3}\tilde{t}})\leq 2\mathrm{exp}(-\frac{\tilde{t}^{2}(\delta_{\mathrm{min}}+\tau)}{2+\frac{2}{3}\tilde{t}})\\
	&=2\mathrm{exp}(-4\mathrm{log}(n^{\alpha}K^{-\beta})\frac{1}{\frac{8(\tau+\delta_{\mathrm{min}})}{C^{2}\theta_{\mathrm{max}}\|\theta\|_{1}}+\frac{8}{3C}\sqrt{\frac{\mathrm{log}(n^{\alpha}K^{-\beta})}{\theta_{\mathrm{max}}\|\theta\|_{1}}}})\leq 2\frac{K^{4\beta}}{n^{4\alpha}},
	\end{align*}
	where we have used the facts that $\tau+\delta_{\mathrm{min}}\leq C\theta_{\mathrm{max}}\|\theta\|_{1}$ and assumption (A1) in the  last inequality (for sufficiently large $n$,  we have $\frac{8(\tau+\delta_{\mathrm{min}})}{C^{2}\theta_{\mathrm{max}}\|\theta\|_{1}}+\frac{8}{3C}\sqrt{\frac{\mathrm{log}(n^{\alpha}K^{-\beta})}{\theta_{\mathrm{max}}\|\theta\|_{1}}}\leq 1$).
	
	Since
	\begin{align*}
	\|I-D_{\tau}^{1/2}\mathscr{D}_{\tau}^{-1/2}\|=\mathrm{max}_{1\leq i\leq n}|\sqrt{\frac{D(i,i)+\tau}{\mathscr{D}(i,i)+\tau}}-1|\leq\mathrm{max}_{1\leq i\leq n}|\frac{D(i,i)+\tau}{\mathscr{D}(i,i)+\tau}-1|,
	\end{align*}
	we have
	\begin{align*}
	\mathbb{P}(\|I-D_{\tau}^{1/2}\mathscr{D}_{\tau}^{-1/2}\|\geq \tilde{t})&\leq \mathbb{P}(\mathrm{max}_{1\leq i\leq n}|\frac{D(i,i)+\tau}{\mathscr{D}(i,i)+\tau}-1|\geq \tilde{t})\\
	&\leq\mathbb{P}(\cup_{1\leq i\leq n}\{|(D(i,i)+\tau)-(\mathscr{D}(i,i)+\tau)|\geq \tilde{t}(\mathscr{D}(i,i)+\tau)\})\\
	&=\mathbb{P}(\cup_{1\leq i\leq n}\{|D(i,i)-\mathscr{D}(i,i)|\geq \tilde{t}(\mathscr{D}_{\tau}(i,i)+\tau)\})\\
	&\leq 2\frac{K^{4\beta}}{n^{4\alpha-1}}.
	\end{align*}
	Therefore, we have
	\begin{align*}
	\|D_{\tau}^{-1/2}AD^{-1/2}_{\tau}-\mathscr{D}^{-1/2}_{\tau}A\mathscr{D}^{-1/2}_{\tau}\|\leq2\|I-D_{\tau}^{1/2}\mathscr{D}_{\tau}^{-1/2}\|+\|I-D_{\tau}^{1/2}\mathscr{D}_{\tau}^{-1/2}\|^{2}\leq 2\tilde{t}+\tilde{t}^{2},
	\end{align*}
	with probability at least $1-o(\frac{K^{4\beta}}{n^{4\alpha-1}})$.
	
	Combining the two parts yields
	\begin{align*}
	&\|L_{\tau}-\mathscr{L}_{\tau}\|\leq \tilde{t}^{2}+3\tilde{t}=O(\frac{\theta_{\mathrm{max}}\|\theta\|_{1}\mathrm{log}(n^{\alpha}K^{-\beta})}{(\tau+\delta_{\mathrm{min}})^{2}})+O(\frac{\sqrt{\theta_{\mathrm{max}}\|\theta\|_{1}\mathrm{log}(n^{\alpha}K^{-\beta})}}{\tau+\delta_{\mathrm{min}}})\notag\\
	&=\begin{cases}
	O(\frac{\sqrt{\theta_{\mathrm{max}}\|\theta\|_{1}\mathrm{log}(n^{\alpha}K^{-\beta})}}{\tau+\delta_{\mathrm{min}}}), & \mbox{when } C\sqrt{\theta_{\mathrm{max}}\|\theta\|_{1}\mathrm{log}(n^{\alpha}K^{-\beta})}\leq \tau+\delta_{\mathrm{min}}\leq C\theta_{\mathrm{max}}\|\theta\|_{1}, \\
	O(\frac{\theta_{\mathrm{max}}\|\theta\|_{1}\mathrm{log}(n^{\alpha}K^{-\beta})}{(\tau+\delta_{\mathrm{min}})^{2}}), & \mbox{when~} \tau+\delta_{\mathrm{min}}<C\sqrt{\theta_{\mathrm{max}}\|\theta\|_{1}\mathrm{log}(n^{\alpha}K^{-\beta})}.
	\end{cases},
	\end{align*}
	with probability at least $1-o(\frac{K^{4\beta}}{n^{4\alpha-1}})$.
\end{proof}

\subsection{Proof of Lemma 3.2}
\begin{proof}
	To prove this lemma, we apply Theorem 4.2.1 \citep{chen2020spectral} and Lemma 5.1 \citep{lei2015consistency} where Lemma 5.1 \citep{lei2015consistency} is obtained based on the Davis-Kahan theorem \citep{yu2015a}. First, we use Theorem 4.2.1 \citep{chen2020spectral} to bound $\|\hat{V}\mathrm{sgn}(H)-V\|_{2\rightarrow\infty}$ where $\mathrm{sgn}(H)$ is defined below. Let $H=\hat{V}'V$, and $H=U_{H}\Sigma_{H}V'_{H}$ be the SVD decomposition of $H$ with $U_{H},V_{H}\in \mathbb{R}^{n\times K}$, where $U_{H}$ and $V_{H}$ represent respectively the left and right singular matrices of $H$. Define $\mathrm{sgn}(H)=U_{H}V'_{H}$. Since $\mathbb{E}(A(i,j)-\Omega(i,j))=0$, $\mathbb{E}[(L_{\tau}(i,j)-\mathscr{L}_{\tau}(i,j))^{2}]=\mathbb{E}[(\frac{A(i,j)}{\sqrt{(\tau+D(i,i))(\tau+D(j,j))}}-\frac{\Omega(i,j)}{\sqrt{(\tau+\mathscr{D}(i,i))(\tau+\mathscr{D}(j,j))}})^{2}]\leq \frac{\mathbb{E}[(A(i,j)-\Omega(i,j))^{2}]}{\mathrm{min}((\tau+1)^{2},(\tau+\delta_{\mathrm{min}})^{2})}=\frac{\mathrm{Var}(A(i,j))}{\tilde{\tau}^{2}}=\Omega(i,j)(1-\Omega(i,j))/\tilde{\tau}^{2}\leq\Omega(i,j)/\tilde{\tau}^{2}\leq \frac{\theta^{2}_{\mathrm{max}}}{\tilde{\tau}^{2}}, |L_{\tau}(i,j)-\mathscr{L}_{\tau}(i,j)|\leq \mathrm{max}(\frac{1}{\tau+1},\frac{1}{\tau+\delta_{\mathrm{min}}})=\frac{1}{\tilde{\tau}}$ where we set $\tilde{\tau}=\mathrm{min}(\tau+1,\tau+\delta_{\mathrm{min}})$, then by assumption (A1), Lemma \ref{P2} and basic algebra, we have $c_{b}=\frac{1}{\tilde{\tau} \frac{\theta_{\mathrm{max}}}{\tilde{\tau}}\sqrt{n/(\mu \mathrm{log}(n))}}=\frac{\|V\|_{2\rightarrow\infty}}{\theta_{\mathrm{max}}}\sqrt{\frac{\mathrm{log}(n)}{K}}\leq \frac{1}{\theta_{\mathrm{min}}}\sqrt{\frac{\tau+\delta_{\mathrm{max}}}{\tau+\delta_{\mathrm{min}}}}\sqrt{\frac{\mathrm{log}(n)}{K\lambda_{K}(\Pi'\Pi)}}=C\frac{1}{\theta_{\mathrm{min}}}\sqrt{\frac{\tau+\delta_{\mathrm{max}}}{\tau+\delta_{\mathrm{min}}}}\sqrt{\frac{\mathrm{log}(n)}{n}}=C\frac{\theta_{\mathrm{max}}}{\theta_{\mathrm{min}}}\sqrt{\frac{\tau+\delta_{\mathrm{max}}}{\tau+\delta_{\mathrm{min}}}}\sqrt{\frac{\mathrm{log}(n)}{\theta^{2}_{\mathrm{max}}n}}\leq C\frac{\theta_{\mathrm{max}}}{\theta_{\mathrm{min}}}\sqrt{\frac{\mathrm{log}(n)}{\theta_{\mathrm{max}}\|\theta_{1}\|}}\leq O(1)$ where $\mu=\frac{n\|V\|^{2}_{2\rightarrow\infty}}{K}$. Meanwhile, since we can simply set $O((\tau+\delta_{\mathrm{min}})/\tilde{\tau})=O(1)$,the requirement $|\lambda_{K}|\geq C\frac{\theta_{\mathrm{max}}}{\tilde{\tau}}\sqrt{n\mathrm{log}(n)}$ in Theorem 4.2.1. \cite{chen2020spectral}  reads $|\lambda_{K}|\geq C\frac{\theta_{\mathrm{max}}\sqrt{n\mathrm{log}(n)}}{\tau+\delta_{\mathrm{min}}}$. Now, Theorem 4.2.1. \cite{chen2020spectral} gives that with high probability,
	\begin{align*}
	\|\hat{V}\mathrm{sgn}(H)-V\|_{2\rightarrow\infty}\leq \frac{\kappa(\mathscr{L}_{\tau})\theta_{\mathrm{max}}\sqrt{K\mu}+\theta_{\mathrm{max}}\sqrt{K\mathrm{log}(n)}}{|\lambda_{K}|(\tau+\delta_{\mathrm{min}})}.
	\end{align*}
	By Lemma \ref{P2},  $\mu\leq\frac{(\tau+\delta_{\mathrm{max}})\theta^{2}_{\mathrm{max}}n}{(\tau+\delta_{\mathrm{min}})\theta^{2}_{\mathrm{min}}K\lambda_{K}(\Pi'\Pi)}$. Since $\tilde{\theta}^{2}_{\mathrm{min}}\geq \frac{\theta^{2}_{\mathrm{min}}}{\tau+\delta_{\mathrm{max}}}$, by Lemma \ref{P4}, we have $|\lambda_{K}|\geq\frac{\theta^{2}_{\mathrm{min}}}{\tau+\delta_{\mathrm{max}}}|\lambda_{K}(P)|\lambda_{K}(\Pi'\Pi)$. Then we have
	\begin{align*}	&\|\hat{V}\mathrm{sgn}(H)-V\|_{2\rightarrow\infty}\leq \frac{\kappa(\mathscr{L}_{\tau})\theta_{\mathrm{max}}\sqrt{K\mu}+\theta_{\mathrm{max}}\sqrt{K\mathrm{log}(n)}}{|\lambda_{K}|(\tau+\delta_{\mathrm{min}})}\\
	&\leq C\frac{(\tau+\delta_{\mathrm{max}})\theta_{\mathrm{max}}}{(\tau+\delta_{\mathrm{min}})\theta^{2}_{\mathrm{min}}|\lambda_{K}(P)|\lambda_{K}(\Pi'\Pi)}(\frac{\kappa(\mathscr{L}_{\tau})\theta_{\mathrm{max}}}{\theta_{\mathrm{min}}}\sqrt{\frac{(\tau+\delta_{\mathrm{max}})n}{(\tau+\delta_{\mathrm{min}})\lambda_{K}(\Pi'\Pi)}}+\sqrt{K\mathrm{log}(n)})\\
	&=O(\frac{(\tau+\delta_{\mathrm{max}})\theta_{\mathrm{max}}}{(\tau+\delta_{\mathrm{min}})\theta^{2}_{\mathrm{min}}|\lambda_{K}(P)|\lambda_{K}(\Pi'\Pi)}\mathrm{max}(\frac{\kappa(\mathscr{L}_{\tau})\theta_{\mathrm{max}}}{\theta_{\mathrm{min}}}\sqrt{\frac{(\tau+\delta_{\mathrm{max}})n}{(\tau+\delta_{\mathrm{min}})\lambda_{K}(\Pi'\Pi)}},\sqrt{K\mathrm{log}(n)}))\\
	&=O(\frac{(\tau+\delta_{\mathrm{max}})\theta_{\mathrm{max}}\sqrt{K\mathrm{log}(n)}}{(\tau+\delta_{\mathrm{min}})\theta^{2}_{\mathrm{min}}|\lambda_{K}(P)|\lambda_{K}(\Pi'\Pi)}).
	\end{align*}	
	Second, we apply the principal subspace perturbation introduced in Lemma 5.1 \citep{lei2015consistency} to bound $\|V-\hat{V}\mathrm{sgn}(H)\|_{F}$. We write this lemma as below
	\begin{lem}\label{PSP}
		(Principal subspace perturbation \citep{lei2015consistency}). Assume that $X\in\mathbb{R}^{n\times n}$ is a rank $K$ symmetric matrix with smallest nonzero singular value $\sigma_{K}(X)$. Let $\hat{X}$ be any symmetric matrix and $\hat{U},U\in\mathbb{R}^{n\times K}$ be the $K$ leading eigenvectors of $\hat{X}$ and $X$, respectively. Then there exists a $K\times K$ orthogonal matrix $\hat{O}$ such that
		\begin{align*}
		\|U-\hat{U}\hat{O}\|_{F}\leq \frac{2\sqrt{2K}\|\hat{X}-X||}{\sigma_{K}(X)}.
		\end{align*}
	\end{lem}
	Let $\hat{X}=L_{\tau}, X=\mathscr{L}_{\tau}, U=V, \hat{U}=\hat{V}, \sigma_{K}(X)=|\lambda_{K}|$, by Lemma \ref{PSP}, there exists a $K\times K$ orthogonal matrix $\hat{O}$ such that
	\begin{align*}
	\|V-\hat{V}\hat{O}\|_{F}\leq \frac{2\sqrt{2K}\|L_{\tau}-\mathscr{L}_{\tau}||}{|\lambda_{K}|}.
	\end{align*}
	By the proof of Theorem 2 \citep{yu2015a}, we know that $\hat{O}=\mathrm{sgn}(H)$, combine it with Lemma \ref{P4} and 3.1, we have
	\begin{align*}
	\|V-\hat{V}\hat{O}\|_{F}\leq C\frac{\sqrt{K}err_{n}}{\tilde{\theta}^{2}_{\mathrm{min}}|\lambda_{K}(P)|\lambda_{K}(\Pi'\Pi)}.
	\end{align*}
	Now we are ready to bound $\|\hat{V}\hat{V}'-VV'\|_{2\rightarrow\infty}$. Since
	\begin{align*}
	&\|\hat{V}\hat{V}'-VV'\|_{2\rightarrow\infty}=\mathrm{max}_{1\leq i\leq n}\|e'_{i}(VV'-\hat{V}\hat{V}')\|_{F}\\
	&=\mathrm{max}_{1\leq i\leq n}\|e'_{i}(VV'-\hat{V}\mathrm{sgn}(H)V'+\hat{V}\mathrm{sgn}(H)V'-\hat{V}\hat{V}')\|_{F}\\
	&\overset{\mathrm{By~Lemma~A.1~Yu~et~al.~(2015)}}{\leq}\mathrm{max}_{1\leq i\leq n}\|e'_{i}(V-\hat{V}\mathrm{sgn}(H))\|_{F}+\mathrm{max}_{1\leq i\leq n}\|e'_{i}\hat{V}(\mathrm{sgn}(H)V'-\hat{V}')\|_{F}\\	&=\|V-\hat{V}\mathrm{sgn}(H)\|_{2\rightarrow\infty}+\mathrm{max}_{1\leq i\leq n}\|e'_{i}\hat{V}(\mathrm{sgn}(H)V'-\hat{V}')\|_{F}\\ &\leq\|V-\hat{V}\mathrm{sgn}(H)\|_{2\rightarrow\infty}+\mathrm{max}_{1\leq i\leq n}\|e'_{i}\hat{V}\|_{F}\|\mathrm{sgn}(H)V'-\hat{V}'\|_{F}\\ &=\|V-\hat{V}\mathrm{sgn}(H)\|_{2\rightarrow\infty}+\mathrm{max}_{1\leq i\leq n}\|e'_{i}\hat{V}\|_{F}\|V-\hat{V}\mathrm{sgn}(H)\|_{F}\\ &=\|V-\hat{V}\mathrm{sgn}(H)\|_{2\rightarrow\infty}+\mathrm{max}_{1\leq i\leq n}\|e'_{i}(\hat{V}\mathrm{sgn}(H)-V+V)\|_{F}\|V-\hat{V}\mathrm{sgn}(H)\|_{F}\\ &\leq\|V-\hat{V}\mathrm{sgn}(H)\|_{2\rightarrow\infty}+(\|\hat{V}\mathrm{sgn}(H)-V\|_{2\rightarrow\infty}+\|V\|_{2\rightarrow\infty})\|V-\hat{V}\mathrm{sgn}(H)\|_{F}\\ &\overset{\mathrm{By~Lemma~}\ref{P2}}{\leq}\|V-\hat{V}\mathrm{sgn}(H)\|_{2\rightarrow\infty}+(\|\hat{V}\mathrm{sgn}(H)-V\|_{2\rightarrow\infty}+\frac{\tilde{\theta}_{\mathrm{max}}}{\tilde{\theta}_{\mathrm{min}}\sqrt{\lambda_{K}(\Pi'\Pi)}})\|V-\hat{V}\mathrm{sgn}(H)\|_{F}\\ &=O(\frac{(\tau+\delta_{\mathrm{max}})\theta_{\mathrm{max}}\sqrt{K\mathrm{log}(n)}}{(\tau+\delta_{\mathrm{min}})\theta^{2}_{\mathrm{min}}|\lambda_{K}(P)|\lambda_{K}(\Pi'\Pi)}).
	\end{align*}
\end{proof}

\subsection{Proof of Lemma 3.3}
\begin{proof}
	By Lemma 2.6, we see that $V_{*,1}(\mathcal{I},:)$ satisfies condition 1 in \cite{MaoSVM}.
	Meanwhile, since $(V_{*,1}(\mathcal{I},:)V'_{*,1}(\mathcal{I},:))^{-1}\mathbf{1}>0$, we have $(V_{*,1}(\mathcal{I},:)V'_{*,1}(\mathcal{I},:))^{-1}\mathbf{1}\geq \eta\mathbf{1}$, hence $V_{*,1}(\mathcal{I},:)$ satisfies condition 2 in \cite{MaoSVM}.
	
	By Lemma 2.9, we have $V_{*,2}(\mathcal{I},:)V'_{*,2}(\mathcal{I},:)=V_{*,1}(\mathcal{I},:)V'_{*,1}(\mathcal{I},:)$, hence $V_{*,2}(\mathcal{I},:)$ also satisfies conditions 1 and 2 in \cite{MaoSVM}.  The above analysis shows that we can directly apply Lemma F.1 of \cite{MaoSVM} since the Ideal Mixed-RSC satisfies conditions 1 and 2 in \cite{MaoSVM}. Let $\hat{V}_{*,2}$ and $K$ be the inputs of SVM-cone algorithm, there exists a permutation matrix $\mathcal{P}\in\mathbb{R}^{K\times K}$ such that
	\begin{align*}
	\|\hat{V}_{*,2}(\mathcal{\hat{I}},:)-\mathcal{P}V_{*,2}(\mathcal{I},:)\|_{F}= O(\frac{K\zeta\epsilon}{\lambda^{1.5}_{K}(V_{*,2}(\mathcal{I},:))V'_{*,2}(\mathcal{I},:)}),
	\end{align*}
	where $\zeta\leq \frac{4K}{\eta\lambda^{1.5}_{K}(V_{*,2}(\mathcal{I},:)V'_{*,2}(\mathcal{I},:))}=O(\frac{K}{\eta\lambda^{1.5}_{K}(V_{*,1}(\mathcal{I},:)V'_{*,1}(\mathcal{I},:))})$, and $\epsilon=\mathrm{max}_{1\leq i\leq n}\|\hat{V}_{*,2}(i,:)-(V_{*,2}(i,:)\|_{F}$. Next we bound $\epsilon$. Since
	\begin{align*}
	&\|\hat{V}_{*,2}(i,:)-V_{*,2}(i,:)\|_{F}\leq\|\frac{\hat{V}_{2}(i,:)\|V_{2}(i,:)\|_{F}-V_{2}(i,:)\|\hat{V}_{2}(i,:)\|_{F}}{\|\hat{V}_{2}(i,:)\|_{F}\|V_{2}(i,:)\|_{F}}\|_{F}\leq\frac{2\|\hat{V}_{2}(i,:)-V_{2}(i,:)\|_{F}}{\|V_{2}(i,:)\|_{F}}\\
	&\leq \frac{2\|\hat{V}_{2}-V_{2}\|_{2\rightarrow\infty}}{\|V_{2}(i,:)\|_{F}}\leq\frac{2\varpi_{1}}{\|V_{2}(i,:)\|_{F}}=\frac{2\varpi_{1}}{\|(VV')(i,:)\|_{F}}=\frac{2\varpi_{1}}{\|V(i,:)V'\|_{F}}=\frac{2\varpi_{1}}{\|V(i,:)\|_{F}}\\
	&\leq \frac{2\tilde{\theta}_{\mathrm{max}}\varpi_{1}\sqrt{K\lambda_{1}(\Pi'\Pi)}}{\tilde{\theta}_{\mathrm{min}}},
	\end{align*}
	where the last inequality holds by Lemma \ref{P2}, we have $\epsilon=(\frac{\tilde{\theta}_{\mathrm{max}}\varpi_{1}\sqrt{K\lambda_{1}(\Pi'\Pi)}}{\tilde{\theta}_{\mathrm{min}}})$. Finally, by Lemma \ref{P3}, we have
	\begin{align*}
	\|\hat{V}_{*,2}(\mathcal{\hat{I}},:)-\mathcal{P}V_{*,2}(\mathcal{I},:)\|_{F}=O(\frac{\tilde{\theta}^{7}_{\mathrm{max}}K^{2.5}\varpi_{1}\kappa^{3}(\Pi'\Pi)\sqrt{\lambda_{1}(\Pi'\Pi)}}{\eta \tilde{\theta}^{7}_{\mathrm{min}}}).
	\end{align*}
\end{proof}

\subsection{Proof of Lemma 3.4}
\begin{proof}
	For convenience, we set $V_{*,1}(\mathcal{I},:)=V_{C}, \hat{V}_{*,1}(\mathcal{\hat{I}},:)=\hat{V}_{C}$.Then we have
	\begin{align*}
	&\|e'_{i}(\hat{Y}_{\bullet}-Y_{\bullet}\mathcal{P})\|_{F}=\|e'_{i}(\mathrm{max}(\hat{V}\hat{V}'_{C}(\hat{V}_{C}\hat{V}'_{C})^{-1},0)-VV'_{C}(V_{C}V'_{C})^{-1}\mathcal{P})\|_{F}\\
	&\leq\|e'_{i}(\hat{V}\hat{V}'_{C}(\hat{V}_{C}\hat{V}'_{C})^{-1}-VV'_{C}(V_{C}V'_{C})^{-1}\mathcal{P})\|_{F}\\
	&=\|e'_{i}(\hat{V}-V(V'\hat{V}))\hat{V}'_{C}(\hat{V}_{C}\hat{V}'_{C})^{-1}+e'_{i}(V(V'\hat{V})\hat{V}'_{C}(\hat{V}_{C}\hat{V}'_{C})^{-1}\\
	&~~~-V(V'\hat{V})(\mathcal{P}'(V_{C}V'_{C})(V'_{C})^{-1}(V'\hat{V}))^{-1})\|_{F}\\
	&\leq\|e'_{i}(\hat{V}-V(V'\hat{V}))\hat{V}'_{C}(\hat{V}_{C}\hat{V}'_{C})^{-1}\|_{F}+\|e'_{i}V(V'\hat{V})(\hat{V}'_{C}(\hat{V}_{C}\hat{V}'_{C})^{-1}\\
	&~~~-(\mathcal{P}'(V_{C}V'_{C})(V'_{C})^{-1}(V'\hat{V}))^{-1})\|_{F}\\
	&\leq \|e'_{i}(\hat{V}-V(V'\hat{V}))\|_{F}\|\hat{V}'_{C}\|_{F}\|(\hat{V}_{C}\hat{V}'_{C})^{-1}\|_{F}+\|e'_{i}V(V'\hat{V})(\hat{V}'_{C}(\hat{V}_{C}\hat{V}'_{C})^{-1}\\
	&~~~-(\mathcal{P}'(V_{C}V'_{C})(V'_{C})^{-1}(V'\hat{V}))^{-1})\|_{F}\\
	&= \sqrt{K}\|e'_{i}(\hat{V}-V(V'\hat{V}))\|_{F}\|(\hat{V}_{C}\hat{V}'_{C})^{-1}\|_{F}+\|e'_{i}V(V'\hat{V})(\hat{V}^{-1}_{C}-(\mathcal{P}'V_{C}(V'\hat{V}))^{-1})\|_{F}\\
	&\leq K\|e'_{i}(\hat{V}-V(V'\hat{V}))\|_{F}/\lambda_{K}(\hat{V}_{C}\hat{V}'_{C})+\|e'_{i}V(V'\hat{V})(\hat{V}^{-1}_{C}-(\mathcal{P}'V_{C}(V'\hat{V}))^{-1})\|_{F}\\
	&=
	K\|e'_{i}(\hat{V}\hat{V}'-VV')\hat{V}\|_{F}O(\frac{\tilde{\theta}^{2}_{\mathrm{max}}\kappa(\Pi'\Pi)}{\tilde{\theta}^{2}_{\mathrm{min}}})+\|e'_{i}V(V'\hat{V})(\hat{V}^{-1}_{C}-(\mathcal{P}'V_{C}(V'\hat{V}))^{-1})\|_{F}\\
	&\mathrm{By~Lemma~A.1~in~} Yu~et~al.~(2015)\mathrm{~or~Remark~3.2~in~Mao~et~al.~(2020)}\\
	&\leq K\|e'_{i}(\hat{V}\hat{V}'-VV')\|_{F}O(\frac{\tilde{\theta}^{2}_{\mathrm{max}}\kappa(\Pi'\Pi)}{\tilde{\theta}^{2}_{\mathrm{min}}})+\|e'_{i}V(V'\hat{V})(\hat{V}^{-1}_{C}-(\mathcal{P}'V_{C}(V'\hat{V}))^{-1})\|_{F}\\
	&\leq O(\frac{\tilde{\theta}^{2}_{\mathrm{max}}K\varpi_{1}\kappa(\Pi'\Pi)}{\tilde{\theta}^{2}_{\mathrm{min}}})+\|e'_{i}V(V'\hat{V})(\hat{V}^{-1}_{C}-(\mathcal{P}'V_{C}(V'\hat{V}))^{-1})\|_{F},
	\end{align*}
	where we have used similar idea in the proof of Lemma G.3 \cite{mao2020estimating} such that apply $O(\frac{1}{\lambda_{K}(V_{C}V'_{C})})$ to estimate $\frac{1}{\lambda_{K}(\hat{V}_{C}\hat{V}'_{C})}$, then by Lemma \ref{P3}, we have $\frac{1}{\lambda_{K}(\hat{V}_{C}\hat{V}'_{C})}\leq O(\frac{\tilde{\theta}^{2}_{\mathrm{max}}\kappa(\Pi'\Pi)}{\tilde{\theta}^{2}_{\mathrm{min}}})$.
	
	Now we aim to bound $\|e'_{i}V(V'\hat{V})(\hat{V}^{-1}_{C}-(\mathcal{P}'V_{C}(V'\hat{V}))^{-1})\|_{F}$. For convenience, set $T=V'\hat{V}, S=\mathcal{P}'V_{C}T$. We have
	\begin{align*}
	&\|e'_{i}V(V'\hat{V})(\hat{V}^{-1}_{C}-(\mathcal{P}'V_{C}(V'\hat{V}))^{-1})\|_{F}=\|e'_{i}VTS^{-1}(S-\hat{V}_{C})\hat{V}^{-1}_{C}\|_{F}\\
	&\leq\|e'_{i}VTS^{-1}(S-\hat{V}_{C})\|_{F}\|\hat{V}^{-1}_{C}\|_{F}\leq\|e'_{i}VTS^{-1}(S-\hat{V}_{C})\|_{F}\frac{\sqrt{K}}{|\lambda_{K}(\hat{V}_{C})|}\\
	&=\|e'_{i}VTS^{-1}(S-\hat{V}_{C})\|_{F}\frac{\sqrt{K}}{|\sqrt{\lambda_{K}(\hat{V}_{C}\hat{V}'_{C})}|}\leq \|e'_{i}VTS^{-1}(S-\hat{V}_{C})\|_{F}O(\frac{\tilde{\theta}_{\mathrm{max}}\sqrt{K\kappa(\Pi'\Pi)}}{\tilde{\theta}_{\mathrm{min}}})\\
	&=\|e'_{i}VTT^{-1}V'_{C}(V_{C}V'_{C})^{-1}\mathcal{P}(S-\hat{V}_{C})\|_{F}O(\frac{\tilde{\theta}_{\mathrm{max}}\sqrt{K\kappa(\Pi'\Pi)}}{\tilde{\theta}_{\mathrm{min}}})\\
	&=\|e'_{i}VV'_{C}(V_{C}V'_{C})^{-1}\mathcal{P}(S-\hat{V}_{C})\|_{F}O(\frac{\tilde{\theta}_{\mathrm{max}}\sqrt{K\kappa(\Pi'\Pi)}}{\tilde{\theta}_{\mathrm{min}}})\\
	&=\|e'_{i}Y_{\bullet}\mathcal{P}(S-\hat{V}_{C})\|_{F}O(\frac{\tilde{\theta}_{\mathrm{max}}\sqrt{K\kappa(\Pi'\Pi)}}{\tilde{\theta}_{\mathrm{min}}})\\
	&\leq\|e'_{i}Y_{\bullet}\|_{F}\|S-\hat{V}_{C}\|_{F}O(\frac{\tilde{\theta}_{\mathrm{max}}\sqrt{K\kappa(\Pi'\Pi)}}{\tilde{\theta}_{\mathrm{min}}})\\
	&\overset{\mathrm{By~the~proof~of~Lemma~}3.5}{\leq}\frac{\tilde{\theta}^{2}_{\mathrm{max}}\sqrt{K\kappa(\Pi'\Pi)}}{\tilde{\theta}^{2}_{\mathrm{min}}\sqrt{\lambda_{K}(\Pi'\Pi)}}\|S-\hat{V}_{C}\|_{F}O(\frac{\tilde{\theta}_{\mathrm{max}}\sqrt{K\kappa(\Pi'\Pi)}}{\tilde{\theta}_{\mathrm{min}}})\\
	&=\|\hat{V}_{C}-\mathcal{P}'V_{C}V'\hat{V}\|_{F}\frac{\tilde{\theta}^{3}_{\mathrm{max}}K\kappa(\Pi'\Pi)}{\tilde{\theta}^{3}_{\mathrm{min}}\sqrt{\lambda_{K}(\Pi'\Pi)}}=\|(\hat{V}_{C}\hat{V}'-\mathcal{P}'V_{C}V')\hat{V}\|_{F}\frac{\tilde{\theta}^{3}_{\mathrm{max}}K\kappa(\Pi'\Pi)}{\tilde{\theta}^{3}_{\mathrm{min}}\sqrt{\lambda_{K}(\Pi'\Pi)}}\\
	&\leq\|\hat{V}_{C}\hat{V}'-\mathcal{P}'V_{C}V'\|_{F}\frac{\tilde{\theta}^{3}_{\mathrm{max}}K\kappa(\Pi'\Pi)}{\tilde{\theta}^{3}_{\mathrm{min}}\sqrt{\lambda_{K}(\Pi'\Pi)}}\\
	&\overset{\mathrm{By~Lemma~}2.9}{=}\|\hat{V}_{2C}-\mathcal{P}'V_{2C}\|_{F}\frac{\tilde{\theta}^{3}_{\mathrm{max}}K\kappa(\Pi'\Pi)}{\tilde{\theta}^{3}_{\mathrm{min}}\sqrt{\lambda_{K}(\Pi'\Pi)}}\\
	&\leq(\|\hat{V}_{2C}-\mathcal{P}V_{2C}\|_{F}+\|(\mathcal{P}-\mathcal{P}')V_{2C}\|_{F})\frac{\tilde{\theta}^{3}_{\mathrm{max}}K\kappa(\Pi'\Pi)}{\tilde{\theta}^{3}_{\mathrm{min}}\sqrt{\lambda_{K}(\Pi'\Pi)}}\\
	&\overset{\mathrm{By~Lemma~}3.3}{=}(O(\frac{\tilde{\theta}^{7}_{\mathrm{max}}K^{2.5}\varpi_{1}\kappa^{3}(\Pi'\Pi)\sqrt{\lambda_{1}(\Pi'\Pi)}}{\eta \tilde{\theta}^{7}_{\mathrm{min}}})+\|(\mathcal{P}-\mathcal{P}')V_{2C}\|_{F})\frac{\tilde{\theta}^{3}_{\mathrm{max}}K\kappa(\Pi'\Pi)}{\tilde{\theta}^{3}_{\mathrm{min}}\sqrt{\lambda_{K}(\Pi'\Pi)}}\\
	&\leq(O(\frac{\tilde{\theta}^{7}_{\mathrm{max}}K^{2.5}\varpi_{1}\kappa^{3}(\Pi'\Pi)\sqrt{\lambda_{1}(\Pi'\Pi)}}{\eta \tilde{\theta}^{7}_{\mathrm{min}}})+K\sqrt{2})\frac{\tilde{\theta}^{3}_{\mathrm{max}}K\kappa(\Pi'\Pi)}{\tilde{\theta}^{3}_{\mathrm{min}}\sqrt{\lambda_{K}(\Pi'\Pi)}}\\
	&=O(\frac{\tilde{\theta}^{10}_{\mathrm{max}}K^{3.5}\varpi_{1}\kappa^{4.5}(\Pi'\Pi)}{\eta \tilde{\theta}^{10}_{\mathrm{min}}}).
	\end{align*}
	Then, we have
	\begin{align*}
	&\|e'_{i}(\hat{Y}_{\bullet}-Y_{\bullet}\mathcal{P})\|_{F}\leq O(\frac{\tilde{\theta}^{2}_{\mathrm{max}}K\varpi_{1}\kappa(\Pi'\Pi)}{\tilde{\theta}^{2}_{\mathrm{min}}})+\|e'_{i}V(V'\hat{V})(\hat{V}^{-1}_{C}-(\mathcal{P}'V_{C}(V'\hat{V}))^{-1})\|_{F}\\
	&=O(\frac{\tilde{\theta}^{2}_{\mathrm{max}}K\varpi_{1}\kappa(\Pi'\Pi)}{\tilde{\theta}^{2}_{\mathrm{min}}})+O(\frac{\tilde{\theta}^{10}_{\mathrm{max}}K^{3.5}\varpi_{1}\kappa^{4.5}(\Pi'\Pi)}{\eta \tilde{\theta}^{10}_{\mathrm{min}}})\\
	&=O(\frac{\tilde{\theta}^{10}_{\mathrm{max}}K^{3.5}\varpi_{1}\kappa^{4.5}(\Pi'\Pi)}{\eta \tilde{\theta}^{10}_{\mathrm{min}}}).
	\end{align*}
\end{proof}

\subsection{Proof of Lemma 3.5}
\begin{proof}
	For convenience, denote $\varpi_{2}=\|\hat{V}_{2,*}(\mathcal{\hat{I}},:)-\mathcal{P}V_{2,*}(\mathcal{I},:)\|_{F}$.
	We begin the proof by providing bounds for several items used in our proof.
	\begin{itemize}
		\item For $1\leq i\leq n$, by Lemma \ref{P2}, we have $N(i,i)=\frac{1}{\|V(i,:)\|_{F}}\leq\frac{\tilde{\theta}_{\mathrm{max}}\sqrt{K\lambda_{1}(\Pi'\Pi)}}{\tilde{\theta}_{\mathrm{min}}}$ and $N(i,i)\geq\frac{\tilde{\theta}_{\mathrm{min}}\sqrt{\lambda_{K}(\Pi'\Pi)}}{\tilde{\theta}_{\mathrm{max}}}$.
		\item  Since $\mathscr{L}_{\tau}(\mathcal{I},\mathcal{I})=\tilde{\Theta}(\mathcal{I},\mathcal{I})P\tilde{\Theta}(\mathcal{I},\mathcal{I})=V(\mathcal{I},:)EV'(\mathcal{I},:)$ and $P$ has unit diagonal entries, for $1\leq k\leq K$,  we have $\sqrt{(\mathrm{diag}(V(\mathcal{I},:)EV'(\mathcal{I},:)))(k,k)}=(\tilde{\Theta}(\mathcal{I},\mathcal{I}))(k,k)\geq \tilde{\theta}_{\mathrm{min}}$.  For $1\leq k\leq K$, let $J_{k}$ be the $k$-th diagonal entry of $J$, we have
		\begin{align*}
		J_{k}&=(N(\mathcal{I},\mathcal{I}))(k,k)\sqrt{(\mathrm{diag}(V(\mathcal{I},:)EV'(\mathcal{I},:)))(k,k)}\geq (N(\mathcal{I},\mathcal{I}))(k,k)\tilde{\theta}_{\mathrm{min}}\\
		&\geq\frac{\tilde{\theta}^{2}_{\mathrm{min}}\sqrt{\lambda_{K}(\Pi'\Pi)}}{\tilde{\theta}_{\mathrm{max}}},
		\end{align*}
		and $J_{k}\leq\frac{\tilde{\theta}^{2}_{\mathrm{max}}\sqrt{K\lambda_{1}(\Pi'\Pi)}}{\tilde{\theta}_{\mathrm{min}}}.$ Meanwhile, we also have $\|J\|_{F}\leq\frac{\tilde{\theta}^{2}_{\mathrm{max}}K\sqrt{\lambda_{1}(\Pi'\Pi)}}{\tilde{\theta}_{\mathrm{min}}}$.
		\item For $1\leq i\leq n$, since $Y_{\bullet}=VV'_{*}(\mathcal{I},:)(V_{*}(\mathcal{I},:)V'_{*}(\mathcal{I},:))^{-1}=Y_{\bullet}V^{-1}_{*}(\mathcal{I},:)$, we have
		\begin{align*}
		&\|e'_{i}Y_{\bullet}\|_{F}=\|V(i,:)V^{-1}_{*}(\mathcal{I},:)\|_{F}\leq\|V(i,:)\|_{F}\|V^{-1}_{*}(\mathcal{I},:)\|_{F}\leq \|V(i,:)\|_{F}\frac{\sqrt{K}}{|\lambda_{K}(V_{*}(\mathcal{I},:))|}\\
		&= \|V(i,:)\|_{F}\frac{\sqrt{K}}{|\lambda^{0.5}_{K}(V_{*}(\mathcal{I},:)V'_{*}(\mathcal{I},:))|}\leq\frac{\tilde{\theta}^{2}_{\mathrm{max}}\sqrt{K\kappa(\Pi'\Pi)}}{\tilde{\theta}^{2}_{\mathrm{min}}\sqrt{\lambda_{K}(\Pi'\Pi)}}.
		\end{align*}
	\end{itemize}
	For $1\leq k\leq K$, let $\hat{J}_{k}$ be the $k$-th diagonal entries of $\hat{J}$. Since in Lemma 3.3, we consider permutation matrix $\mathcal{P}$, let $p(k)$ be the index of the $k$-th row of $\hat{V}_{*,2}(\hat{\mathcal{I}},:)$ after considering permutation matrix $\mathcal{P}$. Since $V_{*,1}=V_{*,2}V, \hat{V}_{*,1}=\hat{V}_{*,2}\hat{V}$, we have $V_{*,1}(\mathcal{I},:)=V_{*,2}(\mathcal{I},:)V, \hat{V}_{*,1}(\mathcal{\hat{I}},:)=\hat{V}_{*,2}(\mathcal{\hat{I}},:)\hat{V}$, which gives that $$J=\sqrt{\mathrm{diag}(V_{*,1}(\mathcal{I},:)EV'_{*,1}(\mathcal{I},:))}=\sqrt{\mathrm{diag}(V_{*,2}(\mathcal{I},:)VEV'V'_{*,2}(\mathcal{I},:))},$$ $$\hat{J}=\sqrt{\mathrm{diag}(\hat{V}_{*,1}(\mathcal{\hat{I}},:)\hat{E}\hat{V}'_{*,1}(\mathcal{\hat{I}},:))}=\sqrt{\mathrm{diag}(\hat{V}_{*,2}(\mathcal{\hat{I}},:)\hat{V}\hat{E}\hat{V}'\hat{V}'_{*,2}(\mathcal{\hat{I}},:))}.$$  Again, for convenience, set $\hat{V}_{2C}=\hat{V}_{*,2}(\mathcal{\hat{I}},:), V_{2C}=V_{*,2}(\mathcal{I},:)$. Since $\|V\|=1, \|\hat{V}\|=1, \|E\|=\|\mathscr{L}_{\tau}-L_{\tau}+L_{\tau}\|\leq err_{n}+1=O(1), \|\hat{E}\|=\|\hat{E}-E+E\|\leq O(1)$ by Weyl's inequality, and $\|e'_{k}\mathcal{P}'\hat{V}_{2C}\|=\|\mathcal{P}'\hat{V}_{2C}e_{k}\|= \|\hat{V}_{2C}e_{k}\|\leq \|e'_{k}\hat{V}_{2C}\|_{F}=1$,  we have
	\begin{align*}
	&|J^{2}_{k}-\hat{J}^{2}_{p(k)}|=\|e'_{k}V_{2C}VEV'V'_{2C}e_{k}-e'_{k}\mathcal{P}'\hat{V}_{2C}\hat{V}\hat{E}\hat{V}'\hat{V}'_{2C}\mathcal{P}e_{k}\|\\
	&\leq \|e'_{k}(V_{2C}-\mathcal{P}'\hat{V}_{2C})VEV'V'_{2C}e_{k}\|+\|e'_{k}\mathcal{P}'\hat{V}_{2C}(VEV'-\hat{V}\hat{E}\hat{V}')V'_{2C}e_{k}\|\\
	&~~~+\|e'_{k}\mathcal{P}'\hat{V}_{2C}\hat{V}\hat{E}\hat{V}'(V'_{2C}-\hat{V}'_{2C}\mathcal{P})e_{k}\|\\
	&\leq \|e'_{k}(V_{2C}-\mathcal{P}'\hat{V}_{2C})\|\|V\|\|E\|\|V'\|\|V'_{2C}e_{k}\|+\|e'_{k}\mathcal{P}'\hat{V}_{2C}\|\|VEV'-\hat{V}\hat{E}\hat{V}'\|\|V'_{2C}e_{k}\|\\
	&~~~+\|e'_{k}\mathcal{P}'\hat{V}_{2C}\|\|\hat{V}\|\|\hat{E}\|\|\hat{V}'\|\|(V'_{2C}-\hat{V}'_{2C}\mathcal{P})e_{k}\|\\
	&= \|e'_{k}(V_{2C}-\mathcal{P}'\hat{V}_{2C})\|\|E\|+\|VEV'-\hat{V}\hat{E}\hat{V}'\|+\|(V'_{2C}-\hat{V}'_{2C}\mathcal{P})e_{k}\|\|\hat{E}\|\\
	&\leq \|e'_{k}(V_{2C}-\mathcal{P}'\hat{V}_{2C})\|O(1)+\|VEV'-\hat{V}\hat{E}\hat{V}'\|+\|(V'_{2C}-\hat{V}'_{2C}\mathcal{P})e_{k}\|O(1)\\
	&= \|(\mathcal{P}V_{2C}-\hat{V}_{2C})e_{k}\|_{F}O(1)+\|VEV'-\hat{V}\hat{E}\hat{V}'\|+\|e'_{k}(V_{2C}-\mathcal{P}'\hat{V}_{2C})\|O(1)\\
	&\leq 2\varpi_{2}O(1)+\|VEV'-\hat{V}\hat{E}\hat{V}'\|\leq O(\varpi_{2})+O(err_{n}),
	\end{align*}
	where the last inequality holds by below analysis: from the properties of the SVD, we know that $\hat{V}\hat{E}\hat{V}'$ is the best rank $K$ approximation to $L_{\tau}$ in spectral norm, therefore $\|\hat{V}\hat{E}\hat{V}'-L_{\tau}\|\leq \|\mathscr{L}_{\tau}-L_{\tau}\|$ since $\mathscr{L}_{\tau}=VEV'$ with rank $K$ and $\mathscr{L}_{\tau}$ can also be viewed as a rank $K$ approximation to $L_{\tau}$. This leads to $\|\hat{V}\hat{E}\hat{V}'-VEV'\|=\|\hat{V}\hat{E}\hat{V}'-L_{\tau}+L_{\tau}-\mathscr{L}_{\tau}\|\leq 2\|L_{\tau}-\mathscr{L}_{\tau}\|\leq O(err_{n})$.
	Then, we have
	\begin{align*}
	&|J_{k}-\hat{J}_{p(k)}|= \frac{|J^{2}_{k}-\hat{J}^{2}_{p(k)}|}{J_{k}+\hat{J}_{p(k)}}\leq \frac{|J^{2}_{k}-\hat{J}^{2}_{p(k)}|}{J_{k}}\leq |J^{2}_{k}-\hat{J}^{2}_{p(k)}|\frac{\tilde{\theta}_{\mathrm{max}}}{\tilde{\theta}^{2}_{\mathrm{min}}\sqrt{\lambda_{K}(\Pi'\Pi)}}
	\\
	&\leq(O(\varpi_{2})+O(err_{n}))\frac{\tilde{\theta}_{\mathrm{max}}}{\tilde{\theta}^{2}_{\mathrm{min}}\sqrt{\lambda_{K}(\Pi'\Pi)}}=O(\frac{\tilde{\theta}^{8}_{\mathrm{max}}K^{2.5}\varpi_{1}\kappa^{3.5}(\Pi'\Pi)}{\eta \tilde{\theta}^{9}_{\mathrm{min}}}).
	\end{align*}
	Then, for $1\leq i\leq n$, since $Z=Y_{\bullet}J, \hat{Z}=\hat{Y}_{\bullet}\hat{J}$, we have
	\begin{align*}
	&\|e'_{i}(\hat{Z}-Z\mathcal{P})\|_{F}=\|e'_{i}(\hat{Y}_{\bullet}-Y_{\bullet}\mathcal{P})\hat{J}+e'_{i}Y_{\bullet}\mathcal{P}(\hat{J}-\mathcal{P}'J\mathcal{P})\|_{F}\\
	&\leq\|e'_{i}(\hat{Y}_{\bullet}-Y_{\bullet}\mathcal{P})\|_{F}\|\hat{J}\|_{F}+\|e'_{i}Y_{\bullet}\mathcal{P}\|_{F}\|\hat{J}-\mathcal{P}'J\mathcal{P}\|_{F}\\
	&\leq\|e'_{i}(\hat{Y}_{\bullet}-Y_{\bullet}\mathcal{P})\|_{F}\|\hat{J}-\mathcal{P}'J\mathcal{P}+\mathcal{P}'J\mathcal{P}\|_{F}+\|e'_{i}Y_{\bullet}\mathcal{P}\|_{F}\|\hat{J}-\mathcal{P}'J\mathcal{P}\|_{F}\\
	&\leq\|e'_{i}(\hat{Y}_{\bullet}-Y_{\bullet}\mathcal{P})\|_{F}(\|\hat{J}-\mathcal{P}'J\mathcal{P}\|_{F}+\|J\|_{F})+\|e'_{i}Y_{\bullet}\|_{F}\|\hat{J}-\mathcal{P}'J\mathcal{P}\|_{F}\\
	&=\|e'_{i}(\hat{Y}_{\bullet}-Y_{\bullet}\mathcal{P})\|_{F}(\|J-\mathcal{P}\hat{J}\mathcal{P}'\|_{F}+\|J\|_{F})+\|e'_{i}Y_{\bullet}\|_{F}\|J-\mathcal{P}\hat{J}\mathcal{P}'\|_{F}\\
	&\leq O(\frac{\tilde{\theta}^{10}_{\mathrm{max}}K^{3.5}\varpi_{1}\kappa^{4.5}(\Pi'\Pi)}{\eta \tilde{\theta}^{10}_{\mathrm{min}}})(O(\frac{\tilde{\theta}^{8}_{\mathrm{max}}K^{3}\varpi_{1}\kappa^{3.5}(\Pi'\Pi)}{\eta \tilde{\theta}^{9}_{\mathrm{min}}})+\frac{\tilde{\theta}^{2}_{\mathrm{max}}K\sqrt{\lambda_{1}(\Pi'\Pi)}}{\tilde{\theta}_{\mathrm{min}}})\\
	&~~~+\frac{\tilde{\theta}^{2}_{\mathrm{max}}\sqrt{K\kappa(\Pi'\Pi)}}{\tilde{\theta}^{2}_{\mathrm{min}}\sqrt{\lambda_{K}(\Pi'\Pi)}}O(\frac{\tilde{\theta}^{8}_{\mathrm{max}}K^{3}\varpi_{1}\kappa^{3.5}(\Pi'\Pi)}{\eta \tilde{\theta}^{9}_{\mathrm{min}}})\\
	&=O(\frac{\tilde{\theta}^{10}_{\mathrm{max}}K^{3.5}\kappa^{4}(\Pi'\Pi)\varpi_{1}}{\tilde{\theta}^{11}_{\mathrm{min}}\eta\sqrt{\lambda_{K}(\Pi'\Pi)}}).
	\end{align*}
\end{proof}
\subsection{Proof of Theorem 3.6}
\begin{proof}
	Since the difference between the row-normalized projection coefficients $\Pi$ and $\hat{\Pi}$ can be bounded by the difference between $Z$ and $\hat{Z}$, for $1\leq i\leq n$, we have
	\begin{align*}	\|e'_{i}(\hat{\Pi}-\Pi\mathcal{P})\|_{F}&=\|\frac{e'_{i}\hat{Z}}{\|e'_{i}\hat{Z}\|_{F}}-\frac{e'_{i}Z\mathcal{P}}{\|e'_{i}Z\mathcal{P}\|_{F}}\|_{F}=\|\frac{e'_{i}\hat{Z}\|e'_{i}Z\|_{F}-e'_{i}Z\mathcal{P}\|e'_{i}\hat{Z}\|_{F}}{\|e'_{i}\hat{Z}\|_{F}\|e'_{i}Z\|_{F}}\|_{F}\\	&=\|\frac{e'_{i}\hat{Z}\|e'_{i}Z\|_{F}-e'_{i}\hat{Z}\|e'_{i}\hat{Z}\|_{F}+e'_{i}\hat{Z}\|e'_{i}\hat{Z}\|_{F}-e'_{i}Z\mathcal{P}\|e'_{i}\hat{Z}\|_{F}}{\|e'_{i}\hat{Z}\|_{F}\|e'_{i}Z\|_{F}}\|_{F}\\
	&\leq \frac{\|e'_{i}\hat{Z}\|e'_{i}Z\|_{F}-e'_{i}\hat{Z}\|e'_{i}\hat{Z}\|_{F}\|_{F}+\|e'_{i}\hat{Z}\|e'_{i}\hat{Z}\|_{F}-e'_{i}Z\mathcal{P}\|e'_{i}\hat{Z}\|_{F}\|_{F}}{\|e'_{i}\hat{Z}\|_{F}\|e'_{i}Z\|_{F}}\\
	&=\frac{\|e'_{i}\hat{Z}\|_{F}|\|e'_{i}Z\|_{F}-\|e'_{i}\hat{Z}\|_{F}|+\|e'_{i}\hat{Z}\|_{F}\|e'_{i}\hat{Z}-e'_{i}Z\mathcal{P}\|_{F}}{\|e'_{i}\hat{Z}\|_{F}\|e'_{i}Z\|_{F}}\\
	&=\frac{|\|e'_{i}Z\|_{F}-\|e'_{i}\hat{Z}\|_{F}|+\|e'_{i}\hat{Z}-e'_{i}Z\mathcal{P}\|_{F}}{\|e'_{i}Z\|_{F}}\leq\frac{2\|e'_{i}(\hat{Z}-Z\mathcal{P})\|_{F}}{\|e'_{i}Z\|_{F}}\\
	&\leq \frac{2\|e'_{i}(\hat{Z}-Z\mathcal{P})\|_{F}}{\mathrm{min}_{1\leq j\leq n}\|e'_{j}Z\|_{F}}.
	\end{align*}
	Set $m_{Z}=\mathrm{min}_{1\leq i\leq n}\|e'_{i}Z\|_{F}$ for notation convenience.
	Next, we give a lower bound for $m_{Z}$. Since $Z=Y_{\bullet}J=N^{-1}N_{M_{1}}\Pi$, where $N_{M_{1}}$ is defined in the proof of Lemma 2.3, i.e.,
	$N_{M_{1}}=\begin{bmatrix}
	\frac{1}{\|M_{1}(1,:)\|_{F}} &  & & \\
	& \frac{1}{\|M_{1}(2,:)\|_{F}}& &\\
	& & \ddots&\\
	&&&\frac{1}{\|M_{1}(n,:)\|_{F}}
	\end{bmatrix},$
	where $M_{1}=\Pi\tilde{\Theta}^{-1}(\mathcal{I},\mathcal{I})V(\mathcal{I},:)$. Thus, for $1\leq i\leq n$, we have
	\begin{align*}
	&\|e'_{i}Z\|_{F}=\|N^{-1}(i,i)N_{M_{1}}(i,i)\Pi(i,:)\|_{F}=N^{-1}(i,i)N_{M_{1}}(i,i)\|e'_{i}\Pi\|_{F}\\
	&\geq \mathrm{min}_{1\leq j\leq n}N^{-1}(j,j)\mathrm{min}_{1\leq j\leq n}N_{M_{1}}(j,j)\mathrm{min}_{1\leq j\leq n}\|e'_{j}\Pi\|_{F}\\
	&\geq\frac{1}{K^{0.5}}\mathrm{min}_{1\leq j\leq n}N^{-1}(j,j)\mathrm{min}_{1\leq j\leq n}N_{M_{1}}(j,j)\\
	&=\frac{K^{-0.5}}{\mathrm{max}_{1\leq j\leq n}N(j,j)\mathrm{max}_{1\leq j\leq n}\|e'_{j}M_{1}\|_{F}},
	\end{align*}
	where we use the fact that $\mathrm{min}_{i}\|e'_{i}\Pi\|_{F}\geq 1/K^{0.5}$. Since for any $1\leq i\leq n$, we have
	\begin{align*}
	\|e'_{i}M_{1}\|_{F}&=\|e'_{i}\Pi\tilde{\Theta}^{-1}(\mathcal{I},\mathcal{I})V(\mathcal{I},:)\|_{F}\leq \|e'_{i}\Pi\|_{F}\|\tilde{\Theta}^{-1}(\mathcal{I},\mathcal{I})V(\mathcal{I},:)\|_{F}\\
	&\leq \|\tilde{\Theta}^{-1}(\mathcal{I},\mathcal{I})V(\mathcal{I},:)\|_{F}\leq \|\tilde{\Theta}^{-1}(\mathcal{I},\mathcal{I})\|_{F}\|V(\mathcal{I},:)\|_{F}\leq \frac{\sqrt{K}}{\tilde{\theta}_{\mathrm{min}}}\|V(\mathcal{I},:)\|_{F}\leq \frac{K}{\tilde{\theta}_{\mathrm{min}}}\|V\|_{2\rightarrow\infty}\\
	&\overset{\mathrm{By~Lemma~}\ref{P2}}{\leq}\frac{\tilde{\theta}_{\mathrm{max}} K}{\tilde{\theta}^{2}_{\mathrm{min}}\sqrt{\lambda_{K}(\Pi'\Pi)}}.
	\end{align*}
	By the proof of Lemma 3.5, we have $N(i,i)\leq\frac{\tilde{\theta}_{\mathrm{max}}\sqrt{K\lambda_{1}(\Pi'\Pi)}}{\tilde{\theta}_{\mathrm{min}}}$.
	Combine the uppers bound of $\|e'_{i}M_{1}\|_{F}$  and $N(i,i)$, we have $m_{Z}\geq \frac{\tilde{\theta}^{3}_{\mathrm{min}}}{\tilde{\theta}^{2}_{\mathrm{max}}K^{2}\sqrt{\kappa(\Pi'\Pi)}}$, which gives that
	\begin{align*}
	&\|e'_{i}(\hat{\Pi}-\Pi\mathcal{P})\|_{F}\leq \frac{2\|e'_{i}(\hat{Z}-Z\mathcal{P})\|_{F}}{m_{Z}}\leq\frac{2\|e'_{i}(\hat{Z}-Z\mathcal{P})\|_{F}\tilde{\theta}^{2}_{\mathrm{max}}K^{2}\sqrt{\kappa(\Pi'\Pi)}}{\tilde{\theta}^{3}_{\mathrm{min}}}\\
	&\overset{\mathrm{By~Lemma~}3.5}{\leq}O(\frac{\tilde{\theta}^{12}_{\mathrm{max}}K^{5.5}\varpi_{1}\kappa^{4.5}(\Pi'\Pi)}{\eta \tilde{\theta}^{14}_{\mathrm{min}}\sqrt{\lambda_{K}(\Pi'\Pi)}}).
	\end{align*}
	Now, we give a lower bound for $\eta$. By the proof of Lemma \ref{P2}, we have $(V(\mathcal{I},:)V'(\mathcal{I},:))^{-1}=\tilde{\Theta}^{-1}(\mathcal{I},\mathcal{I})\Pi'\tilde{\Theta}^{2}\Pi\tilde{\Theta}^{-1}(\mathcal{I},\mathcal{I})$, which gives that
	\begin{align*}
	&(V_{*}(\mathcal{I},:)V'_{*}(\mathcal{I},:))^{-1}=(N(\mathcal{I},\mathcal{I})V(\mathcal{I},:)V'(\mathcal{I},:)N(\mathcal{I},\mathcal{I}))^{-1}\\
	&=N^{-1}(\mathcal{I},\mathcal{I})\tilde{\Theta}^{-1}(\mathcal{I},\mathcal{I})\Pi'\tilde{\Theta}^{2}\Pi\tilde{\Theta}^{-1}(\mathcal{I},\mathcal{I})N^{-1}(\mathcal{I},\mathcal{I})\\
	&\geq \frac{\tilde{\theta}^{2}_{\mathrm{min}}}{\tilde{\theta}^{2}_{\mathrm{max}}N^{2}_{\mathrm{max}}}\Pi'\Pi,
	\end{align*}
	where we set $N_{\mathrm{max}}=\mathrm{max}_{1\leq i\leq n}N(i,i)$. By the proof of Lemma 3.5, we have $N_{\mathrm{max}}\leq\frac{\tilde{\theta}_{\mathrm{max}}\sqrt{K\lambda_{1}(\Pi'\Pi)}}{\tilde{\theta}_{\mathrm{min}}}$, which gives that
	\begin{align*}
	(V_{*}(\mathcal{I},:)V'_{*}(\mathcal{I},:))^{-1}\geq \frac{\tilde{\theta}^{4}_{\mathrm{min}}}{\tilde{\theta}^{4}_{\mathrm{max}}K\lambda_{1}(\Pi'\Pi)}\Pi'\Pi.
	\end{align*}
	Since $\mathrm{min}_{1\leq k\leq K}e'_{k}\Pi'\Pi\mathbf{1}=\pi_{\mathrm{min}}$, we have
	$\eta=\mathrm{min}_{1\leq k\leq K}((V_{*}(\mathcal{I},:)V'_{*}(\mathcal{I},:))^{-1}\mathbf{1})(k)\geq\frac{\tilde{\theta}^{4}_{\mathrm{min}}\pi_{\mathrm{min}}}{\tilde{\theta}^{4}_{\mathrm{max}}K\lambda_{1}(\Pi'\Pi)}$, which gives that
	\begin{align*}
	\mathrm{max}_{1\leq i\leq n}\|e'_{i}(\hat{\Pi}-\Pi\mathcal{P})\|_{F}\leq O(\frac{\tilde{\theta}^{12}_{\mathrm{max}}K^{5.5}\varpi_{1}\kappa^{4.5}(\Pi'\Pi)}{\eta \tilde{\theta}^{14}_{\mathrm{min}}\sqrt{\lambda_{K}(\Pi'\Pi)}})\leq O(\frac{\tilde{\theta}^{16}_{\mathrm{max}}K^{6.5}\varpi_{1}\kappa^{4.5}(\Pi'\Pi)\lambda_{1}(\Pi'\Pi)}{\tilde{\theta}^{18}_{\mathrm{min}}\pi_{\mathrm{min}}\sqrt{\lambda_{K}(\Pi'\Pi)}}).
	\end{align*}
	Since $\tilde{\theta}_{\mathrm{max}}=\mathrm{max}_{i}\frac{\theta(i)}{\sqrt{\tau+\mathscr{D}(i,i)}}\leq\frac{\theta_{\mathrm{max}}}{\sqrt{\tau+\delta_{\mathrm{min}}}}$ and $\tilde{\theta}_{\mathrm{min}}=\mathrm{min}_{i}\frac{\theta(i)}{\sqrt{\tau+\mathscr{D}(i,i)}}\geq\frac{\theta_{\mathrm{min}}}{\sqrt{\tau+\delta_{\mathrm{max}}}}$, by Lemma 3.1, we have
	\begin{align*}
	\mathrm{max}_{1\leq i\leq n}\|e'_{i}(\hat{\Pi}-\Pi\mathcal{P})\|_{F}\leq O(\frac{\theta^{16}_{\mathrm{max}}(\tau+\delta_{\mathrm{max}})^{9}K^{6.5}\varpi_{1}\kappa^{4.5}(\Pi'\Pi)\lambda_{1}(\Pi'\Pi)}{\theta^{18}_{\mathrm{min}}(\tau+\delta_{\mathrm{min}})^{8}\pi_{\mathrm{min}}\sqrt{\lambda_{K}(\Pi'\Pi)}}).
	\end{align*}
	By Lemma 3.2, since $\varpi_{1}=O(\frac{(\tau+\delta_{\mathrm{max}})\theta_{\mathrm{max}}\sqrt{K\mathrm{log}(n)}}{(\tau+\delta_{\mathrm{min}})\theta^{2}_{\mathrm{min}}|\lambda_{K}(P)|\lambda_{K}(\Pi'\Pi)})$,
	we have
	\begin{align*}
	\mathrm{max}_{1\leq i\leq n}\|e'_{i}(\hat{\Pi}-\Pi\mathcal{P})\|_{F}\leq O(\frac{\theta^{17}_{\mathrm{max}}(\tau+\delta_{\mathrm{max}})^{10}K^{7}\kappa^{4.5}(\Pi'\Pi)\lambda_{1}(\Pi'\Pi)\sqrt{\mathrm{log}(n)}}{\theta^{20}_{\mathrm{min}}(\tau+\delta_{\mathrm{min}})^{9}|\lambda_{K}(P)|\pi_{\mathrm{min}}\lambda^{1.5}_{K}(\Pi'\Pi)}).
	\end{align*}
	Since $\delta_{\mathrm{max}}=\mathrm{max}_{i}\mathscr{D}(i,i)=\mathrm{max}_{i}\sum_{j=1}^{n}\Omega(i,j)=\mathrm{max}_{i}\theta(j)\sum_{j=1}^{n}\theta(j)P(g_{i},g_{j})\leq \theta_{\mathrm{max}}\|\theta\|_{1}$, combining it with the fact that $\tau+\delta_{\mathrm{min}}\leq C\theta_{\mathrm{max}}\|\theta\|_{1}$ by Lemma 3.1, we have $\tau+\delta_{\mathrm{max}}\leq C\theta_{\mathrm{max}}\|\theta\|_{1}$, which gives that
	\begin{align*}
	\mathrm{max}_{1\leq i\leq n}\|e'_{i}(\hat{\Pi}-\Pi\mathcal{P})\|_{F}\leq O(\frac{\theta^{17}_{\mathrm{max}}(\tau+\delta_{\mathrm{max}})^{9}K^{7}\kappa^{4.5}(\Pi'\Pi)\lambda_{1}(\Pi'\Pi)\theta_{\mathrm{max}}\|\theta\|_{1}\sqrt{\mathrm{log}(n)}}{\theta^{20}_{\mathrm{min}}(\tau+\delta_{\mathrm{min}})^{9}|\lambda_{K}(P)|\pi_{\mathrm{min}}\lambda^{1.5}_{K}(\Pi'\Pi)}).
	\end{align*}
\end{proof}
\subsection{Proofs of Corollaries 3.7 and 3.8}
\begin{proof}
	For Corollary 3.7, since $\tau_{\mathrm{opt}}=O(\theta_{\mathrm{max}}\|\theta\|_{1})$ and $\delta_{\mathrm{max}}\leq \theta_{\mathrm{max}}\|\theta\|_{1}$, we have $(\frac{\tau+\delta_{\mathrm{max}}}{\tau+\delta_{\mathrm{min}}})^{9}=O(1)$,
	hence Corollary 3.7's first result follows. For the sparest case, simply use $\mathrm{log}^{1+2\gamma}(n)$ to replace $\theta_{\mathrm{max}}\|\theta\|_{1}$, and then we can obtain the result.
	
	For Corollary 3.8, when $\theta_{\mathrm{max}}\leq C\theta_{\mathrm{min}}$, we have $\theta_{\mathrm{min}}\sqrt{n}=\sqrt{\theta^{2}_{\mathrm{min}}n}=O(\sqrt{\theta^{2}_{\mathrm{max}}n})=O(\sqrt{\theta_{\mathrm{max}}\|\theta\|_{1}})$. Now, simply substitute $K=O(1), \pi_{\mathrm{min}}=O(n/K)=O(n), \lambda_{1}(\Pi'\Pi)=O(n/K)=O(n), \kappa(\Pi'\Pi)=O(1)$ into Corollary 3.7 and Corollary 3.8 follows.
\end{proof}

\begin{rem}
	In Corollary 3.8, for consistency estimation (i.e., $\mathrm{max}_{1\leq i\leq n}\|e'_{i}(\hat{\Pi}-\Pi\mathcal{P})\|_{F}\leq 1$), we need $|\lambda_{K}(P)|\geq O(\sqrt{\frac{\mathrm{log}(n)}{\theta_{\mathrm{max}}\|\theta\|_{1}}})$.
	
	Recall the condition $|\lambda_{K}|\geq C\frac{\theta_{\mathrm{max}}\sqrt{n\mathrm{log}(n)}}{\tau+\delta_{\mathrm{min}}}$ in Lemma 3.2. Set $\tau$ as $\tau_{\mathrm{opt}}$ in Eq (9), this condition reads $|\lambda_{K}|\geq C\frac{\sqrt{n\mathrm{log}(n)}}{\|\theta\|_{1}}$. By Lemma \ref{P4}, we know that $|\lambda_{K}|\geq \tilde{\theta}^{2}_{\mathrm{min}}|\lambda_{K}(P)|\lambda_{K}(\Pi'\Pi)\geq \frac{\theta^{2}_{\mathrm{min}}|\lambda_{K}(P)|\lambda_{K}(\Pi'\Pi)}{\tau+\delta_{\mathrm{max}}}=O(\frac{\theta^{2}_{\mathrm{min}}|\lambda_{K}(P)|\lambda_{K}(\Pi'\Pi)}{\theta_{\mathrm{max}}\|\theta\|_{1}})$. Under the settings of Corollary 3.8, to make the condition $|\lambda_{K}|\geq C\frac{\sqrt{n\mathrm{log}(n)}}{\|\theta\|_{1}}$ always holds, by Lemma \ref{P4}, we only need $\frac{\theta^{2}_{\mathrm{min}}|\lambda_{K}(P)|\lambda_{K}(\Pi'\Pi)}{\theta_{\mathrm{max}}\|\theta\|_{1}}\geq C\frac{\sqrt{n\mathrm{log}(n)}}{\|\theta\|_{1}}\Leftrightarrow \theta_{\mathrm{max}}|\lambda_{K}(P)|\sqrt{n}\geq C\sqrt{\mathrm{log}(n)}\Leftrightarrow |\lambda_{K}(P)|\sqrt{\theta_{\mathrm{max}}\|\theta\|_{1}}\geq C\sqrt{\mathrm{log}(n)}\Leftrightarrow|\lambda_{K}(P)|\geq O(\sqrt{\frac{\mathrm{log}(n)}{\theta_{\mathrm{max}}\|\theta\|_{1}}})$, which is consistent with the consistency estimation requirement on $|\lambda_{K}(P)|$. Therefore, under the settings of Corollary 3.8, for consistency estimation when the lower bound requirement on $|\lambda_{K}(P)|$ in Lemma 3.2 holds, $\lambda_{K}(P)$ should satisfy
	$$|\lambda_{K}(P)|\geq O(\sqrt{\frac{\mathrm{log}(n)}{\theta_{\mathrm{max}}\|\theta\|_{1}}}).$$
\end{rem}

\section{One-Class SVM and SVM-cone algorithm}\label{OneClassSVMandSVMcone}
In this section, we briefly introduce one-class SVM and SVM-cone algorithm given in \cite{MaoSVM}.

As mentioned in Problem 1 in \cite{MaoSVM}, if a matrix $S\in\mathbb{R}^{n\times m}$ has the form $S=HS_{C}$, where $H\in\mathrm{R}^{n\times K}$ with nonnegative entries, no row of $H$ is 0, and $S_{C}\in\mathbb{R}^{K\times m}$ corresponding to $K$ rows of $S$ (i.e., there exists an index set $\mathcal{I}$ with $K$ entries such that $S_{C}=S(\mathcal{I},:)$), and each row of $S$ has unit $l_{2}$ norm. Then problem of inferring $H$ from $S$ is called the ideal cone problem. The ideal cone problem can be solved by one-class SVM applied to the rows of $S$. the $K$ normalized corners in $S_{C}$ are the support vectors found by a one-class SVM:
\begin{align}\label{OneClassSVM}
\mathrm{maximize~}b~~\mathrm{s.t.}~~\textbf{w}'S(i,:)\geq b(\mathrm{~for~}i=1,2,\ldots,n)~\mathrm{and~~}\|\textbf{w}\|_{F}\leq 1.
\end{align}
The solution  $(\textbf{w}, b)$ for the ideal cone problem when $(S_{C}S'_{C})^{-1}\mathbf{1}>0$ is given by
\begin{align}\label{SolutionOfOneClassSVM}
\textbf{w}=b^{-1}\cdot S'_{C}\frac{(S_{C}S'_{C})^{-1}\mathbf{1}}{\mathbf{1}'(S_{C}S'_{C})^{-1}\mathbf{1}},~~~ b=\frac{1}{\sqrt{\mathbf{1}'(S_{C}S'_{C})^{-1}\mathbf{1}}}.
\end{align}
for the empirical case, if we are given a matrix  $\hat{S}\in\mathbb{R}^{n\times m}$ such that all rows of $\hat{S}$  have unit $l_{2}$ norm, infer $H$ from $\hat{S}$ with given $K$ is called the empirical cone problem (i.e., Problem 2 in \cite{MaoSVM}). For the empirical cone problem, we can apply one-class SVM to all rows of $\hat{S}$ to obtain $\textbf{w}$ and $b$'s estimations $\hat{\textbf{w}}$ and $\hat{b}$. Then apply K-means algorithm to rows of $\hat{S}$ that are close to the hyperplane into $K$ clusters, the $K$ clusters can give the estimation of the  index set $\mathcal{I}$. Below is the SVM-cone algorithm given in \cite{MaoSVM}.
\begin{algorithm}
	\caption{\textbf{SVM-cone}}
	\label{alg:SVMcone}
	\begin{algorithmic}[1]
		\Require $\hat{S}\in \mathbb{R}^{n\times m}$ with rows have unit $l_{2}$ norm, number of corners $K$, estimated distance corners from hyperplane $\gamma$.
		\Ensure The near-corner index set $\mathcal{\hat{I}}$.
		\State Run one-class SVM on $\hat{S}(i,:)$ to get $\hat{\textbf{w}}$ and $\hat{b}$
		\State Run K-means algorithm to the set $\{\hat{S}(i,:)| \hat{S}(i,:)\hat{\textbf{w}}\leq \hat{b}+\gamma\}$ that are close to the hyperplane into $K$ clusters
		\State Pick one point from each cluster to get the near-corner set $\mathcal{\hat{I}}$
	\end{algorithmic}
\end{algorithm}
As suggested in \cite{MaoSVM}, we can start $\gamma=0$ and incrementally increase it until $K$ distinct clusters are found.

Now turn to our Mixed-RSC algorithm. Set $\textbf{w}_{1}=b_{1}^{-1}V'_{*,1}(\mathcal{I},:)\frac{(V_{*,1}(\mathcal{I},:)V'_{*,1}(\mathcal{I},:))^{-1}\mathbf{1}}{\mathbf{1}'(V_{*,1}(\mathcal{I},:)V'_{*,1}(\mathcal{I},:))^{-1}}, b_{1}=\frac{1}{\sqrt{\mathbf{1}'(V_{*,1}(\mathcal{I},:)V'_{*,1}(\mathcal{I},:))^{-1}\mathbf{1}}}$, and $\textbf{w}_{2}=b_{2}^{-1}V'_{*,2}(\mathcal{I},:)\frac{(V_{*,2}(\mathcal{I},:)V'_{*,2}(\mathcal{I},:))^{-1}\mathbf{1}}{\mathbf{1}'(V_{*,2}(\mathcal{I},:)V'_{*,2}(\mathcal{I},:))^{-1}}, b_{2}=\frac{1}{\sqrt{\mathbf{1}'(V_{*,2}(\mathcal{I},:)V'_{*,2}(\mathcal{I},:))^{-1}\mathbf{1}}}$ such that $\textbf{w}_{1}$ and $b_{1}$ are solutions of the one-class SVM in Eq (\ref{OneClassSVM}) by setting $S=V_{*,1}$, and $\textbf{w}_{2}$ and $b_{2}$ are solutions of the one-class SVM in Eq (\ref{OneClassSVM}) by setting $S=V_{*,2}$ . By Lemma \ref{WhyUseKmeansInSVMcone}, we see that if node $i$ is a pure node, then we have $V_{*,1}(i,:)\textbf{w}_{1}=b_{1}$, which suggests that in the SVM-cone algorithm, if the input matrix is $V_{*,1}$, by setting $\gamma=0$, we can find all pure nodes, i.e., the set $\{V_{*,1}(i,:)|V_{*}(i,:)\textbf{w}_{1}=b_{1}\}$ contain all rows of $V_{*,1}$ respective to pure nodes while including mixed nodes. By Lemma 2.3, we see that these  pure nodes belong to $K$ distinct clusters such that if nodes $i,j$ are in the same clusters, then we have $V_{*,1}(i,:)=V_{*,1}(j,:)$, and this is the reason that we need to apply K-means algorithm on the set obtained in step 2 in the SVM-cone algorithm to obtain the $K$ distinct clusters, and this is also the reason that we said SVM-cone returns the index set $\mathcal{I}$ (the $K$ indexes of $\mathcal{I}$ denote the indexes of $K$ pure rows of $V_{*,1}$, one from each cluster) when the input is $V_{*,1}$ in the explanation of Figure 1. Similar arguments hold when the input is $V_{*,2}$ in the SVM-cone algorithm.
\begin{lem}\label{WhyUseKmeansInSVMcone}
	Under $DCMM(n,P,\Theta,\Pi)$, for $1\leq i\leq n$, $V_{*,1}(i,:)$, if node $i$ is a pure node such that $\Pi(i,k)=1$ for certain $k$, we have
	\begin{align*}
	V_{*,1}(i,:)\textbf{w}_{1}=b_{1}\mathrm{~~~and~~~}V_{*,2}(i,:)\textbf{w}_{2}=b_{2},
	\end{align*}
	Meanwhile, if node $i$ is not a pure node, then the above equalities do not hold.
\end{lem}
\begin{proof}
	We only prove that $V_{*,1}(i,:)\textbf{w}_{1}=b_{1}$ when $\Pi(i,k)=1$, since the second equality can be proved similarly. By Lemma \ref{P1}, we know that when node $i$ is a pure node such that $\Pi(i,k)=1$, $V_{*,1}(i,:)$ can be written as $V_{*,1}(i,:)=e'_{k}V_{*,1}(\mathcal{I},:)$, then we have $V_{*,1}(i,:)\textbf{w}_{1}=b_{1}$ surely. And if $i$ is a mixed node, by Lemma \ref{P1}, we know that $r(i)>1$ and $\Phi(i,:)\neq e_{k}$ for any $k=1,2,\ldots, K$, hence $V_{*,1}(i,:)\neq e'_{k}V_{*,1}(\mathcal{I},:)$ if $i$ is mixed, which gives the result.
\end{proof}

\end{document}